%% file: main.tex
\documentclass[a4paper,UKenglish,cleveref,autoref,thm-restate,utf8,utf8x]{lipics-v2021}

\author{Kyle Deeds}{University of Washington Seattle, WA, USA}{kdeeds@cs.washington.edu}{[]}{}
\author{Dan Suciu}{University of Washington Seattle, WA, USA}{suciu@cs.washington.edu}{[]}{}
\author{Magda Balazinska}{University of Washington Seattle, WA, USA}{magda@cs.washington.edu}{[]}{}
\author{Walter Cai}{University of Washington Seattle, WA, USA}{wzcai92@gmail.com}{[]}{}

\authorrunning{K. Deeds, D. Suciu, M. Balazinska, W. Cai}

\title{Degree Sequence Bound For Join Cardinality Estimation} 

\Copyright{CC-BY} 

\ccsdesc[500]{Information systems~Query optimization}
\ccsdesc[500]{Information systems~Query planning}
\ccsdesc[500]{Theory of computation~Database query processing and optimization (theory)}
\ccsdesc[300]{Theory of computation~Data modeling} 

\keywords{Cardinality Estimation, Cardinality Bounding, Degree Bounds, Functional Approximation, Query Planning, Berge-Acyclic Queries} 

\category{} 

\relatedversion{} 

\funding{This work is supported by National Science Foundation grants NSF IIS 1907997 and NSF-BSF 2109922. }

\nolinenumbers 

\EventEditors{John Q. Open and Joan R. Access}
\EventNoEds{2}
\EventLongTitle{42nd Conference on Very Important Topics (CVIT 2016)}
\EventShortTitle{CVIT 2016}
\EventAcronym{CVIT}
\EventYear{2016}
\EventDate{December 24--27, 2016}
\EventLocation{Little Whinging, United Kingdom}
\EventLogo{}
\SeriesVolume{42}
\ArticleNo{23}

\input{macros}

\begin{document}

\maketitle

\input{sec-0-abstract}

\input{sec-1-intro}

\input{sec-2-problem}
\input{sec-3-star}

\input{sec-4-general}
\input{sec-5-compress}

\input{sec-6-discussion}

\bibliography{bib}
\appendix

\input{appendix}

\end{document}

%% file: macros.tex
\usepackage{amsmath}
\usepackage{algorithm}
\usepackage{algorithmic}
\usepackage{bm} 
\usepackage{amsfonts}
\usepackage{amsthm}
\usepackage{graphicx}
\usepackage{tikz}
\usepackage[colorinlistoftodos]{todonotes}

\usepackage{enumitem}
\usepackage{wrapfig}
\usepackage{nicematrix}

\usepackage[small,nohug,heads=vee]{diagrams}
\diagramstyle[labelstyle=\scriptstyle]

\DeclareMathOperator*{\argmin}{arg\,min}

\newcommand{\set}[1]{\{#1\}}                    
\newcommand{\setof}[2]{\{{#1}\mid{#2}\}}        
\newcommand{\R}{{\mathbb R}}    			

\newcommand{\calM}{\mathcal M}

\theoremstyle{plain}                   
\newtheorem{thm}{Theorem}[section]
\newtheorem{lmm}[thm]{Lemma}
\newtheorem{prop}[thm]{Proposition}

\newtheorem{cor}[thm]{Corollary}

\theoremstyle{definition}              
\newtheorem{pbm}{Problem}
\newtheorem{opm}{Question}
\newtheorem{conj}[thm]{Conjecture}
\newtheorem{ex}[thm]{Example}

\newtheorem{defn}[thm]{Definition}

\newtheorem{rmk}[thm]{Remark}

\newcommand{\bi}{\begin{itemize}}
\newcommand{\ei}{\end{itemize}}

\newcommand{\bp}{\begin{proof}}
\newcommand{\ep}{\end{proof}}
\newcommand{\bcor}{\begin{cor}}
\newcommand{\ecor}{\end{cor}}
\newcommand{\bthm}{\begin{thm}}
\newcommand{\ethm}{\end{thm}}
\newcommand{\blmm}{\begin{lmm}}
\newcommand{\elmm}{\end{lmm}}
\newcommand{\bdefn}{\begin{defn}}
\newcommand{\edefn}{\end{defn}}
\newcommand{\bprop}{\begin{prop}}
\newcommand{\eprop}{\end{prop}}
\newcommand{\bconj}{\begin{conj}}
\newcommand{\econj}{\end{conj}}
\newcommand{\bopm}{\begin{opm}}
\newcommand{\eopm}{\end{opm}}
\newcommand{\brmk}{\begin{rmk}}
\newcommand{\ermk}{\end{rmk}}

\newcommand{\defeq}{\stackrel{\text{def}}{=}}

\newcommand{\floor}[1]{\left\lfloor #1 \right\rfloor}
\newcommand{\ceil}[1]{\left\lceil #1 \right\rceil}

\newcommand{\eat}[1]{}

%% file: sec-0-abstract.tex
\begin{abstract}
    Recent work has demonstrated the catastrophic effects of poor cardinality estimates on query processing time. In particular, underestimating query cardinality can result in overly optimistic query plans which take orders of magnitude longer to complete than one generated with the true cardinality. Cardinality bounding avoids this pitfall by computing a strict upper bound on the query's output size using statistics about the database such as table sizes and degrees, i.e. value frequencies. In this paper, we extend this line of work by proving a novel bound called the Degree Sequence Bound which takes into account the full degree sequences and the max tuple multiplicity. This work focuses on the important class of Berge-Acyclic queries for which the Degree Sequence Bound is tight. Further, we describe how to practically compute this bound using a functional approximation of the true degree sequences and prove that even this functional form improves upon previous bounds.
\end{abstract}

%% file: sec-1-intro.tex
\section{Introduction}

\label{sec:intro}

The weakest link in a modern query processing engine is the {\em cardinality estimator}.  There are several major decisions where the system needs to estimate the size of a query's output: the optimizer uses the estimate to compute an effective query plan; the scheduler needs the estimate to determine how much memory to allocate for a hash table and to decide whether to use a main-memory or an out-of-core algorithm; a distributed system needs the estimate to decide how many servers to reserve for subsequent operations. Today's systems estimate the cardinality of a query by making several strong and unrealistic assumptions, such as uniformity and independence.  As a result, the estimates for multi-join queries commonly have relative errors up to several orders of magnitude. An aggravating phenomenon is that cardinality estimators consistently underestimate (this is a consequence of the independence assumption), and this leads to wrong decisions for the most expensive queries~\cite{DBLP:journals/pvldb/LeisGMBK015,DBLP:conf/sigmod/CaiBS19,han2021cardinality}. A significant amount of effort has been invested in the last few years into using machine learning for cardinality estimation~\cite{DBLP:conf/sigmod/Sun0021,DBLP:conf/sigmod/WuC21,DBLP:journals/pvldb/YangKLLDCS20,DBLP:journals/pvldb/ZhuWHZPQZC21,DBLP:journals/pvldb/WangQWWZ21,DBLP:journals/pvldb/LiuD0Z21,DBLP:journals/pvldb/NegiMKMTKA21}, but this approach still faces several formidable challenges, such as the need for large training sets, the long training time of complex models, and the lack of guarantees about the resulting estimates.

\begin{wrapfigure}{l}{0.5\textwidth}
{\tiny
  \begin{tabular}[b]{{|l|l|}} \hline
    {\bf Name} & ... \\ \hline
    Alice & ... \\
    Alice  & ... \\\hline
    Bob  & ... \\ \hline
    Carlos  & ...  \\
    Carlos  & ... \\
    Carlos  & ... \\ \hline
    David  & ... \\ \hline
    Eseah  & ... \\
    Eseah  & ... \\
    Eseah  & ... \\
    Eseah  & ... \\
    Eseah  & ... \\ \hline
    Vivek  & ... \\
    Vivek  & ... \\
    Vivek  & ... \\ \hline
    Gael  & ... \\ \hline
    Hans  & ... \\
    Hans  & ... \\ \hline
    John  & ... \\ \hline
    Karl  & ... \\ \hline
    Lee  & ... \\ \hline
  \end{tabular}
}
    \includegraphics[width=140pt, height=160pt]{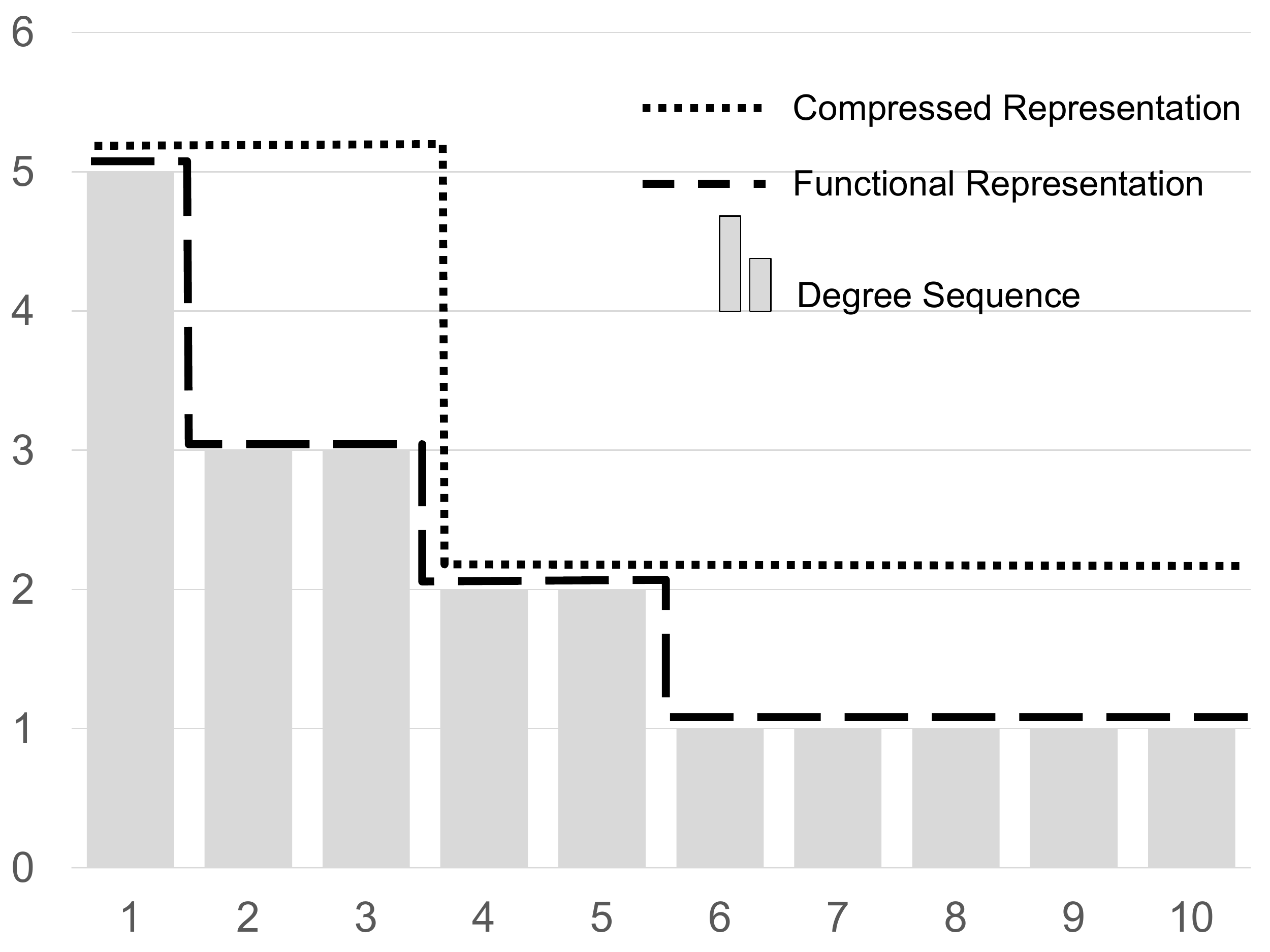}
    \caption{The degree sequence of \texttt{Name}.  The first rank
      represents \texttt{Eseah} whose degree is 5, the next two ranks
      are for \texttt{Carlos} and \texttt{Vivek} whose degrees are 3.
      The degree sequence can be represented compactly using a
      staircase functions, and even more compactly using lossy
      compression.}
  \label{fig:degree-sequence-example}
\end{wrapfigure}

An alternative approach to estimating the cardinality is to compute an
{\em upper bound} for the size of the query answer.  This approach
originated in the database theory community, through the pioneering
results by Grohe and Marx~\cite{DBLP:conf/soda/GroheM06} and Atserias,
Grohe, and Marx~\cite{DBLP:conf/focs/AtseriasGM08}.  They described an
elegant formula, now called the AGM bound, that gives a tight upper
bound on the query result in terms of the cardinalities of the input
tables.  This upper bound was improved by the {\em polymatroid bound},
which takes into account both the cardinalities, and the degree
constraints and includes functional dependencies as a special
case~\cite{DBLP:journals/jacm/GottlobLVV12,DBLP:conf/pods/KhamisNS16,DBLP:conf/pods/Khamis0S17,DBLP:conf/pods/000118}.
In principle, an upper bound could be used by a query optimizer in
lieu of a cardinality estimator and, indeed, this idea was recently
pursued by the systems community, where the upper bound appears under
various names such as bound sketch or pessimistic cardinality
estimator~\cite{DBLP:conf/sigmod/CaiBS19,DBLP:conf/cidr/HertzschuchHHL21}.
In this paper, we will call it a {\em cardinality bound}.  As
expected, a cardinality bound prevents query optimizers from choosing
disastrous plans for the most expensive
queries~\cite{DBLP:conf/sigmod/CaiBS19}, however, their relative error
is often much larger than that of other
methods~\cite{DBLP:conf/sigmod/ParkKBKHH20,DBLP:conf/sigmod/GiladPM21}.
While the appeal of a guaranteed upper bound is undeniable, in
practice overly pessimistic bounds are unacceptable.

In this paper, we propose a new upper bound on the query size based on
{\em degree sequences}. By using a slightly larger memory footprint,
this method has the potential to achieve much higher accuracy than
previous bounds. Given a relation $R$, an attribute $X$, and a value
$u \in \Pi_X(R)$, the {\em degree} of $u$ is the number of tuples in
$R$ with $u$ in the $X$ attribute, formally
$d^{(u)}=|\sigma_{X=u}(R)|$. The {\em degree sequence} of an attribute
$X$ in relation $R$ is the sorted sequence of all degrees for the
values of that attribute,
$d^{(u_1)} \geq d^{(u_2)} \geq \cdots \geq d^{(u_{n})}$. Going
forward, we drop any reference to values and instead refer to degrees
by their index in this sequence, also called their {\em rank},
i.e. $d_1\geq \cdots \geq d_{n}$.\footnote{Note that the degree
  sequence is very similar to a rank-frequency distribution in the
  probability literature and has been extensively used in graph
  analysis ~\cite{bauer2015best,hakimi1978graphs}.} A degree sequence
can easily be computed offline, and can be compressed effectively,
with a good space/accuracy tradeoff due to its monotonicity; see
Fig.~\ref{fig:degree-sequence-example} for an illustration. Degree
sequences offer more information on the database instance than the
statistics used by previous upper bounds.  For example, the AGM bound
uses only the cardinality of the relations, which is $\sum_i d_i$,
while the extension to degree
constraints~\cite{DBLP:conf/pods/Khamis0S17} uses the cardinality,
$\sum_i d_i$, and the maximum degree, $d_1$.


For this new bound we had to develop entirely new techniques over
those used for the AGM and the polymatroid bounds.  Previous
techniques are based on information theory.  If some relation $R(X,Y)$
has cardinality $N$, then any probability space over $R$ has an
entropy that satisfies $H(XY) \leq \log N$; if the degree sequence of
the attribute $X$ is $d_1 \geq d_2 \geq \ldots$, then
$H(Y|X) \leq \log d_1$.  Both the AGM and the polymatroid bound start
from such constraints on the entropy. Unfortunately, these constraints
do not extend to degree sequences, because $H$ is ignorant of
$d_2, d_3, \ldots$ Information theory gives us only three degrees of
freedom, namely $H(XY), H(X), H(Y)$, while the degree sequence has an
arbitrary number of degrees of freedom.  Rather than using information
theory, our new framework models relations as tensors, and formulates
the upper bound as a linear optimization problem.  This framework is
restricted to {\em Berge-acyclic, fully conjunctive
  queries}~\cite{DBLP:journals/jacm/Fagin83} (reviewed in
Sec.~\ref{sec:problem}); throughout the paper we will assume that
queries are in this class.  As we explain in
appendix~\ref{app:background} these are the most common queries found
in applications.

{\bf The Worst-Case Instance} Our main result
(Theorems~\ref{th:main:star} and \ref{th:general}) is a tight
cardinality bound given the degree sequences of all relations. This
bound is obtained by evaluating the query on a {\em worst-case
  instance} that satisfies those degree constraints.\footnote{In graph
  theory, the problem of computing a graph satisfying a given degree
  sequence is called the {\em realization problem}.}  Intuitively,
each relation of the worst-case instance is obtained by matching the
highest degree values in the different columns, and the same principle
is applied across relations. For example, consider the join
$R(X,\ldots) \Join S(X,\ldots)$, where the degree sequences of $R.X$
and $S.X$ are $a_1 \geq a_2 \geq \cdots$ and
$b_1 \geq b_2 \geq \cdots$ respectively. The true cardinality of the
join is $\sum_i a_i b_{\tau(i)}$ for some unknown permutation $\tau$ ,
while the maximum cardinality is\footnote{For example, if
  $a_1 \geq a_2$, $b_1 \geq b_2$, then
  $a_1 b_1+a_2b_2 \geq a_1b_2+a_2b_1$.}  $\sum_i a_i b_i$, and is
obtained when the highest degree values match.  Our degree sequence
bound holds even when the input relations are allowed to be bags.
Furthermore, we prove (Theorem~\ref{th:main:polymatroid}) that this
bound is always below the AGM and polymatroid bounds, although the
latter restrict the relations to be sets.  To prove this we had to
develop a new, explicit formula for the polymatroid bound for
Berge-acyclic queries, which is of independent interest
(Theorem~\ref{thm:polyamtroid:bound:sum:of:chains}).

{\bf Compact Representation} A full degree sequence is about as large
as the relation instance, while cardinality estimators need to run in
sub-linear time.  Fortunately, a degree sequence can be represented
compactly using piece-wise constant function, called a {\em staircase
  function}, as illustrated in Fig.~\ref{fig:degree-sequence-example}.
Our next result, Theorem~\ref{thm:compress}, is an algorithm for the
degree sequence bound that runs in quasi-linear time (i.e. linear plus
a logarithmic factor) in the size of the representation, independent
of the size of the instance.  The algorithm makes some rounding errors
(Lemma~\ref{lmm:compress}), hence its output may be slightly larger
than the exact bound, however we prove that it is still lower than the
AGM and polymatroid bounds
(Theorem~\ref{th:approximate:agm:polymatroid}).  The algorithm can be
used in conjunction with a compressed representation of the degree
sequence.  By using few buckets and upper-bounding the degree sequence
one can trade off the memory size and estimation time for accuracy.
At one extreme, we could upper bound the entire sequence using a
single bucket with the constant $d_1$, at the other extreme we could
keep the complete sequence.  Neither the AGM bound nor the polymatroid
bound have this tradeoff ability.

{\bf Max Tuple Multiplicity} Despite using more information than
previous upper bounds, our bound can still be overly pessimistic,
because it needs to match the most frequent elements in all
attributes.  For example, suppose a relation has two attributes whose
highest degrees are $a_1$ and $b_1$ respectively. Its worst-case
instance is a bag and must include some tuple that occurs
$\min(a_1, b_1)$ times. Usually, $a_1$ and $b_1$ are large, since they
represent the frequencies of the worst heavy hitters in the two
columns, but in practice they rarely occur together $\min(a_1,b_1)$
times. To avoid such worst-case matchings, we use one additional piece
of information on each base table: the max multiplicity over all
tuples, denoted $B$. Usually, $B$ is significantly smaller than the
largest degrees, and, by imposing it as an additional constraint, we
can significantly improve the query's upper bound; in particular, when
$B=1$ then the relation is restricted to be a set. Our main results in
Theorems~\ref{th:main:star} and~\ref{th:general} extend to max tuple
multiplicities, but in some unexpected ways. The worst-case relation,
while still {\em tight}, is not a conventional relation: it may have
tuples that occur more than $B$ times, and, when the relation has 3 or
more attributes it may even have tuples with negative
multiplicities. Nevertheless, these rather unconventional worst-case
relations provide an even better degree sequence bound than by
ignoring $B$.

\begin{ex} \label{ex:intro} To give a taste of our degree-sequence bound, consider the full conjunctive query $Q(\cdots) = R(X,\cdots)\Join S(X,Y,\cdots) \Join T(Y,\cdots)$, where we omit showing attributes that appear in only one of the relations. Alternatively, we can write $Q(X,Y) = R(X)\Join S(X,Y) \Join T(Y)$ where $R, S, T$ are bags rather than sets.  Assume the following degree sequences:
  \begin{align}
    \bm d^{(R.X)} = & \,(3,2,2) & \bm d^{(T.Y)} = & \,(2,1,1,1)&
    \bm d^{(S.X)} = & \,(5,1) & \bm d^{(S.Y)} = & \,(3,2,1) \label{eq:ex:intro:d}
  \end{align}
  The AGM bound uses only the cardinalities, which are:
  \begin{align*}
    |R| = & \,7 & |S| = & \,6 & |T| = & \,5
  \end{align*}
  The AGM bound\footnote{Recall that each of the three relations has
    private variables, e.g. $R(X,U), S(X,Y,V,W), T(Y,Z)$.  The only
    fractional edge cover is $1,1,1$.}  is
  $|R| \cdot |S| \cdot |T| = 210$.  The extension to degree
  constraints in~\cite{DBLP:conf/pods/Khamis0S17} uses in addition the
  maximum degrees:
  \begin{align*}
   \texttt{deg}(R.X) = & \,3 & \texttt{deg}(S.X) = & \,5 
   &\texttt{deg}(S.Y) = & \,3 & \texttt{deg}(T.Y) = & \,2
  \end{align*}
  and the bound is the minimum between the AGM bound and the following
  quantities:
  \begin{small}
  \begin{align*}
    & |R|\cdot \texttt{deg}(S.X) \cdot \texttt{deg}(T.Y) = 7\cdot 5 \cdot 2=70 \\
    & \texttt{deg}(R.X) \cdot |S| \cdot \texttt{deg}(T.Y) = 3 \cdot 6  \cdot 2 = 36 \\
    & \texttt{deg}(R.X) \cdot \texttt{deg}(S.Y) \cdot |T| = 3 \cdot 3  \cdot 5 =45
  \end{align*}
  \end{small}
  Thus, the degree-constraint bound is improved to 36.

  Our new bound is given by the answer to the query on the worst-case instance of the relations $R, S, T$, shown here
  together with their multiplicities (recall that they are bags):
  \begin{small}
  \begin{align*}
    R=\,
    &
      \begin{array}{|l|l}\cline{1-1}
        a & 3 \\
        b & 2 \\
        c & 2 \\ \cline{1-1}
      \end{array},
    & S=\,
          &
            \begin{array}{|l|l|l} \cline{1-2}
              a & u & 3 \\
              a & v & 2 \\
              b & w & 1 \\ \cline{1-2}
            \end{array},
    & T=\,
    &
      \begin{array}{|l|l}\cline{1-1}
        u & 2 \\
        v & 1 \\
        w & 1 \\
        z & 1 \\ \cline{1-1}
      \end{array},
  \end{align*}
  \end{small}

  The three relations have the required degree sequences, for example
  $S.X$ consists of 5 $a$'s and 1 $b$, thus has degree sequence
  $(5,1)$.  Notice the matching principle: we assumed that the most
  frequent element in $R.X$ and $S.X$ are the same value $a$, and that
  the most frequent values in $S.X$ and in $S.Y$ occur together.  On
  this instance, we compute the query and obtain the answer $Q$.
  \begin{small}
  \begin{align*}
    Q =&\,  \begin{array}{|l|l|l} \cline{1-2}
      a & u & 3\cdot 3\cdot 2=18\\
      a & v & 3 \cdot 2 \cdot 1=6\\
      b & w & 2 \cdot 1 \cdot 1 =2\\  \cline{1-2}
    \end{array}
    & S'=&\,
      \begin{array}{|l|l|l} \cline{1-2}
        a & u & 2 \\
        a & v & 2 \\
        a & w & 1 \\
        b & u & 1 \\ \cline{1-2}
      \end{array}    
  \end{align*}
  \end{small}
  The upper bound is the size of the answer on this instance, which is
  $18+6+2=26$, and it improves over 36.  Here, the improvement is relatively minor, but this is a consequence of the short example. In practice, degree sequences often have a long tail, i.e. with a few large leading degrees $d_1, d_2, \ldots$ followed by very many
  small degrees $d_m, d_{m+1}, \ldots, d_n$ (with a large $n$).  In
  that case the improvements of the new bound can be very significant.

  Suppose now that we have one additional information about $S$: every
  tuple occurs at most $B=2$ times.  Then we need to reduce the
  multiplicity of $(a,u)$, and the new worst-case instance,
  denoted $S'$, is the following relation which decreases the cardinality bound to 25.
\end{ex}

%


%% file: sec-2-problem.tex
\section{Problem Statement}

\label{sec:problem}

{\bf Tensors} In this paper, it is convenient to define tensors using a named perspective, where each dimension is associated with a variable. We write variables with capital letters $X, Y, \ldots$ and sets of variables with boldface, $\bm X, \bm Y, \ldots$ We assume that each variable $X$ has an associated finite domain $D_X \defeq [n_X]$ for some number $n_X \geq 1$.  For any set of variables $\bm X$ we denote by $D_{\bm X} \defeq \prod_{Z \in \bm X} D_Z$.  We use lower case for values, e.g. $z \in D_Z$ and boldface for tuples, e.g. $\bm x \in D_{\bm X}$.  An {\em $\bm X$-tensor}, or simply a tensor when $\bm X$ is clear from the context, is $\bm M \in \R^{D_{\bm X}}$. We say that $\bm M$ has $|{\bm X}|$ dimensions.  Given two $\bm X$-tensors $\bm M, \bm N$, we write $\bm M \leq \bm N$ for the component-wise order ($M_{\bm x} \leq N_{\bm x}$, for all $\bm x$). If $\bm X, \bm Y$ are two sets of variables, then we denote their union by $\bm X \bm Y$.  If, furthermore, $\bm X, \bm Y$ are disjoint, and $\bm x \in D_{\bm X}, \bm y \in D_{\bm Y}$, then we denote by $\bm x \bm y \in D_{\bm X \bm Y}$ the concatenation of the two tuples.

\begin{defn} \label{def:tensor:ops} Let $\bm M, \bm N$ be an
  $\bm X$-tensor, and a $\bm Y$-tensor respectively.  Their tensor
  product is the following $\bm X \bm Y$-tensor:
  \begin{align}
    \forall \bm z \in D_{\bm X\bm Y}: &&   (\bm M \otimes \bm N)_{\bm z} \defeq &  M_{\pi_{\bm X}(\bm z)} \cdot  N_{\pi_{\bm Y}(\bm z)} \label{eq:tensor:prod}
  \end{align}
  If $\bm X, \bm Y$ are disjoint and $\bm M$ is an
  $\bm X \bm Y$-tensor then we define its $\bm X$-summation to be the
  following $\bm Y$-tensor:
  \begin{align}
    \forall \bm y \in D_{\bm Y}: && (\texttt{SUM}_{\bm X}(\bm M))_{\bm y} \defeq  & \sum_{\bm x \in D_{\bm X}} M_{\bm x \bm y}
  \end{align}
  If $\bm M, \bm N$ are $\bm X\bm Y$ and $\bm Y\bm Z$ tensors, where
  $\bm X, \bm Y, \bm Z$ are disjoint sets of variables, then their dot
  product is the $\bm X\bm Z$-tensor:
  \begin{align}
    \forall \bm x \in D_{\bm X}, \bm z \in D_{\bm Z}: &&     (\bm M \cdot \bm N)_{\bm x \bm z} \defeq & \texttt{SUM}_{\bm Y}(\bm M \otimes \bm N)_{\bm x \bm z}= \sum_{\bm y \in D_{\bm Y}}M_{\bm x \bm y} N_{\bm y \bm z} \label{eq:def:tensor:dot:prod}
  \end{align}
\end{defn}

In other words, in this paper we use $\otimes$ like a natural join.  For example, if
$\bm M$ is an $IJ$-tensor (i.e. a matrix) and $\bm N$ is an
$KL$-tensor, then $\bm M \otimes \bm N$ is the Kronecker product; if
$\bm P$ is an $IJ$-tensor (like $\bm M$) then $\bm M \otimes \bm P$ is
the element-wise product.  The dot product sums out the common
variables, for example if $\bm a$ is a $J$-tensor, then
$\bm M \cdot \bm a$ is the standard matrix-vector multiplication, and
its result is an $I$-tensor.  The following is easily verified.  If
$\bm M$ is an $\bm X$-tensor, $\bm N$ is an $\bm Y$-tensor and
$\bm X, \bm Y$ are disjoint sets of variables, then:
\begin{align}
\forall \bm X_0 \subseteq \bm X, \forall \bm Y_0 \subseteq \bm Y:&&   \texttt{SUM}_{\bm X_0 \bm Y_0}(\bm M \otimes \bm N) = & \texttt{SUM}_{\bm X_0}(\bm M) \otimes  \texttt{SUM}_{\bm Y_0}(\bm N) \label{eq:push:aggregates:down}
\end{align}

%

{\bf Permutations} A permutation on $D = [n]$ is a bijective function
$\sigma : D\rightarrow D$; the set of permutations on $D$ is denoted
$S_D$, or simply $S_n$.  If $\bm D = D_1 \times \cdots \times D_k$
then we denote by
$S_{\bm D} \defeq S_{D_1} \times \cdots \times S_{D_k}$.  Given an
$\bm X$-tensor $\bm M \in \R^{D_{\bm X}}$ and permutations
$\bm \sigma \in S_{D_{\bm X}}$, the {\em $\bm \sigma$-permuted
  $\bm X$-tensor} is $\bm M \circ \bm \sigma \in \R^{D_{\bm X}}$:
\begin{align*}
\forall \bm x \in D_{\bm X}:&&  (\bm M \circ \bm \sigma)_{\bm x} \defeq & \bm M_{\bm \sigma(\bm x)}
\end{align*}
Sums are invariant under permutations, for example if
$\bm a, \bm b \in \R^{D_Z}$ are $Z$-vectors and $\sigma \in S_{D_Z}$,
then
$(\bm a \circ \sigma) \cdot (\bm b \circ \sigma) = \bm a \cdot \bm b$,
because $\sum_{i \in D_Z} a_{\sigma(i)}b_{\sigma(i)} = \sum_{i \in D_Z} a_ib_i$.


{\bf Queries} A full conjunctive
query $Q$ is:
\begin{align}
  Q(\bm X) = & \Join_{R \in \bm R} R(\bm X_R) \label{eq:q}
\end{align}
where $\bm R \defeq \bm R(Q)$ denotes the set of its relations,
$\bm X$ is a set of variables, and $\bm X_R \subseteq \bm X$ for each
relation $R \in \bm R$.  The {\em incidence graph} of $Q$ is the
following bipartite graph:
$T \defeq (\bm R\cup \bm X, E \defeq \setof{(R,Z)}{Z \in \bm X_R})$.
It can be shown that $Q$ is {\em
  Berge-acyclic}~\cite{DBLP:journals/jacm/Fagin83} iff its incidence
graph is an undirected tree (see appendix~\ref{app:background}).
Unless otherwise stated, all queries in this paper are assumed to be
full, Berge-acyclic conjunctive queries.
We use bag semantics for query evaluation, and represent an {\em
  instance} of a relation $R \in \bm R$ by an $\bm X_R$-tensor,
$\bm M^{(R)}$, where $M^{(R)}_t$ is defined to be the multiplicity of
the tuple $t \in D_{\bm X_R}$ in the bag $R$.  The number of tuples in
the answer to $Q$ is:
\begin{align}
  |Q| =\, & \texttt{SUM}_{\bm X} \left(\bigotimes_{R \in \bm R} \bm M^{(R)}\right) \label{eq:q:size}
\end{align}

\begin{ex} \label{ex:rstk}
  Consider the following query:
  \begin{align*}
  Q(X,Y,Z,U,V,W) = & R(X,Y) \Join S(Y,Z,U) \Join T(U,V) \Join K(Y,W)
  \end{align*}
  Its incidence graph is
  $T = (\set{R,\ldots,K}\cup \set{X, \ldots, W},
  \set{(R,X),(R,Y),(S,Y),\ldots,(K,W)})$ and is an undirected tree.
  An instance of $R(X,Y)$ is represented by a matrix
  $\bm M^{(R)} \in \R^{D_X \times D_Y}$, where $M^{(R)}_{xy}=$ the
  number of times the tuple $(x,y)$ occurs in $R$.  Similarly, $S$ is
  represented by a tensor
  $\bm M^{(S)} \in \R^{D_Y \times D_Z \times D_U}$.  The size of the
  query's output is:
  \begin{align*}
    |Q| = & \sum_{x,y,z,u,v,w} M^{(R)}_{xy} M^{(S)}_{yzu} M^{(T)}_{uv} M^{(K)}_{yw}
  \end{align*}
\end{ex}

{\bf Degree Sequences} We denote by
$\R_+ \defeq \setof{x}{x \in \R, x \geq 0}$ and we say that a vector
$f\in \R_+^{[n]}$ is non-increasing if $f_{r-1} \geq f_r$ for $r=2,n$.

\begin{defn} \label{def:consistent} Fix a set of variables $\bm X$,
  with domains $D_Z$, $Z \in \bm X$.  A {\em degree sequence}
  associated with the dimension $Z \in \bm X$ is a non-increasing
  vector $\bm f^{(Z)} \in \R_+^{D_Z}$.  We call the index $r$ the {\em
    rank}, and $f_r^{(Z)}$ the {\em degree at rank $r$}.  An
  $\bm X$-tensor $\bm M$ is {\em consistent} w.r.t. $\bm f^{(Z)}$ if:
  \begin{align}
    \texttt{SUM}_{\bm X-\set{Z}}(\bm M) \leq & \bm f^{(Z)} \label{eq:marginals}
  \end{align}
  $\bm M$ is consistent with a tuple of degree sequences $\bm f^{(\bm X)} \defeq (\bm f^{(Z)})_{Z \in \bm X}$, if it is consistent with every $\bm f^{(Z)}$.  Furthermore, given $B \in \R_+ \cup \set{\infty}$, called the {\em max tuple multiplicity}, we say that $\bm M$ is {\em consistent} w.r.t. $B$ if $\bm M_t \leq B$ for all $t \in D_{\bm X}$.  
We denote:
  \begin{align}
    \calM_{\bm f^{(\bm X)},B}\defeq & \setof{\bm M \in \R^{D_{\bm X}}}{\bm M\mbox{ is  consistent  with } \bm f^{(\bm X)}, B} \nonumber\\
    \calM^+_{\bm f^{(\bm X)},B}\defeq & \setof{\bm M \in \R_+^{D_{\bm X}}}{\bm M\mbox{ is  non-negative and consistent  with } \bm f^{(\bm X)}, B} \label{eq:calm}
  \end{align}
\end{defn}

For a simple illustration consider two degree sequences
$\bm f \in \R^{[m]}, \bm g\in \R^{[n]}$.
$\calM_{\mathbf{f}, \mathbf{g},\infty}$ is the set of matrices $\bm M$
whose row-sums and column-sums are $\leq \bm f$ and $\leq \bm g$
respectively; $\calM^+_{\mathbf{f}, \mathbf{g},\infty}$ is the subset
of non-negative matrices; $\calM^+_{\mathbf{f}, \mathbf{g},B}$ is the
subset of matrices that also satisfy $M_{ij} \leq B$, $\forall i,j$.


{\bf Problem Statement} Fix a query $Q$.  For each relation $R$, we
are given a set of degree sequences
$\bm f^{(R,\bm X_R)} \defeq \left(\bm f^{(R,Z)}\right)_{Z \in \bm
  X_R}$, and a tuple multiplicity $B^{(R)} \in \R_+\cup \set{\infty}$.
We are asked to find the maximum size of $Q$ over all database
instances consistent with all degree sequences and tuple
multiplicities. To do this, we represent a relation instance $R$ by an
unknown tensor
$\bm M^{(R)}\in\calM^+_{\mathbf{f}^{(R,\bm X_R)},B^{(R)}}$ and an
unknown set of permutations $\bm\sigma^{(R)}\in S_{D_{\bm X_R}}$, and
solve the following problem:

\begin{pbm}[Degree Sequence Bound] \label{prob:main} Solve the
  following optimization problem:
  \begin{align}
    \mbox{Maximize:\ } & |Q|= \texttt{SUM}_{\bm X}\left( \bigotimes_{R \in \bm R}(\bm  M^{(R)} \circ \bm \sigma^{(R)})\right) \label{eq:prob:main}\\
    \mbox{Where:\ } & \forall R \in \bm R,\,\bm \sigma^{(R)} \in S_{D_{\bm X_R}},\,\bm M^{(R)} \in \calM^+_{\bm f^{(R, \bm X_R)},B^{(R)}}\nonumber
  \end{align}
\end{pbm}

This is a non-linear optimization problem: while the set $\calM^+$
defined in Eq.~\eqref{eq:calm} is a set of linear constraints, the
objective~\eqref{eq:prob:main} is non-linear.  In the rest of the
paper we describe an explicit formula for the degree sequence bound,
which is optimal (i.e. tight) when $B^{(R)}=\infty$, for all $R$, and
is optimal in a weaker sense in general.

\begin{ex} \label{ex:intro:cont} Continuing Example~\ref{ex:intro}, the four degree sequences in~\eqref{eq:ex:intro:d} correspond to the variables in each relation $R.X$, $S.X$, $S.Y$, and $T.Y$. Since $S.X$ has a shorter degree sequence than $R.X$, we pad it with a 0, so it becomes $\bm d^{(S.X)}=(5,1,0)$; similarly for $\bm d^{(S.Y)}$.  Instead of values $c,b,a$, we use indices $1,2,3$, similarly $u,v,w,z$ becomes $1,2,3,4$.  For example,
$S =
{\footnotesize
  \begin{array}{|l|l|l} \cline{1-2}
    3 & 1 & 3 \\
    3 & 2 & 2 \\
    2 & 3 & 1 \\ \cline{1-2}
  \end{array}}$
is isomorphic to the instance in Example~\ref{ex:intro}. It is represented by $\bm M \circ (\sigma, \tau)$ where the matrix $\bm M=$
${\footnotesize
  \begin{pmatrix}
    3 & 2 & 0 & 0 \\ 0 & 0 & 1 & 0 \\ 0 & 0 & 0 & 0
  \end{pmatrix}}$, (its row-sums are $5,1,0$ and column-sums are
$3, 2, 1, 0$, as required) and the permutations are, in two-line
notation, $\sigma \defeq {\footnotesize \begin{pmatrix} 1 & 2 & 3 \\ 3 & 2 & 1
  \end{pmatrix}}$ and $\tau\defeq$ the identity. Similarly, the
relations $R, T$, are represented by vectors $\bm a, \bm b$ and
permutations $\theta, \rho$. The bound of $Q$ is the maximum value of
$\sum_{i=1,3} \sum_{j=1,4} M_{\sigma(i)\tau(j)} a_{\theta(i)}
b_{\rho(j)}$, where $\bm M, \bm a, \bm b$ are consistent with the
given degree sequences, and $\sigma, \tau, \theta, \rho$ are
permutations.  This is a special case of Eq.~\eqref{eq:prob:main}.
\end{ex}

%

%% file: sec-3-star.tex
\section{The Star Query}

\label{sec:bowtie}

We start by computing the degree sequence bound for a {\em star
  query}, which is defined as:
\begin{align}
 Q_{\texttt{star}} =  & S(X_1, \ldots, X_d) \Join R^{(1)}(X_1) \Join \cdots \Join R^{(d)}(X_d)
\label{eq:star:query}
\end{align}
Assume that the domain of each variable $X_p$ is $[n_p]$ for some
$n_p > 0$, and denote by
$[\bm n] \defeq [n_1] \times \cdots \times [n_d]$.  Later, in
Sec.~\ref{sec:general}, we will use the bound for $Q_{\texttt{star}}$
as a building block to compute the degree sequence bound of a general
query $Q$.  There, $S$ will be one of the relations of the query, for
which we know the degree sequences $\bm f^{(X_p)} \in \R_+^{[n_p]}$,
$p=1,\ldots,d$ and tuple bound $B$, while the unary relations
$R^{(1)}, \ldots, R^{(d)}$ will be results of subqueries, which are
unknown.  The instance of each $R^{(p)}$ is given by an unknown vector
$\bm a^{(p)} \in \R_+^{[n_p]}$, which we can assume w.l.o.g.  to be
non-increasing, by permuting the domain of $X_p$ in both $S$ and in
$R^{(p)}$.  Therefore, $S$ will be represented by
$\bm M \circ \bm \sigma$, where
$\bm M \in \calM^+_{\bm f^{(\bm X)},B}$ is some tensor and
$\bm \sigma$ some permutation, and the size of $Q_{\texttt{star}}$ is:
\begin{align}
|Q_{\texttt{star}}| =  & \sum_{(i_1,\ldots,i_d)\in [\bm n]}\left(\bm M\circ \bm \sigma \right)_{i_1\cdots i_d}\cdot a^{(X_1)}_{i_1}\cdots a^{(X_d)}_{i_d} \label{eq:probl:star}
\end{align}
Equivalently: $|Q_{\texttt{star}}| =  \texttt{SUM}_{\bm X}\left((\bm M \circ \bm \sigma) \otimes \bigotimes_p \bm a^{(X_p)}\right) = (\bm M \circ \bm \sigma) \cdot\bm a^{(X_1)} \cdots \bm a^{(X_d)}$.

Our goal is to find the unknown $\bm M \circ \bm \sigma$ for which
$|Q_{\texttt{star}}|$ is maximized, no matter what the unary relations
are.  It turns out that $\bm \sigma$ can always be chosen the identity
permutation, thus it remains to find the optimal $\bm M$, which we
denote by $\bm C$.  This justifies:


\begin{pbm}[Worst-Case Tensor] \label{prob:pessimistic:tensor} Fix
  $\bm f^{(\bm X)}$, $B$.  Find a tensor
  $\bm C \in \calM_{\bm f^{(\bm X)},\infty}$ such that, for all
  $\bm \sigma \in S_{[\bm n]}, \bm M \in \calM^+_{\bm f^{(\bm X)},B}$,
  and all non-increasing vectors
  $\bm a^{(X_1)} \in \R_+^{[n_1]}, \ldots, \bm a^{(X_d)} \in
  \R_+^{[n_d]}$:
  \begin{align}
    (\bm M \circ \bm \sigma) \cdot \bm a^{(X_1)} \cdots \bm a^{(X_d)} \leq & \bm C \cdot \bm a^{(X_1)} \cdots \bm a^{(X_d)}
\label{eq:c:is:max}
  \end{align}
  %
\end{pbm}

In the rest of this section we describe the solution $\bm C$.  If all
entries in $\bm C$ are $\geq 0$ and $\leq B$, then
$\bm C \in \calM^+_{\bm f^{(\bm X)},B}$ and, by setting
$\bm M \defeq \bm C$ and $\bm \sigma\defeq$ the identity permutations,
the relation $S$ represented by $\bm M \circ \bm \sigma$ maximizes
$|Q_{\texttt{star}}|$, achieving our goal.  But, somewhat
surprisingly, we found that sometimes this worst-case $\bm C$ has
entries $>B$ or $<0$, yet it still achieves our goal of a tight upper
bound for $|Q_{\texttt{star}}|$.  This is why we allow
$\bm C \in \calM_{\bm f^{(\bm X)},\infty}$.

%
%
%
%
%
%
%

Let $\Delta_Z$ denote the {\em discrete derivative} of a
$\bm X$-tensor w.r.t. a variable $Z \in \bm X$, and $\Sigma_Z$ denote
the {\em discrete integral}.  Formally, if $\bm a \in \R^{[n]}$ is an
$Z$-vector, then, setting $a_0 \defeq 0$:
\begin{align}
  \forall i \in [n]: && (\Delta_Z \bm a)_i \defeq & a_i-a_{i-1} & (\Sigma_Z \bm a)_i = & \sum_{j=1,i} a_j
\label{eq:delta:sigma}
\end{align}
Notice that:
\begin{align}
\Sigma_Z (\Delta_Z \bm a) = & \Delta_Z (\Sigma_Z \bm a) = \bm a &  \texttt{SUM}_Z(\Delta_Z \bm a)=a_n \label{eq:sigma:delta:cancel}
\end{align}
The subscript in $\Delta, \Sigma$ indicates on which variable they
act.  For example, if $\bm M$ is an $XYZ$-tensor, then
$(\Delta_Y \bm M)_{xyz} \defeq M_{xyz} - M_{x(y-1)z}$.
One should think of the three operators
$\Delta_X, \Sigma_X, \texttt{SUM}_X$ as analogous to the
continuous operators $\frac{d\cdots}{dx}$, $\int \cdots dx$,
$\int_0^n\cdots dx$.
\\
\begin{defn} \label{def:pessimiistic:tensor} The {\em value} tensor,
  $\bm V^{\bm f^{(\bm X)},B}\in \R_+^{[\bm n]}$, is defined by the
  following linear optimization problem:
\begin{align}
\forall \bm m \in [\bm n]: && V_{\bm m}^{\bm f^{(\bm X)},B} \defeq  \mbox{\ Maximize:\ } &  \sum_{\bm s \leq \bm m} M_{\bm s}\label{eq:def:v}\\
  && \mbox{Where:\ } & \bm M \in \calM^+_{\bm f^{(\bm X)},B}\nonumber
\end{align}

The {\em worst-case} tensor,
$\bm C^{\bm f^{(\bm X)},B}\in \R^{[\bm n]}$, is defined as:
  \begin{align}
    && \bm C^{\bm f^{(\bm X)},B} \defeq & \Delta_{X_1}\cdots \Delta_{X_d} \bm V^{\bm f^{(\bm X)},B} \label{eq:def:c}
  \end{align}
\end{defn}

We will drop the superscripts when clear from the context, and write
simply $\bm V, \bm C$. Our main result in this section is:

\begin{thm} \label{th:main:star} Let $\bm f^{(\bm X)}, B$ be given as
  above, and let $\bm V, \bm C$ defined
  by~\eqref{eq:def:v}-\eqref{eq:def:c}. Then:
  \begin{enumerate}
  \item \label{item:th:main:star:2} $\bm C$ is a solution to
    Problem~\ref{prob:pessimistic:tensor}, i.e.
    $\bm C \in \calM_{\bm f^{(\bm X)},\infty}$ and it satisfies
    Eq.~\eqref{eq:c:is:max}.  Furthermore, it is tight in the
    following sense: there exists a tensor
    $\bm M \in \calM^+_{\bm f^{(\bm X)},B}$ and non-increasing vectors
    $\bm a^{(p)} \in \R_+^{[n_p]}$, $p=1,d$, such that
    inequality~\eqref{eq:c:is:max} (with $\bm \sigma$ the identity) is
    an equality.
  \item \label{item:th:main:star:3} If there exists any solution
    $\bm C' \in \calM^+_{\bm f^{(\bm X)}, B}$ to
    Problem~\ref{prob:pessimistic:tensor}, then $\bm C'=\bm C$.
  \item \label{item:th:main:star:6} When the number of dimensions is
    $d=2$ then $\bm C$ is integral and non-negative. If $d\geq 3$,
    $\bm C$ may have negative entries.
  \item \label{item:th:main:star:1} If $B < \infty$, then $\bm C$ may
    not be consistent with $B$, even if $d=2$.
  \item \label{item:th:main:star:4} For any non-increasing vectors
    $\bm a^{(X_p)} \in \R_+^{[n_p]}$, $p=2,d$, the vector
    $\bm C \cdot \bm a^{(X_2)} \cdots \bm a^{(X_d)}$ is in
    $\R^{[n_1]}_+$ and non-increasing.
  \item \label{item:th:main:star:5} Assume $B = \infty$.  Then the following holds:
    \begin{align}
      \forall \bm m \in [\bm n]: &&    V_{\bm m} = & \min\left(F^{(X_1)}_{m_1}, \ldots, F^{(X_d)}_{m_d}\right) \label{eq:v:is:max}
    \end{align}
    where $F^{(X_p)}_r \defeq \sum_{j\leq r} f^{(X_p)}_j$ is the CDF
    associated to the PDF $\bm f^{(X_p)}$, for $p=1,d$. Moreover,
    $\bm C$ can be computed by Algorithm~\ref{alg:fast_C_alg}, which
    runs in time $\mathbf{O}(\sum_p n_p)$.  This further implies that
    $\bm C \geq 0$, in other words
    $\bm C \in \calM^+_{\bm f^{(\bm X)}, \infty}$.
  \end{enumerate}
\end{thm}


\begin{algorithm}
\caption{Efficient construction of $\bm C$ when $B=\infty$}
\label{alg:fast_C_alg}
\begin{algorithmic}
\STATE $\forall p=1,d: s_p \leftarrow 1$; \ \ \ $\bm C=0$;
\WHILE{$\forall p : s_p < n_p$}
\STATE $p_{min} \leftarrow \argmin_p(f^{(X_p)}_{s_p})$ \ \ \ \ $d_{min} \leftarrow \min_p(f^{(X_p)}_{s_p})$
\STATE $C_{s_1,\ldots,s_d} \leftarrow d_{min}$
\STATE $\forall p=1,d: f^{(X_p)}_{s_p} \leftarrow f^{(X_p)}_{s_p}-d_{min}$
\STATE $s_{p_{min}} \leftarrow s_{p_{min}} + 1$
\ENDWHILE
\RETURN $C$
\end{algorithmic}
\end{algorithm}

In a nutshell, the theorem asserts that the tensor $\bm C$ defined
in~\eqref{eq:def:c} is the optimal solution to
Problem~\ref{prob:pessimistic:tensor}; this is stated in
item~\ref{item:th:main:star:2}.  Somewhat surprisingly, $\bm C$ may be
inconsistent w.r.t. $B$, and may even be negative.  When that happens,
then, by item~\ref{item:th:main:star:3}, no consistent solution exists
to Problem~\ref{prob:pessimistic:tensor}, hence we have to make do
with $\bm C$.  In that case $\bm C$ may not represent a traditional
bag $S$, for example if it has entries $<0$. However, this will not be
a problem for computing the degree sequence bound in
Sec.~\ref{sec:general}, because all we need is to compute the product
$\bm C \cdot \bm a^{(X_2)} \cdots \bm a^{(X_d)}$, which we need to be
non-negative, and non-increasing: this is guaranteed by
item~\ref{item:th:main:star:4}. The last item gives more insight into
$\bm V$ and, by extension, into $\bm C$.  Recall that $V_{\bm m}$,
defined by~\eqref{eq:def:v}, is the largest possible sum of values of
a consistent $m_1 \times m_2 \times \cdots \times m_d$ tensor $\bm M$.
Since the sum in each hyperplane $X_1 = r$ of $\bm M$ is
$\leq f^{(X_1)}_r$, it follows that
$\sum_{\bm s \leq \bm m}M_{\bm s} \leq \sum_{r=1,m_1} f^{(X_1)}_r
\defeq F^{(X_1)}_{m_1}$.  Repeating this argument for each dimension
$X_p$ implies that $V_{\bm m} \leq \min_{p=1,d} (F^{(X_p)}_{m_p})$.
Item~\ref{item:th:main:star:5} states that this becomes an equality,
when $B=\infty$.

\begin{ex}
  \label{ex:simple:c} Suppose that we want to maximize
  $\bm a^T \cdot \bm M \cdot \bm b$, where $\bm M$ is a $3\times 4$
  matrix with degree sequences $\bm f=(6,3,1)$ and $\bm g=(4,3,2,1)$;
  assume $B= \infty$.  The vectors $\bm a, \bm b$ are non-negative and
  non-increasing, but otherwise unknown.  The theorem asserts that
  this product is maximized by the worst-case matrix $\bm C$.  We
  show here the matrices $\bm C$ and $\bm V$ defined
  by~\eqref{eq:def:v} and~\eqref{eq:def:c}, together with degree
  sequences $\bm f, \bm g$ next to $\bm C$, and the cumulative
  sequences $\bm F = \Sigma \bm f, \bm G = \Sigma \bm g$ next to
  $\bm V$:
  \begin{small}
  \begin{align*}
    \bm C =\,\,
    &
      \begin{pNiceMatrix}[first-row,first-col]
        \ & 4 & 3 & 2 & 1 \\
        6 & 4 & 2 & 0 & 0 \\
        3 & 0 & 1 & 2 & 0 \\
        1 & 0 & 0 & 0 & 1
      \end{pNiceMatrix}
   & \bm V=
   & 
     \begin{pNiceMatrix}[first-row,first-col]
       \ & 4 & 7 & 9 & 10 \\
       6 & 4 & 6 & 6 & 6 \\
       9 & 4 & 7 & 9 & 9 \\
       10 & 4 & 7 & 9 & 10
     \end{pNiceMatrix}
  \end{align*}
  \end{small}
  We can check that $V_{m_1m_2} = \min(F_{m_1},G_{m_2})$; for example
  $V_{31} = \min(10,4)=4$.  The worst-case matrix $\bm C$ is defined
  as the second discrete derivative of $\bm V$, more precisely
  $C_{m_1m_2}=V_{m_1m_2}-V_{m_1-1,m_2}-V_{m_1,m_2-1}+V_{m_1-1,m_2-1}$.
  Alternatively, $\bm C$ can be computed greedily, using
  Algorithm~\ref{alg:fast_C_alg}: start with
  $C_{11} \leftarrow \min(f_1,g_1) = 4$, decrease both $f_1, g_1$ by
  4, set the rest of column 1 to 0 (because now $g_1=0$) and continue
  with $C_{12}$, etc.  Another important property, which we will prove
  below in the Appendix (Eq.~\eqref{eq:step:3}), is that, for all
  $m_1, m_2$, $\sum_{i \leq m_1, j\leq m_2} C_{ij}=V_{m_1m_2}$; for
  example $\sum_{i\leq 2, j\leq 3} C_{ij} = 4+2+1+2=9=V_{23}$.
\end{ex}

While the proof of Theorem~\ref{th:main:star} provides interesting insight into the structure of the degree sequence bound, it is not necessary for understanding the remainder of the paper and requires the introduction of additional notation and machinery. Therefore, for the sake of space and clarity, we omit it from the main text and instead include a proof of each item in the appendix section \ref{app:thm:main:star}.

%% file: sec-4-general.tex
\section{The Berge-Acyclic Query}

\label{sec:general}

We now turn to the general problem~\ref{prob:main}.  Fix a
Berge-acyclic query $Q$ with relations $\bm R \defeq \bm R(Q)$, degree
sequences $\bm f^{R,Z}$, and max tuple multiplicities $B^{(R)}$ as in
problem~\ref{prob:main}.

\subsection{The Degree Sequence Bound}

\begin{thm}
\label{th:general}
For any tensors
$\bm M^{(R)} \in \calM^+_{\bm f^{(R, \bm X_R)},B^{(R)}}$ and
permutations $\bm \sigma^{(R)}$, for $R \in \bm R$, the following
holds:
\begin{align}
  \texttt{SUM}_{\bm X}\left( \bigotimes_{R \in \bm R}(\bm  M^{(R)} \circ \bm \sigma^{(R)})\right)
\leq &
  \texttt{SUM}_{\bm X}\left( \bigotimes_{R \in \bm R} \bm C^{\bm f^{(R, \bm X_R)},B^{(R)}} \right)
\defeq DSB(Q)
       \label{eq:th:general}
\end{align}
where $\bm C^{\bm f^{(R, \bm X_R)},B^{(R)}}$ is the worst-case tensor
from Def.~\ref{def:pessimiistic:tensor}.  
\end{thm}

The theorem simply says that the upper bound to the query $Q$ can be
computed by evaluating $Q$ on the worst case instances, represented by
the worst case tensors $\bm C^{\bm f^{(R, \bm X_R)},B^{(R)}}$.  We
call this quantity the {\em degree sequence bound} and denote it by
$DSB(Q)$.  When all max tuple multiplicities $B^{(R)}$ are $\infty$,
then the bound is tight, because in that case every worst-case tensor
$\bm C^{\bm f^{(R, \bm X_R)},\infty}$ is in
$\calM^+_{\bm f^{(R, \bm X_R)},\infty}$ (by Th.~\ref{th:main:star}
item~\ref{item:th:main:star:5}); otherwise the bound may not be tight,
but it is locally tight, in the sense of Th.~\ref{th:main:star}
item~\ref{item:th:main:star:2}.

Before we sketch the main idea of the proof, we note that an immediate
consequence is that the degree sequence bound can be computed using a
special case of the FAQ algorithm~\cite{DBLP:conf/pods/KhamisNR16}.
We describe this briefly in Algorithm~\ref{alg:bottom:up}.  Recall
that the incidence graph of $Q$ is a tree $T$.  Choose an arbitrary
relation $\texttt{ROOT}\in \bm R(Q)$ and designate it as root, then
make $T$ a directed tree by orienting all its edges away from the
root. Denote by $\texttt{parent}(R) \in \bm X_R$ the parent node of a
relation $R\neq \texttt{ROOT}$, associate an $X$-vector $\bm a^{(X)}$
to each variable $X$, and a $\texttt{parent}(R)$-vector $\bm w^{(R)}$
to each relation name $R$, then compute these vectors by traversing
the tree bottom-up, as shown in Algorithm~\ref{alg:bottom:up}. Notice that, when $X$ is a leaf variable, then $\text{children}(X)=\emptyset$ and $\bm a^{(X)}=(1,1,\ldots,1)^T$; similarly, if $R(X)$ is leaf relation of arity 1 with variable $X$, then $\bm w^{(R)}$ is the degree sequence of its variable, because $\bm w^{(R)}= \bm C^{(\bm f^{(R,X)},B^{(R)})}=\bm f^{(R,X)}$.
%
We provide an example in Sec.~\ref{app:th:general}.  It follows:

\begin{cor} \label{cor:o:n} The degree sequence bound $DSB(Q)$ can be
  computed in time polynomial in the size of the largest domain (data
  complexity).
\end{cor}

\begin{algorithm}[t]
  \caption{Computing
    $DSB(Q) = \texttt{SUM}_{\bm X}\left( \bigotimes_{R \in \bm R} \bm C^{\bm f^{(R, \bm X_R)},B^{(R)}} \right)$}
\label{alg:bottom:up}
\begin{algorithmic}
\FOR{each variable  $X \in \bm X$ and non-root relation $R \in \bm R$, $R\neq \texttt{root}$, in bottom-up order}{
   \STATE $\bm a^{(X)} \defeq \bigotimes_{R \in \text{children}(X)} \bm w^{(R)}$\ \ \ \ // element-wise product
   \STATE $\bm w^{(R)} \defeq \bm C^{\bm f^{(R, \bm X_R)},B^{(R)}}\cdot\bm a^{(X_2)}\cdots\bm a^{(X_k)}$
\ // where $\bm X_R=(X_1,\ldots,X_k)$, $X_1=\texttt{parent}(R)$
 }
\ENDFOR
\RETURN $C^{\bm f^{(\texttt{root}, \bm X_{\texttt{ROOT}})},B^{(\texttt{ROOT})}}\cdot\bm a^{(X_1)}\cdot\bm a^{(X_2)}\cdots\bm a^{(X_k)}$
\end{algorithmic}
\end{algorithm}

In the rest of this section we sketch the proof of
Theorem~\ref{th:general}, 
mostly to highlight the role of item~\ref{item:th:main:star:4} of
Theorem~\ref{th:main:star}, and defer the formal details to
Appendix~\ref{app:th:general}.  Fix tensors $\bm M^{(R)}$ and
permutations $\bm \sigma^{(R)}$, for each $R \in \bm R$.  Choose one
relation, say $S \in \bm R$, assume it has $k$ variables
$X_1, \ldots, X_k$, then write the LHS of~\eqref{eq:th:general} as:
\begin{align}
  &
    \texttt{SUM}_{\bm X_S}\left( \left(\bm M^{(S)} \circ \bm \sigma^{(S)}\right)
    \otimes \bm b_1 \otimes \cdots \otimes \bm b_k\right) \label{eq:tree:product}
\end{align}
where each $\bm b_p$ is a tensor expression sharing only variable
$X_p$ with $S$, where we sum out all variables except $X_p$ (using
Eq.~\eqref{eq:push:aggregates:down}).  Compute the vectors $\bm b_p$
first, sort them in non-decreasing order, let $\tau_p$ be the
permutation that sorts $\bm b_p$, and
$\bm \tau\defeq (\tau_1,\ldots,\tau_k)$.  Then~\eqref{eq:tree:product}
equals:
\begin{align}
 &   \texttt{SUM}_{\bm X_S}\left( \left(\bm M^{(S)} \circ \bm  \sigma^{(R)}\circ \bm \tau\right) \otimes (\bm b_1 \circ\tau_1)  \otimes \cdots \otimes (\bm b_k \circ\tau_k) \right)
\label{eq:tree:product:2}
\end{align}
because sums are invariant under permutations.
Since each $\bm b_p \circ \tau_p$ is sorted, by
item~\ref{item:th:main:star:2} of Theorem~\ref{th:main:star}, the
expression above is $\leq$ to the expression obtained by replacing
$\bm M^{(S)} \circ \bm \sigma^{(S)}\circ \bm \tau$ with the worst-case
tensor $\bm C^{f^{(S,\bm X_S)},B^{(S)}}$.  Thus, every tensor could be
replaced by the worst-case tensor, albeit at the cost of applying some
new permutations $\tau_p$ to other expressions.  To avoid introducing
these permutations, we proceed as follows.  We choose an orientation
of the tree $T$, as in Algorithm~\ref{alg:bottom:up}, then prove
inductively, bottom-up the tree, that each tensor
$\bm M \circ \bm \sigma$ can be replaced by the worst-case tensor
$\bm C$ without decreasing the LHS of~\eqref{eq:th:general}, {\em and}
that the resulting vector (in the bottom-up computation) is sorted.
To prove this, we re-examine Eq.~\eqref{eq:tree:product}, assuming
$X_1$ is the parent variable of $S$.  By induction, all the tensors
occurring in $\bm b_2, \ldots, \bm b_k$ have already been replaced
with worst-case tensors, {\em and} their results are non-increasing
vectors.  Then, in Eq.~\eqref{eq:tree:product:2} it suffices to apply
the permutation $\tau$ to the parent expression $\bm b_1$ (which still
has the old tensors $\bm M \circ \bm \sigma$), use
item~\ref{item:th:main:star:2} of Theorem~\ref{th:main:star} to
replace $\bm M^{(S)} \circ \bm \sigma^{(S)}\circ \bm \tau$ by
$\bm C^{f^{(S,\bm X_S)},B^{(S)}}$, and, finally, use
item~\ref{item:th:main:star:4} of Theorem~\ref{th:main:star} to prove
that the result returned by the node $S$ is a non-decreasing vector,
as required.

\subsection{Connection to the AGM and Polymatroid Bounds}

\label{sec:connection:to:agm:polymatroid:bounds}

We prove now that $DSB(Q)$ is always below the
AGM~\cite{DBLP:conf/focs/AtseriasGM08} and the polymatroid
bounds~\cite{DBLP:conf/pods/Khamis0S17,DBLP:conf/pods/000118}.

The AGM bound is expressed in terms of the cardinalities of the
relations.  For each relation $R$, let $N_R$ be an upper bound on its
cardinality.  Then the AGM bound is
$AGM(Q) \defeq \min_{\bm w} \prod_R N_R^{w_R}$, where the vector
$\bm w = (w_R)_{R\in \bm R}$ ranges over the fractional edge covers of
the hypergraph associated to $Q$.  If a database instance satisfies
$|R| \leq N_R$ for all $R$, then the size of the query is
$|Q| \leq AGM(Q)$, and this bound is tight, i.e. there exists an
instance for which we have equality.

The polymatroid bound uses both the cardinality constraints $N_R$ and
the maximum degrees.  The general bound
in~\cite{DBLP:conf/pods/Khamis0S17} considers maximum degrees for any
subset of variables, but throughout this paper we restrict to degrees
of single variables, in which case the polymatroid bound is expressed
in terms of the quantities $N_R$ and $f_1^{(R,X)}$, one for each
relation $R$ and each of its variables $X$.  The AGM bound is the
special case when $f_1^{(R,X)} = N_R$ for all $R$.  We review the
general definition of the polymatroid bound in
Sec.~\ref{app:chain:bound}, but will mention that no closed formula is
known for polymatroid bound, similar to the AGM formula.  We give here
the first such closed formula, for the case of Berge-acyclic queries.
Let $Q$ be a Berge-acyclic query with incidence graph $T$ (which is a
tree).  Choose an arbitrary relation $\texttt{ROOT}\in \bm R(Q)$ to
designate as the root of $T$, and for each other relation $R$, denote
by $Z_R \defeq \texttt{parent}(R)$, i.e.  its unique variable pointing
up the tree.  Denote by:
\begin{align}
PB(Q,\texttt{ROOT}) \defeq & N_{\texttt{ROOT}}\prod_{R\neq \texttt{ROOT}}f_1^{(R,Z_R)}
\label{eq:pb:one:component}
\end{align}
One can immediately check that the query answer on any database
instance consistent with the statistics satisfies
$|Q| \leq PB(Q,\texttt{ROOT})$.  A {\em cover} of $Q$ is set
$\bm W = \set{Q_1, Q_2, \ldots, Q_m}$, for some $m \geq 1$, where each
$Q_i$ is a connected subquery of $Q$, and each variable of $Q$ occurs
in at least one $Q_i$, and we denote by:
\begin{align}
  PB(\bm W) \defeq & \prod_{i=1,m} \min_{\texttt{ROOT}\in \bm R(Q_i)}PB(Q_i,\texttt{ROOT}_i)\label{eq:pb:all}
\end{align}
Since $|Q| \leq |Q_1|\cdot |Q_2| \cdots |Q_m|$, we also have
$|Q| \leq PB(\bm W)$.  We prove in Sec.~\ref{app:chain:bound}:
\begin{thm} \label{thm:polyamtroid:bound:sum:of:chains} The
  polymatroid bound of a Berge-acyclic query $Q$ is
  $PB(Q) \defeq \min_{\bm W} PB(\bm W)$, where $\bm W$ ranges over all
  covers.
\end{thm}

\begin{ex} \label{eq:chain:bound} Let $Q=R(X,Y),S(Y,Z),T(Z,U),K(U,V)$.
  Then $PB(Q,S) = f_1^{(R,Y)}N_S f_1^{(T,Z)}f_1^{(K,U)}$,
  $PB(\set{R,TK}) = N_R \cdot \min\left(N_T f^{(K,U)},
    f^{(T,U)}N_k\right)$, and $PB(\set{R,T,K}) = N_R N_T N_K$.
\end{ex}

If we restrict the formula to the AGM bound, i.e. all max degrees are
equal to the cardinalities, $f_1^{(R,X)}=N_R$, then
Eq.~\eqref{eq:pb:one:component} becomes $\prod_{R \in \bm R(Q)}N_R$,
while the polymatroid bound~\eqref{eq:pb:all} becomes
$\min_{\bm W} \prod_{R \in \bm W} N_R$, where $\bm W$ ranges over
integral covers of $Q$.  In particular, the AGM bound of a
Berge-acyclic query can be obtained by restricting to integral edge
covers, although this property fails for $\alpha$-acylic queries.  For
example, consider the query $R(X,Y),S(Y,Z),T(Z,X),K(X,Y,Z)$; when
$|R|=|S|=|T|=|K|$ then the AGM bound is obtained by the edge cover
$0,0,0,1$, but when $|R|=|S|=|T| \ll |K|$ one needs the fractional
cover $1/2, 1/2, 1/2, 0$.  Next, we prove next that the degree
sequence bound is always better.

\begin{lmm} \label{lemma:connection:to:pb} (1) For any choice of
  root relation, $\texttt{ROOT}\in \bm R(Q)$:
  $DSB(Q) \leq PB(Q,\texttt{ROOT})$.
  \newline (2) For any cover $Q_1, \ldots, Q_m$ of $Q$,
  $DSB(Q) \leq DSB(Q_1) \cdots DSB(Q_m)$
\end{lmm}

\begin{proof}
  (1) Referring to Algorithm~\ref{alg:bottom:up}, we prove by
  induction on the tree that, for all $R \neq \texttt{ROOT}$, and
  every index $i$,
  $w^{(R)}_i \leq \prod_{S \in \texttt{tree}(R)} f_1^{(S,Z_S)}$.  In
  other words, each element of the vector $\bm w^{(R)}$ is $\leq$ the
  product of all max degrees in the subtree rooted at $R$.  Assuming
  this holds for all children of $R$, consider the definition of
  $\bm w^{(R)}$ in Algorithm~\ref{alg:bottom:up}.  By induction
  hypothesis, for each vector $\bm a^{(X_p)}$ we have
  $a^{(X_p)}_{i_p}\leq \prod_{S \in \texttt{tree}(X_p)}f_1^{(S,Z_S)}$,
  a quantity that is independent of the index $i_p$, and therefore we
  obtain the following:
  \begin{align*}
  w^{(R)}_{i_1} = & \left(\bm C^{\bm f^{(R, \bm X_R)},B^{(R)}} \cdot\bm a^{(X_2)}\cdots\bm a^{(X_k)}\right)_{i_1}
  \leq \left(\sum_{i_2 i_3 \cdots i_k} C^{\bm f^{(R, \bm X_R)},B^{(R)}}_{i_1i_2i_3\cdots i_k}\right)\cdot \prod_{S \in \texttt{tree}(R),S\neq R} f_1^{(S,Z_S)}
  \end{align*}
  and we use the fact that
  $\sum_{i_2 i_3 \cdots i_k} C^{\bm f^{(R, \bm
      X_R)},B^{(R)}}_{i_1i_2\cdots i_k}\leq f^{(R,X_1)}_{i_1}$
  because, by Theorem~\ref{th:main:star}
  item~\ref{item:th:main:star:2},
  $\bm C^{\bm f^{(R, \bm X_R)},B^{(R)}}$ is consistent with the degree
  sequence $f^{(R,X_1)}_1$, and, finally,
  $f^{(R,X_1)}_{i_1}\leq f^{(R,X_1)}_1$. This completes the inductive
  proof.  The algorithm returns
  $C^{\bm f^{(\texttt{root}, \bm
      X_{\texttt{ROOT}})},B^{(\texttt{ROOT})}}\cdot\bm
  a^{(X_1)}\cdot\bm a^{(X_2)}\cdots\bm a^{(X_k)}\leq
  \texttt{SUM}(C^{\bm f^{(\texttt{root}, \bm
      X_{\texttt{ROOT}})},B^{(\texttt{ROOT})}})\cdot\prod_{R \neq
    \texttt{ROOT}}f_1^{(R,Z_R)}\leq |\texttt{ROOT}|\cdot\prod_{R \neq
    \texttt{ROOT}}f_1^{(R,Z_R)}$, which is $=PB(Q,\texttt{ROOT})$, as
  required.

  (2) We prove the statement only for $m=2$ (the general case is
  similar) and show that $DSB(Q) \leq DSB(Q_1) \cdot DSB(Q_2)$.  Since
  $DSB$ is the query answer on the worst case instance, we need to
  show that $|Q_1 \Join Q_2| \leq |Q_1| \cdot |Q_2|$.  This is not
  immediately obvious because the worst case instance may have
  negative multiplicities.  Let $X$ be the unique common variable of
  $Q_1, Q_2$, and let $\bm a$, $\bm b$ be the $X$-vectors representing
  the results of $Q_1$ and $Q_2$ respectively.  It follows from
  Theorem~\ref{th:main:star} item~\ref{item:th:main:star:4} that
  $\bm a, \bm b$ are non-negative, therefore,
  $|Q|=\sum_i a_i b_i \leq (\sum_i a_i)(\sum_i b_i) = |Q_1|\cdot
  |Q_2|$.
\end{proof}


Our discussion implies:

\begin{thm}
  Let $Q$ be a Berge-acyclic query.  We denote by
  $DSB(Q, \bm f, \bm B)$ the DSB computed on the statistics
  $\bm f \defeq (\bm f^{R,Z})_{R \in \bm R(Q), Z \in \bm X_R}$ and
  $\bm B \defeq (B^{(R)})_{R \in \bm R(Q)}$.  Then:
\begin{equation}
    |Q| \leq DSB(Q, \bm f, \bm 1) \leq DSB(Q, \bm f, \bm B) \leq DSB(Q, \bm f, \bm\infty) \leq PB(Q) \leq AGM(Q)
\end{equation}
where $|Q|$ is the answer to the query on an database instance
consistent with the given statistics.
\label{th:main:polymatroid}
\end{thm}

Recall that both AGM and PB bounds are defined over set semantics
only.  While the AGM bound is tight, the PB bound is known to not be
tight in general, and it is open whether it is tight for Berge-acyclic
queries.  Our degree sequence bound under either set or bag semantics
improves over PB and, in the case of bag semantics ($B=\infty$) DSB  is tight.

%% file: sec-5-compress.tex
\section{Functional Representation}

\label{sec:compress}

A degree sequence requires, in general, $\Omega(n)$ space, where
$n = \max_{X \in \bm X} n_X$ is the size of the largest domain, while
cardinality estimators require sublinear space and time.  However, a
degree sequence can be {\em represented} compactly, using a staircase
function as illustrated in Fig.~\ref{fig:degree-sequence-example}.  In
this section we show how the degree sequence bound, DSB, be
approximated in quasi-linear time in the size of the functional
representation. We call this approximate bound FDSB, show that
$\text{DSB}\leq \text{FDSB} \leq PB$, and show that the staircase
functions can be further compressed, allowing a tradeoff between the
memory size and computation time on one hand, and accuracy of the FDSB
on the other hand.  We restrict our discussion to $B^{(R)}= \infty$.

In this section we denote a vector element by $F(i)$ rather than
$F_i$.  For a non-decreasing vector $\bm F \in \R_+^{[n]}$, we denote
by $\bm F^{-1}: \R_+ \rightarrow \R_+$ any function satisfying the
following, for all $v$, $0 \leq v \leq F(n)$: if $F(i) < v$ then
$i < F^{-1}(v)$, and if $F(i) > v$ then $i > F^{-1}(v)$.  Such a
function always exists\footnote{E.g. define it as follows: if
  $\exists i$ s.t.  $F(i-1)<v<F(i)$ then set $F^{-1}(v)\defeq i-1/2$,
  otherwise set $F^{-1}(v) = i$ for some arbitrary $i$
  s.t. $F(i)=v$.}, but is not unique.
%
%
Then:

\begin{lmm} \label{lmm:compress} Let
  $\bm F_1 \in \R_+^{[n_1]}, \ldots, \bm F_d \in \R_+^{[n_d]}$ be
  non-decreasing vectors satisfying $F_1(0)=0$ and, for all $p=1,d$,
  $F_1(n_1) \leq F_p(n_p)$.  Let
  $a_1 \in \R_+^{[n_1]}, \ldots, a_d \in \R_+^{[n_d]}$ be
  non-increasing vectors. Denote by $\bm C, \bm w$ the following
  tensor and vector:
\begin{align}
    C_{i_1\cdots i_d} \defeq & \Delta_{i_1}\cdots\Delta_{i_d}\max(F_1(i_1),\ldots,F_d(i_d))\label{eq:compress:c}\\
    w(i_1) \defeq & \sum_{i_2=1}^{n_2}\ldots \sum_{i_d=1}^{n_d}C_{i_1\cdots i_d}\prod_{p\in[2,d]}a_p(i_p)\label{eq:compress:w}
\end{align}
Then the following inequalities hold:
\begin{align}
    w(i_1) \geq& \left(\Delta_{i_1}F_1(i_1)\right)\prod_{p\in[2,d]}a_p\left(\floor{F_p^{-1}(F_1(i_1))}+1\right) \label{eq:w:efficient:1}\\ 
    w(i_1) \leq& \left(\Delta_{i_1}F_1(i_1)\right)\prod_{p\in[2,d]}a_p\left(\ceil{F_p^{-1}(F_1(i_1-1))}\right)\label{eq:w:efficient:2}
\end{align}
\end{lmm}

We give the proof in Appendix~\ref{app:lmm:compress}.  The lemma
implies that, in Algorithm~\ref{alg:bottom:up}, we can use
inequality~\eqref{eq:w:efficient:2} to upper bound the computation $w^{(R)} = \bm C \cdot \bm a^{(X_2)}\cdots \bm a^{(X_k)}$.  Indeed, in
that case each $F_p(r) \defeq \sum_{i=1,r}f_p(r)$ is the cdf of a
degree sequence $f_p$, hence $F_p(0)=0$ and $F_p(n_p)=$ the
cardinality of $R$, while the tensor $\bm C$ is described in
item~\ref{item:th:main:star:5} of Theorem~\ref{th:main:star}, hence
the assumptions of the lemma hold.

%
%
We say that a vector $\bm f\in \R_+^n$ is {\em represented} by a
function $\hat f : \R_+\rightarrow \R_+$ if $f(i)=\hat f(i)$ for all
$i=1,n$.  A function $\hat f$ is a {\em staircase function with $s$
  steps}, in short an {\em $s$-staircase}, if there exists dividers
$m_0 \defeq 0 < m_1 < \cdots < m_s \defeq n$ such that $\hat f(x)$ is
a nonnegative constant on each interval
$\setof{x}{m_{q-1} < x \leq m_q}$, $q=1,s$. The sum or product of an
$s_1$-staircase with an $s_2$-staircase is an
$(s_1+s_2)$-staircase. We denote the summation of a staircase
$\hat f(x)$ as $\hat F(x)=\int_0^x \hat f(t)dt$ which is then an
increasing piecewise-linear function. Its standard inverse
$\hat F^{-1}:\R_+\rightarrow \R_+$ is also increasing and
piecewise-linear. If $\hat F$ represents the vector $F$, then
$\hat F^{-1}$ is an inverse $F^{-1}$ of that vector (as discussed
above).

\begin{algorithm}[t]
  \caption{$FDSB(Q,\texttt{ROOT})$}
\label{alg:bottom:up:f}
\begin{algorithmic}
\FOR{each variable  $X \in \bm X$ and non-root relation $R \in \bm R$, $R\neq \texttt{root}$, in bottom-up order}{
   \STATE $\bm{\hat a}^{(X)} \defeq \bigotimes_{R \in \text{children}(X)} \bm{\hat w}^{(R)}$
   \STATE $\forall i_1:\ \ \hat w^{(R)}(i_1) \defeq \left(\hat f^{(R,X_1)}(i_1)\right)\prod_{p\in[2,d]}a^{(X_p)}\left(\max(1, (\hat F^{(R,X_p)})^{-1}(\hat F^{(R,X_1)}(i_1-1)))\right)$
 }
\ENDFOR
\RETURN $\sum_{i=1,|\texttt{ROOT}|}\prod_{p=1,k}a^{(X_p)}(\max(1,(F^{\texttt{ROOT},X_p})^{-1}(i-1)))$
\end{algorithmic}
\end{algorithm}

Fix a Berge-acyclic query $Q$, and let each degree sequence
$\bm f^{(R,Z)}$ be represented by some $s_{R,Z}$-staircase
$\hat{f}^{R,Z}$, and we denote by $\hat{F}^{(R,Z)}$ its summation.
Fix any relation $\texttt{ROOT} \in \bm R(Q)$ to designated as root.
The {\em Functional Degree Sequence Bound at $\texttt{ROOT}$},
$FDSB(Q,\texttt{ROOT})$, is the value returned by
Algorithm~\ref{alg:bottom:up:f}.  This algorithm is identical to
Algorithm~\ref{alg:bottom:up}, except that it replaces both
$\bm w^{(R)}$ with a functional upper bound justified by the
inequality~\ref{eq:w:efficient:2} of Lemma~\ref{lmm:compress}, and
similarly for the returned result.  All functions $\bm{\hat a}^{(X)}$
and $\bm{\hat w}^{(R)}$ are staircase functions, and can be computed
in linear time, plus a logarithmic time need for a binary search to
lookup a segment in a staircase.  Using this, we prove the following
in Appendix~\ref{app:thm:compress}:
%


\begin{thm} \label{thm:compress} (1)
  $FDSB(Q, \texttt{ROOT})\geq DSB(Q)$. (2) $FDSB(Q,\texttt{ROOT})$ can
  be computed in time
  $T_{\text{FDSB}}\defeq \tilde O(m\cdot \sum_{R,
    Z}(\text{arity}(R)\cdot s_{R,Z}))$, where $\tilde O$ hides a
  logarithmic term, and $m = |\bm R(Q)|$ is the number of relations in
  $Q$.
\end{thm}

The theorem says that $FDSB(Q,\texttt{ROOT})$ is still an upper bound
on $|Q|$, and can be computed in quasi-linear time in the size of the
functional representations of the degree sequences.  Next, we check if
$FDSB$ is below the polymatroid bound.  Consider the computation
of $\hat w^{(R)}(i_1)$ by the algorithm.  On one hand
$\hat f^{(R,X_1)}(i_1) \leq \hat f^{(R,X_1)}(1)$; on the other hand
$a^{(X_p)}(\max(1,\ldots)) \leq a^{(X_p)}(1)$.  This allows us to
prove (inductively on the tree,  in  ~\ref{app:lmm:compress:polymatroid}):
\begin{lmm}\label{lmm:compress:polymatroid}
  $FDSB(Q, \texttt{ROOT}) \leq PB(Q, \texttt{ROOT})$, where $PB$ is
  defined in~\eqref{eq:pb:one:component}.
\end{lmm}

When we proved $DSB\leq PB$ in Lemma~\ref{lemma:connection:to:pb}, we
used two properties of $DSB$: $DSB(Q, \texttt{ROOT})$ is independent
of the choice of $\texttt{ROOT}$, and
$DSB(Q_1 \Join \cdots \Join Q_m) \leq DSB(Q_1) \cdots DSB(Q_m)$, for
any cover $\bm W = \set{Q_1, \ldots, Q_m}$.  Both hold because
$DSB(Q)$ is standard query evaluation: it is independent of the query
plan (i.e. choice of $\texttt{ROOT}$) and it can only increase if we
remove join conditions.  But $FDSB$ is no longer standard query
evaluation and these properties may fail.  For that reason we
introduce a stronger functional degree sequence bound:
%
%
\begin{align}
    FDSB(Q) = \min_{\bm W}\prod_{i=1,m}\min_{ROOT\in\bm R(Q)}FDSB(Q_i, ROOT)
\end{align}
where $\bm W$ range over the covers of $Q$.  We prove in Appendix~\ref{app:thm:compress:dynamic}:


\begin{thm} \label{thm:compress:covers} $FDSB(Q)$ can be computed in
  time $O(2^{m}\cdot(2^m+ m\cdot T_{FDSB}))$ (where $T_{FDSB}$ is defined in
  Theorem~\ref{thm:compress}).
\end{thm}


%
Mirroring our results from Theorem \ref{th:main:polymatroid}, we prove
the following in  Appendix~\ref{app:thm:compress:polymatroid}:


\begin{thm} \label{th:approximate:agm:polymatroid}
Suppose $Q$ is a Berge-acyclic query.  Then the following hold:
\begin{align}
  && |Q| \leq  & FDSB(Q) \leq PB(Q) \leq AGM(Q)
\end{align}
\end{thm}

Together, Theorems~\ref{thm:compress:covers}
and~\ref{th:approximate:agm:polymatroid} imply that we can compute in
quasi-linear time in the size of the representation an upper bound to
the query $Q$ that is guaranteed to improve over the polymatroid
bound.  In practice, we expect this bound to be significantly lower
than the polymatroid bound, because it accounts for the entire degree
sequence $\bm f$, not just $f_1$.

Finally, we show that one can tradeoff the size of the representation
for accuracy, by simply choosing more coarse staircase approximations
of the degree sequences.  They only need to be non-increasing, and lie
above the true degree sequences.


\begin{thm} 
  Fix a query $Q$, let $\bm f^{(R,Z)}, B^{(R)}$ be statistics as in
  Problem~\ref{prob:main}, and let $U$ be the cardinality bound
  defined by~\eqref{eq:prob:main}.  Let
  $\hat{\bm f}^{(R,Z)}, \hat{B}^{(R)}$ be a new set of statistics, and
  $\hat U$ the resulting cardinality bound.  If
  $\bm f^{(R,\bm X_R)} \leq \hat{\bm f}^{(R,\bm X_R)}$ and
  $B^{(R)} \leq \hat{B}^{(R)}$ for all $R$, $Z \in X_R$, then
  $U \leq \hat U$.
\end{thm}

\begin{proof}
  The proof follows immediately from the observation that the set of
  feasible solutions can only increase (see
  Def.~\ref{def:consistent}):
  $\calM^+_{\bm f^{(R, \bm X_R)},B^{(R)}}\subseteq \calM^+_{\hat{\bm
      f}^{(R, \bm X_R)},\hat{B}^{(R)}}$.
\end{proof}


%% file: sec-6-discussion.tex
\section{Conclusions}

\label{sec:discussion}

We have described the {\em degree sequence bound} of a conjunctive
query, which is an upper bound on the size of its answer, given in
terms of the degree sequences of all its attributes.  Our results
apply to Berge-acyclic queries, and strictly improve over previously
known AGM and polymatroid
bounds~\cite{DBLP:conf/focs/AtseriasGM08,DBLP:conf/pods/Khamis0S17}. On
one hand, our results represent a significant extension, because they
account for the full degree sequences rather than just cardinalities
or just the maximum degrees. On the other hand, they apply only to a
restricted class of acyclic queries, although, we argue, this class is
the most important for practial applications.  While the full degree
sequence can be as large as the entire data, we also described how to
approximate the cardinality bound very efficiently, using compressed
degree sequences.  Finally, we have argued for using the max tuple
multiplicity for each relation, which can significantly improve the
accuracy of the cardinality bound.


%% file: appendix.tex
\onecolumn

\section{Appendix}

\label{sec:appendix}

\subsection{Discussion  the Framework}

\label{app:background}


Throughout this paper we restricted the discussion to Berge-acyclic
queries.  In this section we briefly review their standard definition,
and comment on their utility in practice.

We begin by reviewing the definition of a Berge-cycle,
following~\cite{DBLP:journals/jacm/Fagin83}.  Fix an arbitrary
conjunctive query $Q$, which we view as a hypergraph, and consider its
incidence graph, defined as the following bipartite graph:
$T \defeq (\bm R\cup \bm X, E \defeq \setof{(R,Z)}{Z \in \bm X_R})$.
A \textit{Berge-Cycle} is a sequence
$(R_1, X_1, R_2, X_2, \ldots, R_m, X_m, R_{m+1})$ such that,
\begin{enumerate}
    \item $R_i \,\,\forall \,\, 1\leq i\leq m$ is a unique relation in $\bm R$.
    \item $R_1=R_{m+1}$
    \item $X_i \,\,\forall \,\, 1\leq i\leq m$ is a unique variable in $\bm X$.
    \item $m\geq 2$
    \item $X_i\in \bm X_{R_i}$ and $X_i\in \bm X_{R_{i+1}} \,\,\forall\,\, 1\leq i\leq m$
\end{enumerate}
It follows immediately that a Berge-cycle is a cycle in the incidence
graph, and vice versa.  A \textit{Berge-Acyclic} query is one that
does not contain a Berge-Cycle, or, equivalently, one whose incidence
graph is a tree.  This coincides with our definition in
Sec.~\ref{sec:problem}.

Berge-acyclic queries capture many traditional flavors of queries
found in practice, such as chain, star, and snowflake queries. As an
illustration of this, we note that many modern cardinality estimation
benchmarks include exclusively Berge-Acyclic queries, for example the
JOB benchmark \cite{DBLP:journals/pvldb/LeisGMBK015} and the STATS-CEB
benchmark~\cite{han2021cardinality}. Additionally, many state of the
art approaches to cardinality estimation have this limitation as well.
For example, NeuroCard, FLAT, and BayesCard cannot produce estimates
for cyclic queries~\cite{yang2020neurocard, zhu2020flat,
  wu2020bayescard}.

While rare, cyclic queries do occur in practice, and, in that case,
the degree sequence bound is still useful, as follows.  Consider all
possible spanning trees of the incidence graph.  The query defined by
each such spanning tree must produce a larger output than the original
query. This can be seen by viewing joins as a set of filters on the
Cartesian product of the input tables. A spanning tree operates on the
same Cartesian product, but it applies a subset of the cyclic query's
filters, and, therefore, it must produce at least as many rows. This
means that we can produce a bound on cyclic queries by taking the
minimum of the degree sequence bound for each spanning tree.  In
general, this bound is not tight, but is still an upper bound, and
still useful.  In contrast, cardinality estimation frameworks like
NeuroCard, FLAT, BayesCard aim at {\em estimating} the cardinality,
but an estimate on the cardinality of a spanning tree is useless for
estimating the cardinality of the cyclic query.  In other words, while
our degree sequence bound makes similar restrictions on the queries as
other state-of-the-art cardinality estimators, the latter cannot be
used on cyclic queries, while ours is still useful.


\input{theorem-3-2.tex}

\subsection{Proof of Theorem \ref{th:general}}

\label{app:th:general}

We first illustrate Algorithm~\ref{alg:bottom:up} with an example.

\begin{ex}
  Let:
  \begin{align*}
 Q(X,Y,Z,U) := R(X,Y)\Join S(Y,Z) \Join T(Z,U) \Join K(U)
  \end{align*}
  Assume we are given degree sequences for each attribute of each
  relation, and compute the corresponding worst case tensors (one for
  each relation):
  $\bm C^{(R)}, \bm C^{(S)}, \bm C^{(T)}, \bm C^{(K)}$: the first
  three are matrices, the last is a vector. The degree sequence bound
  is:
  \begin{align*}
    \texttt{SUM}_{XYZU} \left(\bm C^{(R)}\otimes \bm C^{(S)} \otimes \bm C^{(T)} \otimes \bm C^{(K)}\right)=
    & \sum_{xyzu} C^{(R)}_{xy} C^{(S)}_{yz} C^{(T)}_{zu} C^{(K)}_u
  \end{align*}
  Algorithm~\ref{alg:bottom:up} allows us to compute this expression
  as follows.  We choose any relation to designate as root: we will chose (arbitrarily) $S$ as the root, associate a vector
  $\bm a^{(X)}, \bm a^{(Y)}, \bm a^{(Z)}, \bm a^{(U)}$ to each
  variable, and a vector
  $\bm w^{(R)}, \bm w^{(S)}, \bm w^{(T)}, \bm w^{(K)}$ to each
  relation.  Then, we compute the degree sequence bound as follows,
  where we show the type of each vector, assuming that the domains of
  the variables $X, Y, Z, U$ are $[n_X], [n_Y], [n_Z], [n_U]$
  respectively:
  \begin{align*}
    a^{(U)} := &  C^{(K)} & \in & \R_+^{[n_U]}\\
    w^{(T)} := & C^{(T)} \cdot a^{(U)} & \in & \R_+^{[n_Z]} \\
    a^{(Z)}:= & w^{(T)}  & \in & \R_+^{[n_Z]}\\
    a^{(X)} := & \bm 1  & \in & \R_+^{[n_X]} && =\mbox{the column vector $(1,1,\cdots,1)^T$, since $X$ has no children}\\
    w^{(R)} := & C^{(R)} \cdot a^{(X)}  & \in & \R_+^{[n_Y]} \\
    a^{(Y)} := & w^{(R)}  & \in & \R_+^{[n_Y]}&& \\
    \texttt{return} & \left( C^{(S)} \cdot a^{(Y)} \cdot a^{(Z)}\right)
  \end{align*}
\end{ex}


We now turn to the proof of Theorem \ref{th:general}, and begin by
making two observations.  Let $T$ be the incidence graph of the query
(Sec.~\ref{sec:problem}), and recall that in this paper we always
assume that $T$ is an undirected tree.  First, consider any relation
$R \in \bm R$, and assume w.l.o.g. that its variables are
$\bm X_R = \set{X_1, \ldots, X_k}$.  By removing the node $R$ from the
tree, we partition the tree into $k$ disjoint connected components,
each containing one of the variables $X_i$.  Then we can group the
products in Eq.~\eqref{eq:prob:main} into $k+1$ factors::
\begin{align*}
  \text{Eq.~\eqref{eq:prob:main}} = 
  &
    \texttt{SUM}_{\bm X}\left( \left(\bm M^{(R)} \circ \bm \sigma^{(R)}\right)
    \otimes \bm E_1 \otimes \cdots \otimes \bm E_k\right)
\end{align*}
where $\bm E_i$ is the tensor expression containing all relations in
the $i$'th component.  Let the variables of $\bm E_i$ be
$\set{X_i} \cup \bm Y_i$, therefore
$\bm X = \bm X_R \cup \bm Y_1 \cup \cdots \cup \bm Y_k$, and define
$\bm b_i \defeq \texttt{SUM}_{\bm Y_i}(\bm E_i)$, i.e. we sum out all
variables other than $X_i$ to obtain a vector $\bm b_i$. Using
Eq.~\eqref{eq:push:aggregates:down} we derive:
\begin{align}
  \text{Eq.~\eqref{eq:prob:main}}=
  &
    \texttt{SUM}_{\bm X_R}\left( \left(\bm M^{(R)} \circ \bm \sigma^{(R)}\right)  \otimes \left(\texttt{SUM}_{\bm Y_1}\bm E_1\right) \otimes \cdots  \otimes \left(\texttt{SUM}_{\bm Y_k}\bm E_k\right)\right)\\
    &=
    \texttt{SUM}_{\bm X_R}\left( \left(\bm M^{(R)} \circ \bm \sigma^{(R)}\right)
    \otimes \bm b_1 \otimes \cdots \otimes \bm b_k\right) \label{eq:partitioned}
\end{align}

Second, we will compute the expression \eqref{eq:prob:main} bottom-up
on the tree, following Algorithm~\ref{alg:bottom:up}.  Recall that the
algorithm works by choosing arbitrarily a variable $X_0 \in \bm X$ to
designate as root of the tree $T$, then orienting the tree such that
all edges point away from the root.  Each relation $R$ has a unique
parent variable, call it $X_1$, and zero or more children, call them
$X_2, \ldots, X_k$.  Algorithm~\ref{alg:bottom:up} computes the
expression \eqref{eq:prob:main} bottom up, by computing a $Z$-vector
$\bm a^{(Z)}$ for each variable $Z$, and an $X_1$-vector $w^{(R)}$ for
each relation $R$ with parent $X_1$.  Then, expression
\eqref{eq:prob:main} is equal to $\texttt{SUM}_{X_0}(\bm a^{(X_0)})$.

We can now describe the proof formally. Call a subset
$\bm S \subseteq \bm X \cup \bm R$ {\em closed} if, for any node in
$S$, all its children (in the directed tree $T$) are also in $S$.
\begin{defn}
  A set $\bm S \subseteq \bm X \cup \bm R$ is {\em good} if it is
  closed and satisfies the following two conditions:
  \begin{enumerate}
  \item The following inequality holds:
    \begin{align}
      \mbox{Eq.~\eqref{eq:prob:main}}=  \texttt{SUM}_{\bm X}\left( \bigotimes_{R \in \bm R}(\bm  M^{(R)} \circ \bm \sigma^{(R)})\right)\leq 
      & \texttt{SUM}_{\bm X}\left(\left( \bigotimes_{R \in \bm S}\bm C^{\bm f^{(R, \bm X_R)},B^{(R)}}\right)
        \otimes \left( \bigotimes_{R \in \bm R-\bm S}(\bm  M^{(R)}
        \circ \bm \sigma^{(R)})\right)\right) \label{eq:good}
    \end{align}
    In other words, the value of the expression \eqref{eq:prob:main}
    does not decrease if we substitute each tensor
    $\left(\bm M^{(R)} \circ \bm \sigma^{(R)}\right)$ for
    $R \in \bm S$, with $\bm C^{\bm f^{(R, \bm X_R)},B^{(R)}}$.
  \item Consider running Algorithm~\ref{alg:bottom:up} to compute the
    RHS of Eq.~\eqref{eq:good}.  Then for every variable
    $Z \in \bm S$, the vector $\bm a^{(Z)}$ is non-negative and
    non-increasing.
  \end{enumerate}
\end{defn}
We will prove that the entire set $\bm X \cup \bm R$ is good, and for
that we prove the following claim.

\begin{claim} \label{claim:good}
  Suppose $\bm S$ is good, and let $X \not\in \bm S$ (or
  $R \not\in \bm S$) be any variable (or a relation name) that is not
  in $\bm S$ such that all its children are in $\bm S$.  Then
  $\bm S \cup \set{X}$ (or $\bm S \cup \set{R}$ respectively) is also
  good.
\end{claim}

The claim immediately implies that $\bm X \cup \bm R$ is good, because
we start with $\bm S = \emptyset$, which is trivially good, then add
one by one each variable and each relation name to $\bm S$, until
$\bm S = \bm X \cup \bm R$.  It remains to prove the claim.

\begin{proof} (Of Claim~\ref{claim:good}) We consider two cases,
  depending on whether the node that we want to add to $\bm S$ is a
  relation name or a variable.

  {\bf Case 1}.  There exists a relation $R \not\in \bm S$ that is not
  yet in $\bm S$, such that all its children $X_2, \ldots, X_k$ are in
  $\bm S$.  We partition the expression \eqref{eq:prob:main} as before:
  \begin{align*}
    \text{Eq.~\eqref{eq:prob:main}}=
    &
      \texttt{SUM}_{\bm X} \left(\left(\bm M^{(R)} \circ \bm \sigma^{(R)}\right)
      \otimes \bm E_1 \otimes \cdots \otimes \bm E_k \right)
  \end{align*}
  For each $i=2,k$ let $\hat{\bm E}_i$ be the expression obtained by
  substituting every tensor
  $\left(\bm M^{(R')} \circ \bm \sigma^{(R')}\right)$ occurring in
  $\bm E_i$ with $\bm C^{\bm f^{(R', \bm X_{R'})},B^{(R')}}$.  Since
  all children $X_2, \ldots, X_k$ of $R$ are in are $\bm S$, it
  follows that reach relation name that occurs in
  $\bm E_2, \ldots, \bm E_k$ is in $\bm S$. Since $\bm S$ is good,
  this implies:
  \begin{align*}
    \texttt{SUM}_{\bm X} \left(\left(\bm M^{(R)} \circ \bm \sigma^{(R)}\right)
    \otimes \bm E_1 \otimes \cdots \otimes \bm E_k \right) \leq
    &       \texttt{SUM}_{\bm X} \left(\left(\bm M^{(R)} \circ \bm \sigma^{(R)}\right)
      \otimes \bm E_1 \otimes \hat{\bm E}_2 \otimes \cdots \otimes \hat{\bm E}_k \right)
  \end{align*}
  We now use the argument in Eq.~\eqref{eq:partitioned}:
  \begin{align*}
        \texttt{SUM}_{\bm X}& \left(\left(\bm M^{(R)} \circ \bm \sigma^{(R)}\right)
      \otimes \bm E_1 \otimes \hat{\bm E}_2 \otimes \cdots \otimes \hat{\bm E}_k \right)
= \\    & \texttt{SUM}_{\bm X_R} \left(\left(\bm M^{(R)} \circ \bm \sigma^{(R)}\right)
      \otimes \left(\texttt{SUM}_{\bm Y_1}\bm E_1\right) \otimes \left(\texttt{SUM}_{Y_2}\hat{\bm E}_2\right) \otimes \cdots \otimes \left(\texttt{SUM}_{Y_k}\hat{\bm E}_k\right)\right)
  \end{align*}
  As before we denote by $\bm b_1 = \texttt{SUM}_{\bm Y_1}\bm E_1$,
  but notice that, for all $i=2,k$, the expression
  $\texttt{SUM}_{Y_i}\hat{\bm E}_i$ is precisely the vector
  $\bm a^{(X_i)}$ used in the bottom-up computation of
  Algorithm~\ref{alg:bottom:up}:
  \begin{align*}
    \texttt{SUM}_{\bm X_R}& \left(\left(\bm M^{(R)} \circ \bm \sigma^{(R)}\right)
      \otimes \left(\texttt{SUM}_{\bm Y_1}\bm E_1\right) \otimes \left(\texttt{SUM}_{Y_2}\hat{\bm E}_2\right) \otimes \cdots \otimes \left(\texttt{SUM}_{Y_k}\hat{\bm E}_k\right)\right)
 = \\& \texttt{SUM}_{\bm X_R} \left(\left(\bm M^{(R)} \circ \bm \sigma^{(R)}\right) \otimes \bm b_1 \otimes \bm a^{(X_2)} \otimes \cdots \otimes \bm a^{(X_k)}\right)
  \end{align*}
  Since $\bm S$ is good, it follows that the vectors
  $\bm a^{(X_2)}, \ldots, \bm a^{(X_k)}$ are already sorted.  We sort
  $\bm b_1$ as well, by applying some permutation $\tau$, and obtain:
  \begin{align*}
    \texttt{SUM}_{\bm X_R}& \left(\left(\bm M^{(R)} \circ \bm \sigma^{(R)}\right) \otimes \bm b_1 \otimes \bm a^{(X_2)} \otimes \cdots \otimes \bm a^{(X_k)}\right)
=\\&    \texttt{SUM}_{\bm X_R} \left(\left(\bm M^{(R)} \circ \bm \sigma^{(R)}\circ (\tau,\texttt{id},\ldots,\texttt{id})\right)   \otimes (\bm b_1\circ \tau) \otimes \bm a^{(X_2)} \otimes \cdots \otimes \bm a^{(X_k)}\right)
  \end{align*}
  Since all vectors
  $\bm b_1 \circ \tau, \bm a^{(X_2)}, \ldots, \bm a^{(X_k)}$ are
  sorted, we can use item~\ref{item:th:main:star:2} of
  Theorem~\ref{th:main:star} to derive:
  \begin{align*}
    \texttt{SUM}_{\bm X_R}& \left(\left(\bm M^{(R)} \circ \bm \sigma^{(R)}\circ (\tau,\texttt{id},\ldots,\texttt{id})\right)   \otimes (\bm b_1\circ \tau) \otimes \bm a^{(X_2)} \otimes \cdots \otimes \bm a^{(X_k)}\right)
\leq\\ & \texttt{SUM}_{\bm X_R} \left(\bm C^{\bm f^{(R, \bm X_R)},B^{(R)}} \otimes (\bm b_1\circ \tau) \otimes \bm a^{(X_2)} \otimes \cdots \otimes \bm a^{(X_k)}\right)
  \end{align*}
  Combining all inequalities above implies:
  \begin{align*}
    \text{Eq.~\eqref{eq:prob:main}}\leq & \texttt{SUM}_{\bm X_R} \left(\bm C^{\bm f^{(R, \bm X_R)},B^{(R)}} \otimes (\bm b_1\circ \tau) \otimes \bm a^{(X_2)} \otimes \cdots \otimes \bm a^{(X_k)}\right)
  \end{align*}
  which proves that $\bm S \cup \set{R}$ is good. This completes the
  proof of Case 1.

  {\bf Case 2} Suppose there exists a variable $X \not\in \bm S$ such
  that all its children relations $R_1, R_2, \ldots, R_m$ are in
  $\bm S$.  Let $R_i$ be one of its children.  Since $R_i$ is in
  $\bm S$, by induction hypothesis its tensor in the RHS of
  Eq.~\eqref{eq:good} is the pessimistic tensor
  $\bm C \defeq \bm C^{\bm f^{(R_i, \bm X_{R_i})},B^{(R_i)}}$.  Let
  $X_2, \ldots, X_K$ be the other variables of $R_i$ (i.e. other than
  its parent variable $X$).  Since $\bm S$ is closed, we have
  $X_2, \ldots, X_k \in \bm S$, and, since $\bm S$ is good, it follows
  that the vectors $\bm a^{(X_2)}, \ldots, \bm a^{(X_k)}$ are
  non-negative, and non-increasing.  Algorithm~\ref{alg:bottom:up}
  computes $\bm w^{(R_i)}$, as follows:
  \begin{align*}
    \bm w^{(R_i)} := & \bm C \cdot \bm a^{(X_2)}\cdots \bm a^{(X_k)}
  \end{align*}
  Item~\ref{item:th:main:star:2} of Theorem~\ref{th:main:star} implies
  that $\bm w^{(R_i)}$ is also non-negative and non-decreasing.  Next,
  the algorithm computes $\bm a^{(X)}$ as:
  \begin{align*}
    \bm a^{(X)} := & \bm w^{(R_1)} \otimes \cdots \otimes \bm w^{(R_m)}
  \end{align*}
  where the tensor product is an element-wise product.  Since each
  $\bm w^{(R_i)}$ is non-negative and non-decreasing, their product
  $\bm a^{(X)}$ will also be non-negative and non-decreasing, as
  required.  Thus, $\bm S \cup \set{X}$ is also good, which completes
  the proof of case 2.
\end{proof}


\input{chain-bound}


\subsection{Proof of Lemma~\ref{lmm:compress}}

\label{app:lmm:compress}

We start by proving:

\begin{claim} \label{claim:compress} For all $p=2,d$, the following hold:
  \begin{itemize}
  \item If $F_1(i_1-1)\geq F_p(i_p)$ then $C_{i_1\cdots i_p}=0$.
  \item If $F_p(i_p-1)\geq F_1(i_1)$ then $C_{i_1\cdots i_p}=0$.
  \end{itemize}
\end{claim}

\begin{proof}
  We prove the first implication only; the second implication is
  similar and omitted.

  Assume that $F_1(i_1-1) \geq F_p(i_p)$.  We also have
  $F_1(i_1)\geq F_p(i_p)$ (because $F_1$ is non-decreasing), and
  therefore the following hold:
  \begin{align*}
    \min(F_1(i_1-1),F_2(i_2),\ldots,F_d(i_d))= & \min(F_2(i_2),\ldots,F_d(i_d))\\
    \min(F_1(i_1),F_2(i_2),\ldots,F_d(i_d))= & \min(F_2(i_2),\ldots,F_d(i_d))
  \end{align*}
  This implies:
  \begin{align}
    \Delta_1 \min(F_1(i_1), \ldots, F_d(i_d))= & 0 \label{eq:here:1}
  \end{align}
  On the other hand, $F_1(i_1) \geq F_p(i_p-1)$  (because $F_p$ is
  non-decreasing), and therefore we can repeat the argument above for
  $i_p-1$ instead of $i_p$, and obtain:
  \begin{align}
    \Delta_1 \min(F_1(i_1), \ldots, F_p(i_p-1), \ldots F_d(i_d)) = &0\label{eq:here:2}
  \end{align}
  Combining~\eqref{eq:here:1} and~\eqref{eq:here:2} we obtain:
  \begin{align*}
    \Delta_p\Delta_1 \min(F_1(i_1), \ldots, F_d(i_d)) = & 0
  \end{align*}
  Since this holds for all indices $i_q$, $q \neq 1, q \neq p$, it
  follows that:
  \begin{align*}
    C_{i_1\cdots i_d} = & \Delta_1 \cdots \Delta_d \min(F_1(i_1),\ldots, F_d(i_d))=0
  \end{align*}
  The second implication of the claim is proven similarly.
\end{proof}

We now prove Lemma~\ref{lmm:compress}.  Assume that
$C_{i_1\cdots i_d} \neq 0$.  The claim implies the following two
conditions:
\begin{align*}
  F_p(i_p)>& F_1(i_1-1) \\
  F_p(i_p-1) < & F_1(i_1)
\end{align*}
By our definition of $F^{-1}$ in Sec.~\ref{sec:compress}, this
implies:
\begin{align*}
  i_p > & F_p^{-1}(F_1(i_1-1))\defeq j \\
  i_p < & F_p^{-1}(F_1(i_1))+1 \defeq k
\end{align*}
(Notice that $j,k$ may be real values.)  This means that in the sum
below we can restrict $i_p$ to range between $j$ and $k$ only:
  \begin{align}
    \sum_{i_p=1,n_p}& C_{i_1\cdots i_d}\cdot a_p(i_p) = \sum_{i_p=\ceil{j},\floor{k}}C_{i_1\cdots i_d}\cdot a_p(i_p) \label{eq:sum:restricted}
  \end{align}
Since $\bm a_p$ is non-increasing we have
  \begin{align*}
 \sum_{i_p=\ceil{j},\floor{k}}C_{i_1\cdots i_d}\cdot a_p(i_p)\geq & \left(\sum_{i_p=\ceil{j},\floor{k}}C_{i_1\cdots i_d}\right)\cdot a_p(\floor{k}) =\left(\sum_{i_p=1,n_p}C_{i_1\cdots i_d}\right)\cdot a_p(\floor{F_p^{-1}(F_1(i_1))}+1)
  \end{align*}
and similarly:
  \begin{align*}
 \sum_{i_p=\ceil{j},\floor{k}}C_{i_1\cdots i_d}\cdot a_p(i_p)
\leq & \left(\sum_{i_p=\ceil{j},\floor{k}}C_{i_1\cdots i_d}\right)\cdot a_p(\ceil{j}) =\left(\sum_{i_p=1,n_p}C_{i_1\cdots i_d}\right)\cdot a_p(\ceil{F_p^{-1}(F_1(i_1-1))})
  \end{align*}
  Inequalities~\eqref{eq:w:efficient:1} and~\eqref{eq:w:efficient:2}
  follow by repeating the argument above for all $p=2,d$, then
  observing that
  \begin{align*}
    \sum_{i_2=1,n_2}\cdots \sum_{i_d=1,n_d}C_{i_1i_2\cdots i_d}=&\Delta_1 \min(F_1(i_1),F_2(n_2),\ldots,F_d(n_d))=\Delta_1 F_1(i_1)
  \end{align*}
  by Eq.~\eqref{eq:delta:sigma} and
  $F_1(i_1) \leq F_1(n_1)\leq F_p(n_p)$.

\subsection{Proof of Theorem~\ref{thm:compress}}
\label{app:thm:compress}
%
%

We begin by proving item (1) of the theorem: $FDSB(Q,\bm f) \geq DSB(Q, \bm f)$. To do this, we first prove the following lemma.
\begin{lmm}
\begin{align}
    DSB(Q, \bm f) \leq FDSB(Q, \bm f, \texttt{ROOT})
\end{align}
\end{lmm}
\begin{proof}
Let $Q^*$ be an alteration of the query $Q$ where everything is identical except for the addition of a variable $X_0$ which is only present in relation $\texttt{ROOT}$. Additionally, we set the degree sequence equal to $\bm 1$, i.e. $f^{ROOT, X_0}(i)=1\,\,\forall\,\,i\in[1,|ROOT|]$. We claim that $DSB(Q^*, \bm f, \infty) = DSB(Q, \bm f, \infty)$. To see this, consider the final summation for $DSB(Q^*, \bm f, \infty)$, $DSB(Q, \bm f, \infty)$, and $FDSB(Q, \bm f, \texttt{ROOT})$ in algorithm 2,
\begin{align*}
    DSB(Q^*, \bm f) &= \sum_{i_1=1,D_1}\ldots\sum_{i_k=1,D_k}\sum_{i_0=1,|ROOT|}C_{i_1,\ldots,i_k, i_0}\cdot a^{(X_0)}(i_0) \cdot \prod_{p=1,k}a^{(X_p)}_{i_p}\\
    DSB(Q, \bm f) &= \sum_{i_1=1,D_1}\ldots\sum_{i_k=1,D_k}\sum_{i_0=1,|ROOT|}C_{i_1,\ldots,i_k, i_0}\cdot a^{(X_0)}(i_0) \cdot \prod_{p=1,k}a^{(X_p)}_{i_p}\\
    FDSB(Q, \bm f, \texttt{ROOT}) &= \sum_{i=1,|\texttt{ROOT}|}\prod_{p=1,k}a_p(max(1,F^{-1}_p(i)))
\end{align*}
Because $X_0$ is not shared by any other relations, we have $a^{(X_0)}(i_0)=1$. Further, we can break $\sum_{i_0=1,|ROOT|}C_{i_1,\ldots,i_k,i_0}$ into $\sum_{i_0=1,|ROOT|}\Delta_{i_0}\Delta_{i_1}\ldots\Delta_{i_k}V_{i_1,\ldots,i_k,i_0}$. Here, the summation and delta cancel out to leave us with  $\Delta_{i_1}\ldots\Delta_{i_k}V_{i_1,\ldots,i_k, |ROOT|}$. Because $B=\infty$, we know from Thm.~\ref{th:main:star} item 6 that $V$ is a minimum of the form $\Delta_{i_1}\ldots\Delta_{i_k}\min{F_1(i_1),\ldots,F_k(i_k), |ROOT|}$. Therefore, we can remove this final term and we get back to our expression for $DSB(Q,\bm f)$. 

We now prove that $DSB(Q^*,\bm f)\leq FDSB(Q,\bm f, \texttt{ROOT})$. To do this, we rewrite the expression for $DSB(Q^*,\bm f)$ as follows,
\begin{align*}
    DSB(Q^*, \bm f) &=\sum_{i_0=1,|ROOT|} \bm{w}^{(\texttt{ROOT})}(i)
\end{align*}
This is possible because the root relation now has a parent variable $X_0$. Further, according to Lemma~\ref{lmm:compress}, $\bm{w}\leq \hat{\bm w}$, so replacing $\bm{w}$ for $\hat{\bm w}$ throughout the algorithm will only increase the result. Doing this replacement, we get,
\begin{align*}
    DSB(Q^*, \bm f) &\leq \sum_{i_0=1,|ROOT|}  \hat{\bm{w}}^{(\texttt{ROOT})}(i)\\
    DSB(Q^*, \bm f) &\leq \sum_{i_0=1,|ROOT|} f_0(i) \cdot \prod_{i=1,p}a_p(\max(1, F_p^{-1}(F_0(i))))
\end{align*}
Because the degree sequence of $X_0$ is equal to $\bm 1$, we know $f_0(i)=1$ and $F_0(i)=i$.
\begin{align*}
    DSB(Q^*, \bm f) &\leq \sum_{i_0=1,|ROOT|} 1 \cdot \prod_{i=1,p}a_p(\max(1, F_p^{-1}(i)))
\end{align*}
Recognizing the RHS as the summation for $FDSB$, we get the result.
\begin{align*}
    DSB(Q^*, \bm f) &\leq FDSB(Q, \bm f, \texttt{ROOT})
\end{align*}
\end{proof}
Therefore, for any cover $Q_1,\ldots,Q_m$, we have $\prod_{i=1,m}DSB(Q_i, \bm f)\leq \prod_{i=1,m}\min_{ROOT\in\bm R(Q_i)}FDSB(Q_i, \bm f, \texttt{ROOT})$. To finish the proof of part 1, we simply note the result from \ref{lemma:connection:to:pb} that $DSB(Q,\bm f)\leq DSB(Q_1)\cdots DSB(Q_m)$ for any cover $Q_1,\ldots,Q_m$.

We prove the second part of the theorem in three steps. First, we prove some basic facts about the complexity of computing the composition and multiplication of piece-wise functions. Second, we bound the complexity of computing $\hat{\bm w}^{(R)}$ given its inputs are piece-wise functions. Lastly, we extend this bound to fully computing Alg. \ref{alg:bottom:up} which proves the theorem.\\\\
We start by proving the tractability of piece-wise composition,
\begin{lmm}\label{lmm:piecewise:comp}
Let $F(x), G(x)$ be non-decreasing piece-wise linear functions with the set of separators $S_F=[m_1,\ldots,m_{s_F}],S_G=[n_1,\ldots,n_{s_G}]$ with cardinalities $s_F,s_G$. Then, their composition, $F(G(x))$, is a piece-wise linear function with separators $S'=G^{-1}(S_F)\cup S_G = [l_1,\ldots,l_{|S'|}]$ and has the following complexity:
\begin{equation*}
    C_{comp} = O((s_F + s_G)(\log(s_F)+\log(s_G)))
\end{equation*}
\end{lmm}
\begin{proof}
Because $S'$ refines the separators of $S_G$, $G(x)$ is a simple linear function when restricted to one interval of $S'$. Further, consider a interval defined by this new set of separators $(l_{i}, l_{i+1}]$, then this interval is contained within an interval of $S_G$ because otherwise there would be a separator of $S_G$ which is not within $S'$. Further, $(G(l_i),G(l_{i+1})]$ is within an interval of $S_F$. Suppose the latter were not true, then there would be a separator $m_j$ such that $l_{i}<G^{-1}(m_j)<l_{i+1}$. Because $G^{-1}(m_j)$ is in the set $S'$, this is impossible. Therefore, the function $F(G(x))$ can be fully defined by simply looking up the linear function at $G(m)$ and at $F(G(m))$ and composing them for every interval $(l_i,l_{i+1}]$ in $S'$.

Looking at the computational time piece-by-piece, we have to compute the result of $G^{-1}$ $s_F$ times and lookup the linear functions of $F$ and $G$ at a particular point $s_F+s_G$ times. Each of these computations/lookups takes $\log(s_G),\log(s_F)$ operations, respectively, so the overall computational complexity is $O((s_F+s_G)(\log(s_G) + \log(s_F)))$.
\end{proof}
The lemma for multiplication proceeds very similarly,
\begin{lmm}\label{lmm:piecewise:mult}
Let $f_1(x),...,f_d(x)$ be non-decreasing $s_i$-staircase functions with the set of separators denoted $S_{F_i}$. Then, their multiplication, $\prod_{p\in[1,d]}F_p(x)$, is an $\sum_{i=1}^ds_i$-staircase function with divisors $S' = \bigcup_{p\in[1,d]}S_{F_p}=[l_1,\ldots,l_{|S'|}]$ and has the following complexity:
\begin{equation*}
    C_{mult} = O((\sum_{p=1}^ds_{F_p})(\sum_{p=1}^d\log(s_{F_p})))
\end{equation*}
\end{lmm}
\begin{proof}
Because the joint set of separators $S'$ refines each factor's separators $S_{F_p}$, within a particular interval $(l_i,l_{i+1}]$ each function is a simple polynomial. Therefore, the function $\prod_{p=1}^dF_p(x)$ can be fully defined by simply looking up each function's constant at $l_i$ and multiplying them for each interval $(l_i,l_{i+1}]$ in $S'$.

Considering the computational time, computing the product requires looking up each factor's constant and computing these constant's product within each interval of $S'$. Each of these lookups takes $\log(s_{F_p})$ operations. So, the overall computational complexity is $O((\sum_{p=1}^ds_{F_p})(\sum_{p=1}^d\log(s_{F_p}))$.
\end{proof}

Given these lemmas, we can bound the computational time of computing $\hat{\bm w}^{(R)}$ and describe the resulting function,
\begin{lmm}
Let each vector $f_p\defeq f^{(R,X_p)}$ be an $s_{R,X_p}$-staircase and $a_p\defeq a^{(X_p)}$ be $t_{X_P}$-staircases with dividers $S_{F_p}$ and $S_{a_p}$, respectively. Then, the vector $\hat{\bm w}^{(R)}$ as defined in Alg. \ref{alg:bottom:up:f} is a $(\sum_{p=1}^d s_{R,X_p}+\sum_{p=2}^dt_{X_p})$-staircase and can be computed in time.
\begin{equation*}
    C_{total,R} = O((ds_{R,X_1}+\sum_{p=2}^d s_{R,X_p}+\sum_{p=2}^dt_{X_p})
    \cdot(d\log(s_{R,X_1})+\sum_{p=2}^d\log(s_{R, X_p})+\sum_{p=2}^d\log(t_{X_p})+ (\sum_{p=2}^dk_{X_p})^2))
\end{equation*}
\end{lmm}
\begin{proof}
We start by restating the expression for $\hat{\bm w}^{(R)}$,
\begin{align*}
 \hat{\bm w}^{(R)}  \defeq & \left(\Delta_{i_1} F_1(i_1)\right)\prod_{p\in[2,d]}a_p\left( F_p^{-1}(F_1(i_1-1))\right)
\end{align*}
We can now calculate the set of separators of $\hat{\bm w}^{(R)}$ using the previous two lemmas'. First, we consider the separators for each factor of the product. Based on Lemma \ref{lmm:piecewise:comp} and the fact that $F_p^{-1}\circ F_1 = F_1^{-1}\circ F_p$, we get the following expression for their separators,
\begin{align*}
    S_{a_p\circ F_p^{-1}\circ F_1} = F_1^{-1}(F_p(S_{a_p}))\cup F_1^{-1}(S_{F_p^{-1}})\cup S_{F_1}
\end{align*}
At this point, we take the union of the factors' separators to get the full set of separators for $\hat{\bm w}^{(R)}$,
\begin{align*}
S_{\hat{\bm w}^{(R)}} &= S_{F_1}\cup\left(\bigcup_{p=2}^dF_1^{-1}(F_p(S_{a_p}))\cup F_1^{-1}(S_{F_p^{-1}})\cup S_{F_1}\right) = S_{F_1}\cup\left(\bigcup_{p=2}^dF^{-1}_1(F_p(S_{a_p}))\right)\cup\left(\bigcup_{p=2}^d F_1^{-1}(S_{F_p^{-1}})\right)
\end{align*}
Therefore, $\hat{\bm w}^{(R)}$ has at most $\sum_{p=1}^ds_{R, X_p} + \sum_{p=1}^dt_{X_p}$ separators. 

Next, we consider the computational complexity and degree in two steps: 1) $C_{comp,R}$, defined as the time to calculate all of the $a_p(F^{-1}_p(F_1(i_1)))$ expressions  and 2) $C_{mult,R}$, defined as the time to calculate the product of these expressions. Starting with the compositions, calculating a single expression of the form $a_p(F^{-1}_p(F_1(i_1)))$ requires calculating two compositions $F^{-1}_p \circ F_1$ and $a_p\circ(F^{-1}_p\circ F_1)$. This first composition takes the following operations by Lmm. \ref{lmm:piecewise:comp},
$$O((s_{R,X_p}+s_{R,X_1})(\log(s_{R,X_p}) +\log(s_{R,X_1})))$$
The second composition takes,
$$O((t_{X_p}+s_{R,X_p}+s_{R,X_1})(\log(t_{X_p})+\log(s_{R,X_p}+s_{R,X_1})))$$
Because $\log(s_{R, X_p}+s_1)\leq \log(s_{R, X_p}) + \log(s_1)$, the sum of these two times is upper bounded by,
$$O((t_{X_p}+s_{R,X_p}+s_{R,X_1})(\log(t_{X_p})+\log(s_{R,X_p})+\log(s_{R,X_1})))$$
So, the first step has the following complexity and results in a staircase function because composition with a constant function results in a constant function,
$$C_{comp,R} = O\left(\sum_{p=2}^d(t_{X_p}+s_{R, X_p}+s_{R,X_1})(\log(s_{R, X_p})+\log(s_{R,X_1})+\log(t_{X_p}))\right)$$ 

The product of these functions then takes the following operations by Lmm. \ref{lmm:piecewise:mult}, 
\begin{equation*}
    C_{mult,R} =O\left(\left(s_{R,X_1}+(\sum_{p=2}^ds_{R, X_p}+t_{X_p}+s_{R,X_1})\right)\left(\log(s_{R,X_1})+\sum_{p=2}^d\log(s_{R, X_p}+t_{X_p}+s_{R,X_1})\right)\right)
\end{equation*}
Splitting the logarithms results in the following cleaner expression,
\begin{equation*}
    C_{mult,R} = O\left(\left(ds_{R,X_1}+\sum_{p=2}^ds_{R, X_p}+\sum_{p=2}^dt_{X_p}\right)\left(d\log(s_{R,X_1})+\sum_{p=2}^d\log(s_{R, X_p})+\sum_{p=2}^d\log(t_{X_p})\right)\right)
\end{equation*}
Because $C_{comp,R}\leq C_{mult,R}$ (the sum of a product is less than the product of sums), we can then say that the whole computation has complexity:
\begin{equation*}
    C_{total,R} = O\left(\left(ds_{R,X_1}+\sum_{p=2}^ds_{R, X_p}+\sum_{p=2}^dt_{X_p}\right)\left(d\log(s_{R,X_1})+\sum_{p=2}^d\log(s_{R, X_p})+\sum_{p=2}^d\log(t_{X_p})\right)\right)
\end{equation*}
\end{proof}

Next, we extend this analysis to the full bound computation for a given root,
\begin{thm} Let $Q$ be a Berge-acyclic query with $M$ relations (as in Eq.~\eqref{eq:q}), whose degree sequences are represented by $s_{R,Z}$-staircase functions.  Then, (1) $FDSB(Q, \texttt{ROOT})\geq DSB(Q)$, and (2) it can be computed in the following polynomial time combined complexity,
   \begin{equation*}
     T_{FDSB} = O(M(s+\max_{R\in\bm R,Z\in\bm X_R}(\text{Arity}(R)\cdot s_{R,Z}))(s_{log}+\max_{R\in\bm R,Z\in\bm X_R}(\text{Arity}(R)\cdot\log(s_{R,Z}))))
   \end{equation*}
   where $s = \sum_{R,Z} s_{R,Z}$ and
   $s_{log} = \sum_{R,Z} \log(s_{R,Z})$. 
\end{thm}
\begin{proof}
To produce this bound, we bound the maximum computation necessary to compute any $\hat{\bm w}^{(R)}$ then we multiply that bound by $|R|$, i.e. the number of $w^{(R)}$ expressions that we have to compute. To produce this first bound, we start with the complexity of computing an arbitrary $\hat{\bm w}^{(R)}$,
\begin{equation*}
    C_{total,R} = O\left(\left(ds_{R,X_1}+\sum_{p=2}^ds_{R, X_p}+\sum_{p=2}^dt_{X_p}\right)\left(d\log(s_{R,X_1})+\sum_{p=2}^d\log(s_{R, X_p})+\sum_{p=2}^d\log(t_{X_p})\right)\right)
\end{equation*}
First, we replace $\sum_p s_{R, X_p}+\sum_p t_{X_p}$ with the sum of all segments in the database, similarly for the log terms. To see why this is valid, we note that $t_{X_p}$ is equal to the number of segments in relations in the sub-tree of the incidence graph rooted at $X_p$. By the same logic, $\sum_ps_{R,X_p}+\sum_pt_{X_p}$ is equivalent to the sum of all segments in relations in the subtree of the incidence graph rooted at $R$. Therefore, this expression must be upper bounded by the number of segments in the entire incidence graph, i.e. $\sum_{R\in \bm R}\sum_{Z\in \bm X_R} s_{R,Z})$. This gives us the following expression,
\begin{multline*}
    C_{total,R} = O((d\cdot s_{R,X_1}+\sum_{R\in \bm R}\sum_{Z\in \bm X_R} s_{R,Z})\cdot(d\log(s_{R,X_1})+    \sum_{R\in \bm R}\sum_{Z\in \bm X_R} \log(s_{R,Z})))
\end{multline*}
Next, we need to remove the reliance on $dS_{R,X_1}$ as it refers to the particular table $R$. We do so by replacing it with its maximum over all relations,
\begin{multline*}
C_{total,R} = O((\max_{\bm R, Z\in\bm X_R}(Arity(R)\cdot s_{R,Z})+\sum_{R\in   \bm R}\sum_{Z\in \bm X_R} s_{R,Z})\\(\max_{\bm R, Z\in\bm X_R}(Arity(R)\cdot\log(s_{R,Z}))+\sum_{R\in \bm R}\sum_{Z\in \bm X_R} \log(s_{R,Z})))
\end{multline*}
Therefore, we can bound the overall computation by this maximum
per-table-computation multiplied by the number of tables,
\begin{align*}
C_{total} &= O(\sum_{R\in \bm R}C_{total,R})\\
C_{total}  &= O(|\bm R|\cdot(\max_{\bm R, Z\in\bm X_R}(Arity(R)\cdot s_{R,Z})+\sum_{R\in \bm R}\sum_{Z\in \bm X_R} s_{R,Z})\cdot\\
&\hspace{115pt}(\max_{\bm R, Z\in\bm X_R}(Arity(R)\cdot\log(s_{R,Z}))+\sum_{R\in \bm R}\sum_{Z\in \bm X_R} \log(s_{R,Z})))
\end{align*}
\end{proof}

\subsubsection{Proof of Theorem \ref{thm:compress:covers}}
\label{app:thm:compress:dynamic}

\begin{thm}Let $Q$ be a Berge-acyclic query with $M$ relations (as in Eq.~\eqref{eq:q}), whose degree sequences are represented by $s_{R,Z}$-staircase functions.  Then, (1) $FDSB(Q)\geq DSB(Q)$, and (2) it can be computed in polynomial time combined complexity $O(2^{m}\cdot(2^m + M\cdot T_{FDSB}))$.
\end{thm}

\begin{proof}
Claim 1 follows directly from part 2 of Lemma~\ref{lemma:connection:to:pb} and part 1 of Theorem~\ref{thm:compress}, i.e. for all covers $Q_1,\ldots,Q_k$ and roots $\texttt{ROOT}$,
\begin{align}
    DSB(Q) &\leq DSB(Q_1)\cdots DSB(Q_k)\\
    DSB(Q) &\leq FDSB(Q, \texttt{ROOT}
\end{align}
Let $Q_1,\ldots,Q_k$ be the partitions of the optimal cover used in $FDSB(Q)$ and let $\texttt{ROOT}_i\,\,\forall\,\,1\leq i\leq k$ be the optimal roots. Combining the above statements, we get the first part of our theorem,
\begin{align}
    DSB(Q) \leq FDSB(Q_1,\texttt{ROOT}_1)\cdots FDSB(Q_k,\texttt{ROOT}_k) = FDSB(Q)
\end{align}

Claim 2 is the result of the dynamic programming algorithm~\ref{alg:compress:dynamic} which uses $FDSB(Q, \texttt{ROOT})$ as a sub-routine.
\begin{algorithm}[H]
  \caption{Computing $FDSB(Q)$}
\label{alg:compress:dynamic}
\begin{algorithmic}
\STATE $FDSB(Q) = \infty$
\FOR{$i\,\,\in\,\, 1, |\bm R|$}{
    \FOR{subset S in all size i subsets of $\bm R$}{
        \STATE $S_1,\ldots, S_l = ConnectedComponents(S)$
        \STATE $COST(S) = \prod_{p=1,l}\min_{\texttt{ROOT}\in S_p} FDSB(S_p, \texttt{ROOT})$
        \FOR{$Q_1$, $Q_2$ in Partitions of S}{
            \STATE $COST(S) = min(COST(S), COST(Q_1)*COST(Q_2))$
        }\ENDFOR
        \IF{S is cover}{
            \STATE $FDSB(Q) = min(FDSB(Q), COST(S))$
        }\ENDIF
    }\ENDFOR
}
\ENDFOR
\RETURN $FDSB(Q)$
\end{algorithmic}
\end{algorithm}
This algorithm runs in time $O(\cdot 2^m(2^m + mT_{FDSB}))$. This is because the number of subsets of $\bm R$ is $2^m$, so the outer two for loops will only result in $2^m$ iterations. Within each of those iterations, computing the minimum $FDSB$ for each root within connected component requires checking at most $m$ roots, resulting in $O(mT_{FDSB})$. Further, you need to consider every partition of $S$ into two sub-queries which can occur at most $2^m$ ways. However, constant work is done for each of those sub-queries because $COST(Q_1)$ and $COST(Q_2)$ have already been calculated in previous iterations.

We now show that this algorithm is correct, i.e. it returns $FDSB(Q)$. Because it returns a multiplication of FDSB's calculated over partitions in coverings of $Q$, it is immediate to see that the return value is greater than or equal to $FDSB(Q)$. Therefore, we just need to show that its return value is less than or equal to $FDSB(Q)$. Let $Q_1,\ldots,Q_k$ be the optimal partition of $Q$, $\texttt{ROOT}_1,\ldots, \texttt{ROOT}_k$ be the optimal roots, and $S^* = R_1,\ldots,R_r$ be the relations in the optimal covering. Lastly, let $\bm P(S)$ represent the set of disjoint partitions of $S$.

We proceed inductively on the iterations of the outer loop with the following hypothesis,
\begin{align*}
    COST(S)\leq \min_{S_1,\ldots,S_k\in\bm P(S')}\prod_{i=1,k}\min_{\texttt{ROOT}}FDSB(S_i,\texttt{ROOT}) \forall\,\, S\subseteq \bm R
\end{align*}
Therefore, the following is our inductive assumption,
\begin{align*}
    COST(S')\leq \min_{S_1,\ldots,S_k\in\bm P(S')}\prod_{i=1,k}\min_{\texttt{ROOT}}FDSB(S_i,\texttt{ROOT}) \forall\,\, |S'|<i, S'\subset \bm R
\end{align*} 
Let $S\subseteq \bm R$ be a subset of size $i$. We note that any partition of $S$ can be described as a series of binary partitions. From this, we can see immediately that $COST(S)$ satisfies our inductive hypothesis because it is the minimum over all binary partitions as well as the trivial partition, $Q_1 = S, Q_2=\{\}$. Lastly, our base case of $i=1$ is immediate as the trivial partition is the only available partition. This inductive proof immediately gives us that $COST(S)$ is less than or equal to the product of minimum $FDSB$ over all roots for each partition which is equivalent to the definition of $FDSB(Q)$.
\end{proof}

\subsection{Proof of Lemma \ref{lmm:compress:polymatroid}}
\label{app:lmm:compress:polymatroid}

\begin{lmm}
Suppose $Q$ is a Berge-acyclic query with degree sequences $\bm f$, then the following is true,
\begin{align}
     FDSB(Q, \bm f, ROOT) \leq PB(Q, ROOT)
\end{align}
\end{lmm}
\begin{proof} 
We prove by induction on the tree that, for all $R\neq ROOT$, and every value $i\geq 1$, $\hat w^{(R)}(i)\leq \prod_{S\in tree(X_p)}f^{(S,Z_S)}(1)$. Assuming this holds for all children of $R$, we consider the definition of $\hat{\bm w}^{(R)}$. Because the functions $a_p$ are non-increasing and the maximum restricts their input to $\geq 1$, we know that $\hat w(i) \leq f_1(i)\cdot \prod_{i=2,p} a_p(1)$. By induction hypothesis, for each function $a_p$, we know $a_p(1)\leq \prod_{S\in tree(X_p)}f_1^{(S,Z_S)}$. Therefore, we immediately get,
\begin{align*}
    \hat w(i) \leq f_1(i)\cdot \prod_{i=2,p} \prod_{S\in tree(X_P)}f^{(S,Z_s)}(1) \leq \prod_{S\in tree(R)}f^{(S,Z_S)}(1)
\end{align*}
As before, this completes the proof of our inductive step. The base case of a leaf relation is trivial as in that case $w(i) = f_1(i) \leq f_1(1)$. Lastly, we consider the expression for the final sum and see that the lemma holds.
\begin{align*}
\sum_{i=1,|\texttt{ROOT}|}\prod_{p=1,k}a_p(\max(1,F_p^{-1}(i)))&\leq \sum_{i=1,|\texttt{ROOT}|} \prod_{p=1,k}\prod_{S\in tree(X_p)}f^{(S,Z_S)}(1)\\
&\leq |\texttt{ROOT}| \prod_{S\neq \texttt{ROOT} \in tree(\texttt{ROOT})}f^{(S,Z_S)}(1) = PB(Q, \texttt{ROOT})
\end{align*}
\end{proof}

\subsection{Proof of Theorem  \ref{th:approximate:agm:polymatroid}}
\label{app:thm:compress:polymatroid}
\begin{thm}
Suppose $Q$ is a Berge-acyclic query.  Then the following hold:
\begin{align}
  && |Q| \leq  & FDSB(Q) \leq PB(Q) \leq AGM(Q)
\end{align}
\end{thm}

\begin{proof}
The theorem follows directly from lemma~\ref{lmm:compress:polymatroid} and the definitions of $FDSB(Q)$ and $PB(Q)$ which we restate below,
\begin{align*}
  PB(Q) \defeq & \min_{\bm W}\prod_{i=1,m} \min_{\texttt{ROOT}\in \bm R(Q_i)}PB(Q_i,\texttt{ROOT}_i)\\
  FDSB(Q) \defeq & \min_{\bm W}\prod_{i=1,m} \min_{\texttt{ROOT}\in \bm R(Q_i)}FDSB(Q_i,\texttt{ROOT}_i)
\end{align*}
Because we know $FDSB(Q_i, \texttt{ROOT}_i)\leq PB(Q_i, \texttt{ROOT}_i)$ from lemma ~\ref{lmm:compress:polymatroid}, our result immediately follows.
\begin{align*}
  PB(Q) \leq & \min_{\bm W}\prod_{i=1,m} \min_{\texttt{ROOT}\in \bm R(Q_i)}FDSB(Q_i,\texttt{ROOT}_i) = FDSB(Q)\\
\end{align*}
\end{proof}

%% file: theorem-3-2.tex
\subsection{Proof of Theorem~\ref{th:main:star}}
\label{app:thm:main:star}
%
In this section we prove each of the items in the theorem and restate them at each point for convenience. 

\subsubsection{Permutations $\sigma$ Are the Identity}

We start by proving that it suffices to restrict $\bm \sigma$ in
Problem~\ref{prob:pessimistic:tensor} to be the identity permutations.
More precisely, we prove:

\begin{lmm}
  \label{lemma:sigma:identity} For any tensor
  $\bm M \in \calM^+_{\bm f^{(\bm X)},B}$ and any tuple of permutations
  $\bm \sigma \in S_{[\bm n]}$ there exists a tensor
  $\bm N \in \calM^+_{\bm f^{(\bm X)},B}$ such that, for all
  non-increasing vectors $\bm a^{(X_p)} \in \R_+^{[n_p]}$, $p=1,d$,
  the following holds:
    \begin{align*}
          (\bm M \circ \bm \sigma) \cdot \bm a^{(X_1)} \cdots \bm a^{(X_d)} \leq & \bm N \cdot \bm a^{(X_1)} \cdots \bm a^{(X_d)}
    \end{align*}
\end{lmm}

The lemma implies that $\bm \sigma$ can be chosen to be the identity
in Problem~\ref{prob:pessimistic:tensor}, because we can simply
replace $\bm M$ and $\bm \sigma$, with $\bm N$ and the identity
permutation.  In other words that, in order to compute the upper bound
to $Q_{\texttt{start}}$, it suffices to restrict the relation instance
$S$ to have the highest degrees aligned with the highest degrees of
the unary relations $R^{(p)}$, $p=1,d$.
%
%

To prove Lemma~\ref{lemma:sigma:identity}, we establish the following:

\begin{lmm} \label{lemma:matching} Let $\bm M \in \R_+^{[\bm n]}$ be
  an $\bm X$-tensor s.t.
  $\bm M \circ \bm \sigma^{-1} \in \calM^+_{\bm f^{(\bm X)}, B}$ for
  some $\bm \sigma$.  Then there exists an $\bm X$-tensor
  $\bm N \in \R_+^{[\bm n]}$ satisfying:
  \begin{itemize}
  \item $||\bm M||_\infty = ||\bm N||_\infty$.
  \item For every variable $X_p$,
    $\texttt{SUM}_{\bm X-\set{X_p}}(\bm N)$ is non-increasing, and is
    equal, up to a permutation, to
    $\texttt{SUM}_{\bm X-\set{X_p}}(\bm M)$.  In particular,
    $\bm N \in \calM^+_{\bm f^{(\bm X)}, B}$.
  \item For any non-increasing vectors
    $\bm a^{(X_p)} \in \R_+^{[n_p]}$:
  \begin{align}
    \bm M \cdot \bm a^{(X_1)}\cdots \bm a^{(X_d)} \leq & \bm N \cdot \bm a^{(X_1)}\cdots \bm a^{(X_d)}  \label{eq:m:leq:n}
  \end{align}
  \end{itemize}
\end{lmm}

We first show that Lemma~\ref{lemma:matching} implies
Lemma~\ref{lemma:sigma:identity}.  Let $\bm M, \bm \sigma$ be given as
in Lemma~\ref{lemma:sigma:identity}.  We apply
Lemma~\ref{lemma:matching} to $\bm M \circ \bm \sigma$, and obtain
some tensor $\bm N \in \calM^+_{\bm f^{(\bm X)}, B}$ such that
$(\bm M \circ \bm \sigma) \cdot \bm a^{(X_1)} \cdots \bm a^{(X_d)}
\leq \bm N \cdot \bm a^{(X_1)} \cdots \bm a^{(X_d)}$.  Thus, it
remains to prove Lemma~\ref{lemma:matching}.

The proof of Lemma~\ref{lemma:matching} consists of modifying $\bm M$
by sorting each dimension $X_p$ separately.  To sort the dimension
$X_p$, we cannot simply switch some hyperplane $X_p=i$ with another
hyperplane $X_p=j$, because that doesn't guarantee~\eqref{eq:m:leq:n}.
Instead, we move mass ``up'' in the tensor along the dimension $X_p$.
We first illustrate the proof idea with an example.

\begin{ex}
  We want to maximize the product $\bm a^T \cdot \bm M \cdot \bm b$,
  for some $2 \times 3$ matrix $\bm M$, constraint by the degree
  sequences $\bm f = (16,12)$ and $\bm g=(14,11,3)$, and element bound
  $B=10$.  Our current candidate is
  \begin{align*}
    \bm M =
    & \begin{pNiceMatrix}[first-row,first-col]
      \  & 14 & 3 & 11 \\
      12 & 10 & 1 & 1  \\
      16 & 4  & 2 & 10 
    \end{pNiceMatrix}
  \end{align*}
  We will modify $\bm M$ to a new matrix $\bm N$, whose row-sums are
  $16,12$, while
  $\bm a^T \cdot \bm M \cdot \bm b \leq \bm a^T \cdot \bm N \cdot \bm
  b$, for all $\bm a \bm b$.  Simply switching rows 1 and 2 doesn't
  work: indeed, if we set $\bm a^T \defeq (1,0)^T$,
  $\bm b^T\defeq (1,0,0)$ then $\bm a^T \cdot \bm M \cdot \bm b= 10$,
  but after we switch the two rows,
  $\bm a^T \cdot \bm N \cdot \bm b = 4$.  Instead, we move positive
  mass up, separately in each column.  The maximum we can move in each
  column $k$ is $v_k \defeq (M_{k2}-M_{k1})^+$, where
  $x^+ \defeq \max(x,0)$.  Denoting by $\bm v$ the vector
  $(v_1, v_2, v_3)$, we have:
  \begin{align*}
\bm v = &  \begin{pmatrix}
      0 & 1 & 9 
    \end{pmatrix}
  \end{align*}
  The row-sum in $\bm v$ is 10, but we only need to move a total mass
  of 4 (the difference $16-12$), hence we will move only a
  $\theta = 4/10$-fraction of $\bm v$, namely
  $\theta\cdot \bm v = \begin{pmatrix} 0 & 0.4 & 3.6 \end{pmatrix}$,
  and obtain:
  \begin{align*}
    \bm N \defeq
    & \begin{pNiceMatrix}[first-row,first-col]
      \  & 14 & 3   & 11  \\
      16 & 10 & 1.4 & 4.6 \\
      12 & 4  & 1.6 & 7.4   
    \end{pNiceMatrix}
  \end{align*}
  Finally, we check that the quantity
  $\bm a^T \cdot \bm M \cdot \bm b$ has not decreased:
  \begin{align*}
    \bm a^T \cdot \bm N \cdot \bm b - &  \bm a^T \cdot \bm M \cdot \bm b =
     \left(\bm a^T \cdot (\bm N - \bm M)\right) \cdot \bm b = \\
    = & \theta\cdot
    \begin{pmatrix}
      (a_1-a_2)v_1 & (a_1-a_2)v_2 & (a_1-a_2)v_3
    \end{pmatrix}
\cdot \bm b
  \end{align*}
  This quantity is $\geq 0$ because $a_1 \geq a_2$ and $\bm b \geq 0$.
  Observe that $\bm N$ has exactly the same column-sums as $\bm M$,
  and $||\bm M||_\infty = ||\bm N||_\infty$, while its rows are sorted
  in non-increasing order of their sums.  If we repeat the same
  argument for the columns 2 and 3, then we obtain a matrix whose
  row-sums and column-sums are $(16, 12)$, and $(14,11,3)$, and which
  upper bounds the quantity $\bm a^T \cdot \bm M \cdot \bm b$ for all
  non-increasing $\bm a, \bm b$.
\end{ex}

Before we prove Lemma~\ref{lemma:matching} we need some
notations. An {\em inversion} in a vector $\bm f \in \R^{[n]}$ is a
pair $i < j$ such that $f_i \geq f_j$.  The vector is non-increasing
iff it has the maximum number of inversions, $n(n-1)/2$.  If $(i,j)$
is not an inversion, then $\bm f \circ \tau_{ij}$ has strictly more
inversions, where $\tau_{ij}$ is the permutation swapping $i$ and
$j$.  The proof of Lemma~\ref{lemma:matching} follows from the
following Proposition:

\begin{prop} \label{prop:matching:1}
  Let $\bm M \in \R_+^{[\bm n]}$ be $\bm X$-tensor, $q \in [d]$ be one
  of its dimensions, and $\bm f \defeq \texttt{SUM}_{\bm X-\set{X_q}}(\bm M)$.  If $\bm f$
  as no inversion at $(i,j)$, then there exists a $\bm X$-tensor
  $\bm N \in \R_+^{\bm n}$ with the following properties:
  \begin{itemize}
  \item $||\bm M||_\infty = ||\bm N||_\infty$.
  \item $\texttt{SUM}_{\bm X-\set{X_p}}(\bm M) = \texttt{SUM}_{\bm X-\set{X_p}}(\bm N)$ for all $p \neq q$.
  \item $\texttt{SUM}_{\bm X-\set{X_q}}(\bm N) = \bm f \circ \tau_{ij}$. 
  \item Inequality~\eqref{eq:m:leq:n} holds for any vectors
    $\bm a^{(X_p)} \in \R_+^{[n_p]}$ such that $\bm a^{(X_q)}$ is
    non-increasing,
  \end{itemize}
\end{prop}

By repeatedly applying the proposition to dimension $q$ we can ensure
that $\texttt{SUM}_{\bm X-\set{X_q}}(\bm M)$ is non-increasing; the
proof of Lemma~\ref{lemma:matching} follows by repeating this for each
dimension $q$.

\begin{proof} (of Proposition~\ref{prop:matching:1}) Define:
  \begin{align*}
    \delta \defeq & f_{j}-f_{i}
  \end{align*}
  The idea of the proof is to modify $\bm M$ by pushing a total mass
  of $\delta$ from the hyperplane $j$ to the hyperplane $i$.  When
  doing so we must ensure that the new values in hyperplane $i$ don't
  exceed $B$, those in hyperplane $j$ don't become $<0$, and we don't
  push more mass than $\delta$.

  Denote by $x^+ \defeq \max(x, 0)$ and define the following
  quantities, where $\bm k \in \prod_{p=1,d; p\neq q} [n_p]$, and
  $\bm k i$ denotes the tuple
  $(k_1, \ldots, k_{q-1},i,k_{q+1},\ldots k_d)$:
  \begin{align*}
    \delta_{\bm k} \defeq & M_{\bm k j} - M_{\bm k i}
    &   \theta \defeq & \frac{\delta}{\sum_{\bm k} \delta_{\bm k}^+} 
    & \varepsilon_{\bm k} \defeq &\theta \cdot \delta_{\bm k}^+
  \end{align*}
  Intuitively, $\delta_{\bm k}^+$ is the maximum amount of mass we
  could move from $M_{\bm kj}$ to $M_{\bm ki}$ without causing
  $M_{\bm kj}<0$ or $M_{\bm ki}>B$ (when $\delta_{\bm k}<0$ then we
  cannot move any mass in coordinate $\bm k$), while $\theta$
  represents the fraction of maximum mass that equals to the desired
  $\delta$.

  To define $\bm N$ formally, we introduce some notations.  For each
  variable $X_p \in \bm X$ and $t \in [n_p]$, let
  $\bm e^{(X, t)} \defeq (0,\ldots,0,1,0,\ldots,0)$ where the sole $1$
  is on position $t$.  Given $\bm t=(t_1, \ldots, t_d) \in [\bm n]$
  denote by
  $\bm D^{(\bm t)} \defeq \bigotimes_{p=1,d} \bm e^{(X_p, t_p)}$.
  This tensor has value $1$ on position $(t_1,\ldots, t_d)$ and $0$
  everywhere else.  Then, we define:
  \begin{align*}
    \bm N \defeq & \bm M + \sum_{\bm k} \varepsilon_{\bm k} (\bm D^{(\bm k i)} - \bm D^{(\bm k j)})
  \end{align*}
  Notice that $\delta = \sum_{\bm k} \delta_{\bm k} > 0$, by the
  assumption that $\bm f$ has no inversion at $i,j$, and
  $0 < \theta \leq 1$.  We start by checking that all entries in
  $\bm N$ are $\geq 0$ and $\leq B$. The only entries that have
  decreased are those in the hyperplane $X_q = j$:
  \begin{align*}
    N_{\bm k j} = & M_{\bm k j} - \theta \delta^+_{\bm k}  \geq  M_{\bm k j} -  \delta^+_{\bm k}\\
    = & M_{\bm k j}-\max(M_{\bm k j}-M_{\bm k i}, 0) = \min(M_{\bm k i}, M_{\bm k, j}) \geq 0
  \end{align*}
  The only entries that have increased are in hyperplane $X_p=i$:
  \begin{align*}
    \bm N_{\bm k i} = & \bm M_{\bm k i} + \theta \delta_{\bm k}^+ \leq \bm M_{\bm k i} + \delta_{\bm k}^+\\
= & M_{\bm k i} + \max(\bm M_{\bm k j} - \bm M_{\bm k i}) = \max(\bm  M_{\bm k j}, \bm M_{\bm k i}) \leq B
  \end{align*}
  By Eq.~\eqref{eq:push:aggregates:down} we have:
  \begin{align*}
    \texttt{SUM}_{\bm X-\set{X_p}}\left(\bm  D^{\bm t}\right)=
    &\texttt{SUM}_{\bm X-\set{X_p}} \left(\bigotimes_{\ell=1,d} \bm
      e^{(X_\ell, t_\ell)}\right) =
\bm e^{(X_p,t_p)} \otimes \bigotimes_{\ell=1,d; \ell\neq
      p}\texttt{SUM}_{X_\ell}(\bm e^{(X_\ell,t_\ell)})=\bm e^{(X_p, t_p)}
  \end{align*}
  because $\texttt{SUM}_{\bm X}\left(\bm e^{(X_\ell, t_\ell)}\right)=1$.
  Therefore, $\texttt{SUM}_{\bm X-\set{X_p}}\left(\bm D^{(\bm k i)}\right)$ is $\bm e^{X_p,
    k_p}$ when $p \neq q$, and is $\bm e^{X_q, i}$ when $p = q$, and
  we obtain:
  \begin{align*}
p\neq q: &&   \texttt{SUM}_{\bm X-\set{X_p}}(\bm N-\bm N) = &
                                                              \sum_{\bm k} \varepsilon_{\bm k} \texttt{SUM}_{\bm X-\set{X_p}}\left(\bm D^{(\bm k i)} - \bm D^{(\bm k j)}\right)
= 0 \\
p=q: && \texttt{SUM}_{\bm X-\set{X_q}}(\bm N-\bm N) = & \sum_{\bm k}
                                                   \varepsilon_{\bm k}
                                                   \texttt{SUM}_{\bm X-\set{X_q}}\left(\bm D^{(\bm k i)} - \bm D^{(\bm k j)}\right)\\
&& = & \left(\sum_{\bm k} \varepsilon_{\bm k}\right) (\bm e^{X_q, i}-\bm e^{X_q, j})
= \delta (\bm e^{X_q, i}-\bm e^{X_q, j})
  \end{align*}
  proving that $\texttt{SUM}_{\bm X-\set{X_q}}(\bm N) = \bm f \circ \tau_{ij}$.

  Finally, we check Inequality~\eqref{eq:m:leq:n}. First, we use the
  definition of dot product in Def.~\ref{def:tensor:ops} and
  Eq.~\eqref{eq:push:aggregates:down} to derive:
  \begin{align*}
    \bm D^{(\bm t)} \cdot \bm a^{(X_1)} \cdots \bm a^{(X_d)}= & \texttt{SUM}_{\bm X} \left(\bm D^{(\bm t)} \otimes \bm a^{(X_1)}  \otimes\cdots\otimes a^{(X_d)}\right) 
    \\=& \texttt{SUM}_{\bm X} \left((\bm e^{(X_1,t_1)}\otimes \bm a^{(X_1)})\otimes \cdots \otimes  (\bm e^{(X_d,t_d)}\otimes \bm a^{(X_d)}) \right)\\
    \\= & \texttt{SUM}_{X_1}(\bm e^{(X_1,t_1)}\otimes \bm a^{(X_1)})\cdots\texttt{SUM}_{X_d}(\bm e^{(X_d,t_d)}\otimes \bm a^{(X_)})
        \\=& a_{t_1}^{(X_1)} \times \cdots \times a_{t_d}^{(X_d)} = \prod_{p=1,d} a_{t_p}^{(X_p)}
  \end{align*}
  This implies:
  \begin{align*}
    \left(\bm N - \bm M\right) \cdot \bm a^{(X_1)} \cdots \bm a^{(X_d)} =
    & \sum_{\bm k} \varepsilon_{\bm k} (\bm D^{(\bm ki)}-\bm D^{(\bm kj)}) \cdot \bm a^{(X_1)} \cdots \bm a^{(X_d)}
    \\ =& \sum_{\bm k} \left(\varepsilon_{\bm k}a_i^{(X_q)}\prod_{p=1,d; p\neq q} a_{k_p}^{(X_p)}\right)-
        \sum_{\bm k} \left(\varepsilon_{\bm k}a_j^{(X_q)}\prod_{p=1,d; p\neq q} a_{k_p}^{(X_p)}\right)\\
    \\= &  \left(a^{(X_q)}_i-a^{(X_q)}_j\right) \sum_{\bm k} \varepsilon_{\bm k} \prod_{p=1,d; p\neq q} a^{(X_p)}_{k_q}
  \end{align*}
  This quantity is $\geq 0$ because $i < j$ and $a^{(X_q)}$ is
  non-increasing by assumption.
\end{proof}

\subsubsection{Proof of  Item~\ref{item:th:main:star:2} (a)}
\label{app:item:1:a}

\noindent\textbf{Item~\ref{item:th:main:star:2}:} \textit{$\bm C$ is a solution to
  Problem~\ref{prob:pessimistic:tensor}, i.e.
  $\bm C \in \calM_{\bm f^{(\bm X)},\infty}$ and it satisfies
  Eq.~\eqref{eq:c:is:max}.  Furthermore, it is tight in the following
  sense: there exists a tensor $\bm M \in \calM^+_{\bm f^{(\bm X)},B}$
  and non-increasing vectors $\bm a^{(p)} \in \R_+^{[n_p]}$, $p=1,d$,
  such that inequality~\eqref{eq:c:is:max} (with $\bm \sigma$ the
  identity) is an equality.}
\\\\
We split Item~\ref{item:th:main:star:2} into two parts:
\begin{description}
\item[(a)] $\bm C$ satisfies inequality~\eqref{eq:c:is:max} and there
  exists tensors $\bm M, \bm a^{(p)}$ such that the inequality becomes
  an equality, and
\item[(b)] $\bm C \in \calM_{\bm f^{(\bm X)},\infty}$.
\end{description}
Here we prove item (a), and defer item (b) to Sec.~\ref{app:item:1:b}.

We have seen in Lemma~\ref{lemma:sigma:identity} that it suffices to
take $\bm \sigma$ to be the identity.  Therefore, in order to $\bm C$
satisfies inequality~\eqref{eq:c:is:max}, it suffices to show that,
for all non-decreasing vectors $\bm a^{(X_p)}\in \R_+^{[n_p]}$,
$p=1,d$, and for every tensor $\bm M \in \calM^+_{\bm f^{(\bm X)},B}$,
the following holds:
\begin{align}
\bm M \cdot \bm a^{(X_1)} \cdots \bm a^{(X_d)} \leq\, & \bm C \cdot \bm a^{(X_1)} \cdots \bm a^{(X_d)}
\label{eq:c:is:max:no:sigma}
\end{align}
Call $\bm a \in \R^n$ a {\em one-zero vector} if
$\bm a = (1,1,\ldots,1,0,\ldots,0)$.  For $m=1,n$, denote by
$\bm b^{(m)}$ the one-zero vector with exactly $m$ 1's.  Every
non-increasing $\bm a \in \R_+^n$ is a positive linear combination of
$\bm b^{(m)}$'s.
\begin{align*}
  \bm a = & (a_1-a_2) \bm b^{(1)} + (a_2-a_3) \bm b^{(2)} + \cdots + (a_{n-1}-a_n) \bm b^{(n-1)} + a_n \bm b^{(n)}
\end{align*}
This implies that, in order to prove
inequality~\eqref{eq:c:is:max:no:sigma}, it suffices to check it when
each vector $\bm a^{(X_p)}$ is some one-zero vector $\bm b^{(X_p,m)}$,
where we include $X_p$ in the superscript to indicate that this is an
$X_p$-vector.  We start by observing that the LHS
of~\eqref{eq:c:is:max:no:sigma} is:
\begin{align}
\bm M \cdot \bm b^{(X_1,m_1)} \cdots \bm b^{(X_d,m_d)} =\, & (\Sigma_{X_1}\cdots \Sigma_{X_d} \bm M)_{\bm m}\label{eq:step:2}
\end{align}
and similarly for the RHS.  Next, by the definition of $\bm V_{\bm m}$
in Eq.~\eqref{eq:def:v},
$(\Sigma_{X_1} \cdots \Sigma_{X_d} \bm M)_{\bm m} = \sum_{\bm s \leq
  \bm m} M_{\bm s} \leq \bm V_{\bm m}$, therefore:
\begin{align}
(\Sigma_{X_1}\cdots \Sigma_{X_d} \bm M)_{\bm m} \leq & \bm V_{\bm m} \label{eq:step:1}
\end{align}
Recall that $\bm C \defeq \Delta_{X_1}\cdots \Delta_{X_d}\bm V$.  By
repeatedly applying Eq.~\eqref{eq:sigma:delta:cancel}, we derive:
\begin{align}
  \Sigma_{X_1}\cdots \Sigma_{X_d} \bm C=\,& \bm V \label{eq:step:3}
\end{align}
Therefore, we obtain that, for all
$\bm M \in \calM^+_{\bm f^{(\bm X)},B}$:
\begin{align*}
  \bm M \cdot \bm b^{(X_1,m_1)} \cdots \bm b^{(X_d,m_d)}\stackrel{\eqref{eq:step:2},\eqref{eq:step:1}}{\leq}&\bm V_{\bm m}
\stackrel{\eqref{eq:step:3}}{=} (\Sigma_{X_1}\cdots \Sigma_{X_d} \bm C)_{\bm m}\\
\stackrel{\eqref{eq:step:2}}{=}&\bm C \cdot \bm b^{(X_1,m_1)} \cdots \bm b^{(X_d,m_d)}
\end{align*}
which proves~\eqref{eq:c:is:max:no:sigma}.  Furthermore, if $\bm M$ is
the optimal solution to the optimization problem defining
$\bm V_{\bm m}$, then~\eqref{eq:step:1} is an equality, and so
is~\eqref{eq:c:is:max:no:sigma}, proving that $\bm C$ is tight.  This
completes the proof of item~\ref{item:th:main:star:2}.

\subsubsection{Proof of Item~\ref{item:th:main:star:3}}

\noindent\textbf{Item~\ref{item:th:main:star:3}:} \textit{If there exists any solution $\bm C' \in \calM^+_{\bm f^{(\bm X)}, B}$ to Problem~\ref{prob:pessimistic:tensor}, then $\bm C'=\bm C$.}
\\\\
Assume that $\bm C'$ satisfies the assumptions in item~\ref{item:th:main:star:3}.  We will use again one-zero vectors.  On one hand, $\bm C$ is a solution to Problem~\ref{prob:pessimistic:tensor}, and
$\bm C' \in \calM^+_{\bm f^{(\bm X)},B}$, therefore:
  \begin{align}
    \bm C' \cdot \bm b^{(X_1,m_1)} \cdots \bm b^{(X_d,m_d)} \leq & \bm C \cdot \bm b^{(X_1,m_1)} \cdots \bm b^{(X_d,m_d)} \label{eq:step:4}
  \end{align}
  This implies
  $\Sigma_{X_1} \cdots \Sigma_{X_d} \bm C' \leq \Sigma_{X_1} \cdots
  \Sigma_{X_d} \bm C$.  On the other hand, $\bm C'$ is a solution to
  Problem~\ref{prob:pessimistic:tensor}, therefore,
  $\forall \bm M \in \calM^+_{\bm f^{(\bm X)}, B}$:
  \begin{align*}
    \bm M \cdot \bm b^{(X_1,m_1)} \cdots \bm b^{(X_d,m_d)} \leq & \bm C' \cdot \bm b^{(X_1,m_1)} \cdots \bm b^{(X_d,m_d)}
  \end{align*}
  We choose $\bm M$ to maximize the LHS, then the LHS becomes
  $V_{\bm m}$:
  \begin{align}
    V_{\bm m} \leq & \bm C' \cdot \bm b^{(X_1,m_1)} \cdots \bm b^{(X_d,m_d)}
                     = (\Sigma_{X_1}\cdots \Sigma_{X_d} \bm C')_{\bm m} \label{eq:step:5}
  \end{align}
  We have obtained
  $\Sigma_{X_1} \cdots \Sigma_{X_d} \bm C = \bm V \leq \Sigma_{X_1}
  \cdots \Sigma_{X_d} \bm C'$.  The two inequalities imply
  $\Sigma_{X_1}\cdots \Sigma_{X_d} \bm C = \Sigma_{X_1}\cdots
  \Sigma_{X_d} \bm C'$, therefore $\bm C = \bm C'$.  This completes the proof of item~\ref{item:th:main:star:3}.

\subsubsection{Linear Algebra Background}
For the next four items,~\ref{item:th:main:star:6}, ~\ref{item:th:main:star:1}, ~\ref{item:th:main:star:4}, and~\ref{item:th:main:star:5}, we need to take a deeper look at the structure of the tensors $\bm V, \bm C$, through the lens of the dual linear program associated with the primal linear program that defines $\bm V$ in Equation~\eqref{eq:def:v}.  Therefore, we introduce a few definitions and theorems from linear algebra theory.

Recall that a linear program has the form:\\
\noindent \fbox{\parbox{0.4\textwidth}{\footnotesize
\begin{align*}
  \mbox{Maximize:\ } & \bm c^T \bm X \\
\mbox{Where:\ } & \bm A \cdot \bm X \leq \bm b \\
   & \bm X \geq 0
\end{align*}
}}\\\\
The dual has the form:\\\\
\noindent \fbox{\parbox{0.4\textwidth}{\footnotesize
\begin{align*}
  \mbox{Minimize:\ } & \bm Y^T \bm b \\
\mbox{Where:\ } & \bm Y^T \bm A \geq \bm c^T \\
   & \bm Y^T \geq 0
\end{align*}
}}\\
\noindent Intuitively, each variable in the dual can be thought of as representing a constraint in the primal. For example, it is known that when variable in the dual is nonzero, then the constraint that it represents in the primal is tight. More importantly, the dual is important due to the following  \textit{Strong Duality Theorem},
\begin{thm}
\label{thm:strong-duality}
Let $X^*$ and  $\bm Y^*$ be optimal solutions to the primal and dual program, respectively. Then, the following holds,
  \begin{equation*}
      c^TX^* = Y^{*T}b
  \end{equation*}
\end{thm}

\begin{defn}
  A matrix $A$ is totally unimodular iff every square submatrix has a
  determinant equal to $0,+1$ or $-1$.
\end{defn}

\noindent To prove that a matrix is totally unimodular we state the following known fact:
\begin{thm}
\label{thm:totally-unimodular-conditions}
If the rows of a matrix $\bm A$ can be separated into sets $S_1, S_2$ and satisfy the following conditions, then $\bm A$ is totally unimodular,
    \begin{enumerate}
    \item Every entry in $\bm A$ is either $0$ or $1$
    \item Every column in $\bm A$ has at most two non-zero entries.
    \item If there are two non-zero entries in a column of $\bm A$, then one entry is in a row in $S_1$ and the other is in a row in $S_2$.
    \end{enumerate}
\end{thm}

\noindent This unimodularity property is useful because of the following fact,
\begin{thm} \label{thm:totally-unimodular-integral} If the constraint
  matrix $A$ of an LP (linear program) is totally unimodular and
  $\bm b$ is integral, then there exists an optimal solution $\bm X$
  to the primal LP that is integral.  If $\bm c^T$ is also integral,
  then the optimal value is integral.
\end{thm}

\subsubsection{Proof of Item~\ref{item:th:main:star:6}}
\noindent\textbf{Item~\ref{item:th:main:star:6} (Formally):} \textit{\begin{enumerate}
    \item If $d=2$ then $\bm C \geq 0$; if $\bm f^{(X_1)}, \bm f^{(X_2)}, B$ are integral, then $\bm C$ is integral. 
    \item If $d \geq 3$ and $B < \infty$ then $\bm C$ may be $<0$.
\end{enumerate}}
We start by proving the positive statement: When $d=2$ then $\bm C \geq 0$ and, if the degree sequences are integral, then so is $\bm C$.

If $\bm A$ is the matrix of an LP, then the matrix of dual LP is its transpose, $\bm A^T$; one is totally unimodular iff the other is. We will prove that, when $d=2$ then the LP in Eq.~\eqref{eq:def:v} that defines $\bm V$ has a matrix that is totally unimodular: this implies that, if both degree sequences are integral, then $\bm V$ is integral, hence so is $\bm C$ (by Eq.~\eqref{eq:def:c}).  We review here both the primal and the dual LP defining $\bm V$, specialized to the case when $d = 2$.  We denote the two degree sequences by $\bm f \in \R_+^{[n_1]}$ and $\bm g \in \R_+^{[n_2]}$ (rather than $\bm f^{(X_1)}, \bm f^{(X_2)}$).  For every $(p,q) \leq (n_1,n_2)$, the primal and dual LPs are the following:
\\\\
\noindent The primal LP:\\\\
\noindent \fbox{\parbox{0.4\textwidth}{\footnotesize
\begin{align}
  \label{eq:primal:d:2} 
  V_{pq} \defeq \mbox{\ Maximize:\ }
  & \sum_{(i,j) \leq (p,q)}  M_{ij} \\
  \mbox{Where:\ }
  & \forall i\leq p: \sum_{j=1,q} M_{ij} \leq f_i \nonumber \\
  & \forall j\leq q: \sum_{i=1,p} M_{ij} \leq g_j \nonumber \\
  & \forall (i,j)\leq(p,q): M_{ij} \leq B \nonumber
\end{align}
}}\\
\noindent The dual LP:\\\\
\noindent \fbox{\parbox{0.4\textwidth}{\footnotesize
\begin{align}
  \label{eq:dual:d:2}
  V_{pq} \defeq \mbox{\ Minimize:\ }
  & \sum_{i\leq p} \alpha_i f_i + \sum_{j \leq q} \beta_j g_j + B\sum_{(i,j)\leq(p,q)}\mu_{ij}\\
  \mbox{Where:\ }
  & \forall (i,j)\leq(p,q): \alpha_i+\beta_j+\mu_{ij} \geq 1\nonumber
\end{align}
}}

\begin{lmm} \label{lemma:totally:unimodular} The matrix $\bm A$ of the LP defining $\bm V$ in~\eqref{eq:primal:d:2} is totally unimodular. It follows that there exists an integral optimal solution $\bm \alpha, \bm \beta, \bm \gamma$ of the dual LP.  Furthermore, if the two degree sequences $\bm f, \bm g$ are also integral, then $\bm V$ and $\bm C \defeq \Delta_I\Delta_J \bm V$ are integral as well.
\end{lmm}

\begin{proof} Let $\bm A$ be the matrix of the primal  LP~\eqref{eq:primal:d:2}. We show that this matrix is unimodular by showing that it satisfies the conditions of Thm. \ref{thm:totally-unimodular-conditions}. We describe the three important sets of rows in $\bm  A$:
  (1) Rows that arise from row-sum constraints of the form $\sum_jM_{ij} \leq f_i$; let $S_1$ denote this set of rows: $|S_1|=n_1$. 
  (2) Rows that arise from column-sum constraints of the form $\sum_iM_{ij} \leq g_j$; let $S_2$ denote this set of rows,
  $|S_2|=n_2$. 
  (3) Rows that arise from pairwise bounds of the form $M_{ij} \leq B$: there are $n_1n_2$ such rows. This can be visualized as follows where the groups appear in order,
    \begin{align*}
      A =   \begin{pNiceMatrix}
        \ldots & 0 & 1 & 1 & 1 & 0 & 0 & 0 & 0 & \ldots \\
        \ldots & 0 & 0 & 0 & 0 & 1 & 1 & 1 & 0 & \ldots \\
        & & & & \vdots & & & &\\ 
        \hline\\
        \ldots & 0 & 1 & 0 & 0 & 1 & 0 & 0 & 0 & \ldots \\
        \ldots & 0 & 0 & 1 & 0 & 0 & 1 & 0 & 0 & \ldots \\
        & & & & \vdots & & & &  \\ 
        \hline\\
        1 & 0 & 0 & 0 & 0 & 0 & 0 & \ldots & 0 & 0 \\
        & & & & \vdots & & & &\\ 
        0 & 0 & \ldots & 0 & 1 & 0 & 0 & \ldots & 0 & 0\\
        & & & & \vdots & & & &\\ 
        0 & 0 & \ldots & 0 &  0 & 0 & 0 & 0 & 0 & 1 \\
        \end{pNiceMatrix}
    \end{align*}

The third group of rows form the $n_1n_2 \times n_1n_2$ identity matrix, which has 1 on the diagonal and 0 everywhere else, because there is one constraint $M_{ij} \leq B$ for every variable $M_{ij}$.  Consider any square submatrix $\bm A_0$ of $\bm A$.  If it contains any row from the third group of rows, then that row is either entirely 0 (hence $\det \bm A_0 = 0$) or it contains a single value 1, hence the determinant is equal to that obtained by removing the row and column of that value 1.  Therefore, it suffices to prove that the submatrix of $\bm A$ consisting of the first two groups of rows is totally unimodular.

So consider the matrix consisting of the rows in the sets $S_1$ and $S_2$ above. Looking at our conditions from Theorem~\ref{thm:totally-unimodular-conditions}, the first condition holds immediately (all entries in $\bm A$ are 0 or 1). The second condition holds because every $M_{ij}$ occurs exactly in one row and in one column, hence it occurs in exactly two constraints.  Lastly, the last condition holds because the set of rows $S_1$ contain all row-wise constraints for $M_{ij}$ and the rows $S_2$ contain all column-wise constraints for $M_{ij}$, hence the two 1-entries will occur one in a row in $S_1$ and the other in a row in $S_2$.  This proves that the conditions Theorem~\ref{thm:totally-unimodular-conditions} are satisfied, hence $A$ is totally unimodular.
\end{proof}

Next, we will prove that $\bm C \geq 0$, and for that we need a deeper
look at the dual LP~\eqref{eq:dual:d:2}.  Before we do that, we build
some intuition by considering a simple example:

\begin{ex} Consider the $2 \times 2$ problem given by the degree
  sequences $\bm f = \bm g = (2B, 1)$, where $B$ is some large number,
  representing the max tuple multiplicity.  To compute $\bm V$ and
  $\bm C$, we observe that the following are optimal solutions to the
  primal LP~\eqref{eq:primal:d:2} and dual LP~\eqref{eq:dual:d:2}
  respectively:
  \begin{align*}
    \mbox{Primal:\ } && M_{11} = & B & M_{ij}=& 1 &\mbox{\ for $(1,1)<(i,j)<(2,2)$}M_{22}=0\\
    \mbox{Dual:\ } && \alpha_1=\beta_1=& 0 & \alpha_2=\beta_2= & 1 & \mu_{11}=1   \mu_{ij}=0 \mbox{\ for $(1,1)<(i,j)\leq(2,2)$}
  \end{align*}
  To check that these are optimal solutions it suffices to observe
  that the objective function of the primal, $\sum_{ij} M_{ij}=B + 2$,
  is equal to the objective function of the dual,
  $\sum_i \alpha_if_i + \sum_j \beta_j g_j + B \sum_{ij} \mu_{ij} = 1
  + 1 + B = B+2$ Notice that the pessimistic matrix is:
  \begin{align*}
    \bm C = & 
    \begin{pNiceMatrix}[first-row,first-col]
      & 2B & 1 \\
      2B  & B & 1\\
      1   & 1 & 0
      \end{pNiceMatrix}
  \end{align*}
\end{ex}

The proof of $\bm C \geq 0$ follows from several lemmas.  First:
\begin{lmm}
\label{lmm:dual:sorting}
There exists an optimal solution, $\bm \alpha, \bm \beta, \bm \mu$, to
the dual LP~\eqref{eq:dual:d:2} satisfying:
  $$\alpha_i\leq \alpha_{i+1}\quad \beta_j\leq \beta_{j+1}\,\,\forall \,\,i,j$$
\end{lmm}

\begin{proof}
  Suppose that there exists an optimal solution such that
  $\alpha_i>\alpha_{i+1}$. Let $\bm \alpha', \bm \mu'$ be an
  alternative solution where we swap the values at $i$ and $i+1$, i.e.
  $\alpha_i'=\alpha_{i+1},\alpha_{i+1}'=\alpha_{i},
  \mu_{ij}'=\mu_{i+1,j},\mu_{i+1,j}'=\mu_{ij}$. We can quickly check
  that this is a valid solution, i.e. it satisfies the constraints for
  all $j$,
  \begin{align*}
  \alpha_{i+1}+\beta_j+\mu_{i+1,j} \geq 1\,\, &\rightarrow  \,\, \alpha_{i}'+\beta_j+\mu_{i,j}' \geq 1\\
  \alpha_{i}+\beta_j+\mu_{i,j} \geq 1\,\, &\rightarrow  \,\, \alpha_{i+1}'+\beta_j+\mu_{i+1,j}' \geq 1
  \end{align*}
  Further, we argue that the alternative solution produces at least as large an objective value. Because the values of $\mu$ have simply been swapped at $i$ and $i+1$, the sum of the $\mu$ values is the same between the two solutions. Therefore, the difference in their objective values is the following,
  \begin{align*}
  \alpha_if_i + \alpha_{i+1}f_{i+1}-(\alpha_i'f(i)+\alpha_{i+1}f_{i+1})&= (\alpha_i-\alpha_{i+1})f_i+(\alpha_{i+1}-\alpha_{i})f_{i+1}\\
  &=(\alpha_i-\alpha_{i+1})(f_i-f_{i+1})\leq 0     
  \end{align*}
  The final inequality comes from the fact that $f$ is non-increasing and $\alpha_i>\alpha_{i+1}$. Therefore, from any unsorted solution, you can generate an equal or better solution from sorting the $\alpha$'s (the same logic holds identically for the $\beta$'s). This implies that there exists an optimal solution with sorted values, proving the lemma.
\end{proof}

The combination of Lemmas~\ref{lemma:totally:unimodular} and
~\ref{lmm:dual:sorting} implies that the solutions to the dual have
the following form:
\begin{align*}
  \bm \alpha = & (\underbrace{0,\ldots,0}_{s \mbox{ values $=0$}},1,\ldots,1) 
& \bm \beta = & (\underbrace{0,\ldots,0}_{t \mbox{ values $=0$}},1,\ldots,1) \\ \\
  (i,j) \leq & (s,t):\ \mu_{ij} = 1 \ \  (i,j)\not\leq (s,t):\ \mu_{ij}=0
\end{align*}

In other words, both $\alpha,\beta$ are 0-1 vectors with $0 \leq s \leq p$ and $0 \leq t\leq q$ values equal to 0 respectively, followed by the remainder 1's. Note that there is a unique optimal $s,t$ for each sub-matrix problem $V_{p,q}$ which we denote $s_{p,q},t_{p,q}$.  As in Theorem~\ref{th:main:star} item~\ref{item:th:main:star:5} we denote by $\bm F =\Sigma \bm f$ and $\bm G = \Sigma\bm g$ the cumulative degree sequences, i.e. $F_k = \sum_{i \leq k} f_i$ and $G_\ell=\sum_{j\leq \ell} g_j$.
Therefore the expression for $V_{pq}$ becomes:
\begin{equation}
  V_{p,q} = F_p-F_{s_{p,q}}+G_q-G_{t_{p,q}}+s_{p,q}t_{p,q}B
\end{equation}

\begin{lmm} \label{lmm:st:fixed}
If we fix a value of $s_{p,q}$, the optimal value of $t_{p,q}$ is the smallest $t$ such that $g_{t+1}\leq s_{p,q}*B \leq g_{t}$.  Similarly, if we fix $t_{p,q}$, then the optimal value of $s_{p,q}$ is such that $f_{s_{p,q}+1}\leq t_{p,q}*B\leq f_{s_{p,q}}$.
\end{lmm}
\begin{proof}
This can be seen by taking the discrete derivative of the above expression with respect to $t_{p,q}$,
\begin{equation}
  \Delta_{t_{p,q}}V_{p,q} = s_{p,q}B-g_{t_{p,q}}
\end{equation}
This demonstrates that if we set $t$ as before, it is the smallest value of $t$ such that increasing it by $1$ would increase $V_{p,q}$. This directly implies that it is the optimal integral value of $t$. The same logic applies directly to $s_{p,q}$ when $t_{p,q}$ is fixed.
\end{proof}

\begin{lmm}\label{lmm:st:fact}
Using the notation of $s_{p,q},t_{p,q}$ for the optimal solution to the dual for $V_{p,q}$,
\begin{align}
s_{p,q} \in \{s_{p-1,q},p\} \quad t_{p,q} \in \{t_{p,q-1},q\}
\end{align}
\end{lmm}

\begin{proof}
We first note the following. If, $s,s'\leq p-1$, and $t,t'\leq q-1$, then,
$$V_{p-1,q}(s,t)\leq V_{p-1,q}(s',t') \leftrightarrow V_{p,q}(s,t)\leq V_{p,q}(s',t')$$
This is easily checked by noticing $V_{p,q}(s,t)=V_{p-1,q}(s,t) +f_p$ and similarly $V_{p,q}(s',t') = V_{p-1,q}(s',t')+f_p$. This implies that if some $s<p-1$ is optimal for the $(p,q)$ sub-problem, then it must be optimal for the $(p-1,q)$ sub-problem. Therefore, a new optimal value can only come from the new 
\end{proof}
At this point, we list a few useful facts about this structure which are easily checked.
\begin{lmm}\label{lmm:2d:facts}
  The following are facts about $s_{p,q},t_{p,q},V_{p,q}$,
  \begin{enumerate}
      \item $s_{p,q}\geq s_{p-1,q}, t_{p,q}\geq t_{p,q-1}$
      \item $t_{p,q}\leq t_{p-1,q}, s_{p,q}\leq s_{p,q-1}$
      \item $V_{p,q}-V_{p-1,q}\leq f_p, V_{p,q}-V_{p,q-1}\leq g_q$
  \end{enumerate}
\end{lmm}
\begin{proof}
The first item comes directly from applying the previous lemma. The second comes from applying first fact and Lmm. \ref{lmm:st:fixed}. If $s_{p,q}\geq s_{p-1,q}$ and $g$ is non-increasing, then the description of $t$ for a fixed $s$ immediately implies $t_{p,q}\leq t_{p-1,q}$. The third comes from the fact that you could always transition from $p-1,q$ to $p,q$ by setting $\alpha_p=1$, similarly for $q$.
\end{proof}

Okay, we are now prepared to prove that $C>0$ when $d=2$. 
\begin{thm}
    In two dimensions, $C$ will have non-negative entries. This is equivalent to the following two statements,
    \begin{align}
    V_{p,q}-V_{p-1,q}&\geq V_{p,q-1}-V_{p-1,q-1}\\
    V_{p,q}-V_{p,q-1}&\geq V_{p-1,q}-V_{p-1,q-1}
    \end{align}
\end{thm}
\begin{proof}
  We will do this by breaking the problem down into a series of cases based on the values of $s_{p,q}$,$s_{p,q-1}$, and $t_{p,q}$.
  \\
  \textbf{CASE 1}: $s_{p,q-1}=s_{p-1,q-1}$\\
  This directly implies $V_{p,q-1}-V_{p-1,q-1}=f_p$. Further, because $s_{p,q}\leq s_{p,q-1}\leq p-1$, this also implies $V_{p,q}-V_{p-1,q}=f_p$, proving the first inequality.\\\\
  \textbf{CASE 2}: $s_{p,q-1}=p$ and $s_{p,q}=s_{p-1,q}$\\
  From fact 3, we know $V_{p,q-1}-V_{p-1,q-1}\leq f_p$, and from $s_{p,q}=s_{p-1,q}$ we know $V_{p,q}-V_{p-1,q}=f_p$, proving the first inequality.\\\\
  \textbf{CASE 3}: $s_{p,q-1}=s_{p,q}=p$ and $t_{p,q}=t_{p,q-1}$\\ From fact 3, we know $V_{p-1,q}-V_{p-1,q-1}\leq g_q$, and from $t_{p,q}=t_{p,q-1}$ we know $V_{p,q}-V_{p,q-1}=g_q$, proving the second inequality.\\\\
  \textbf{CASE 4}:  $s_{p,q-1}=s_{p,q}=p$ and $t_{p,q}=q$\\
  Because $V_{p,q}=p*q*B$ and $V_{p,q-1}\leq p*(q-1)*B$, we can lower bound the first difference as $V_{p,q}-V_{p,q-1}\geq p\cdot B$. 
  
  Looking at $V_{p-1,q}-V_{p-1,q-1}$, setting $mu_{i,q} = 1, \beta_q=0$ and keeping all other variables consistent with the optimal solution to $V_{p-1,q-1}$ is a valid solution to $V_{p-1,q}$. This solution has value $V_{p-1,q-1}+(p-1)B$, so $V_{p-1,q}\leq V_{p-1,q-1}+(p-1)B$. Therefore, we can upper bound the difference $V_{p-1,q}-V_{p-1,q-1}\leq (p-1)B$.  Therefore, the second inequality is true.\\
  
  Lmm \ref{lmm:st:fact} directly implies that our cases are complete.
\end{proof}

  Now we prove the negative result.

  \begin{ex} \label{ex:negative} Consider $d=3$ dimensions, the degree
    constraints $\bm f=\bm g=(\infty, \infty)$, $\bm h=(\infty, 2B)$
    and the max tuple multiplicity $B>0$.  (Here $\infty$ means ``a
    very large number'', much larger than $B$.)  Define:


\begin{figure}[H]
    \centering
    \includegraphics[width=\textwidth]{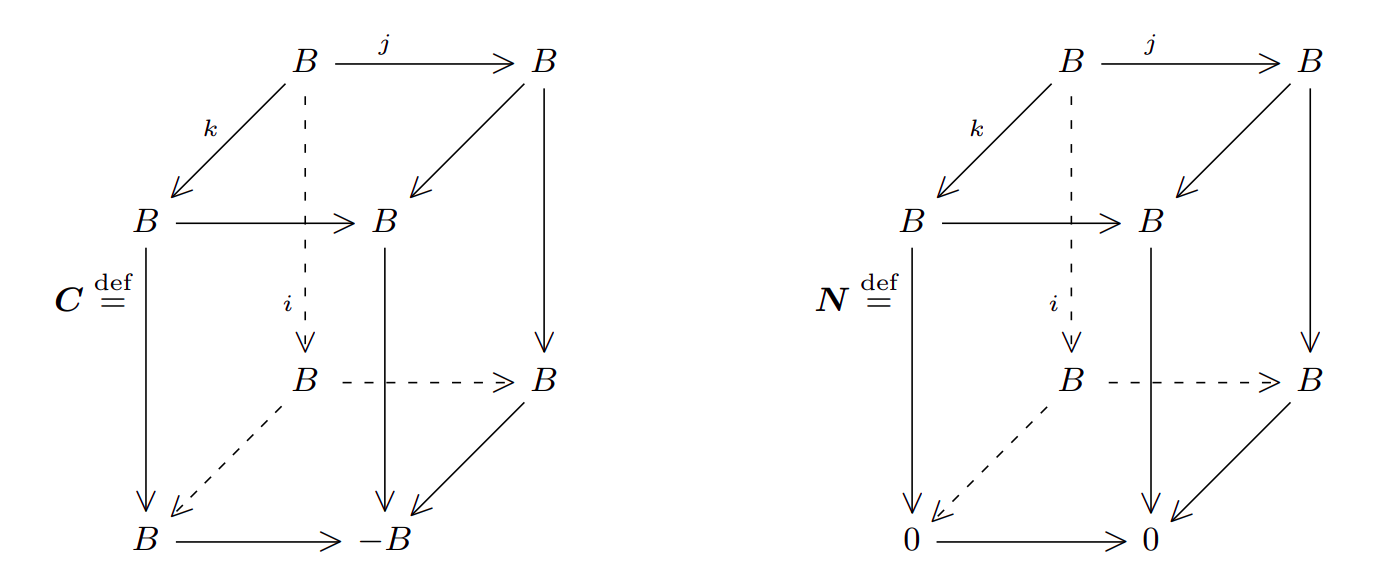}
    \label{fig:negative-example-diagram}
\end{figure}

Here $C_{222}=-B$, proving the claim in item~\ref{item:th:main:star:1}.  It remains to check that $\bm C$ is the pessimistic tensor, which we do as in Example~\ref{ex:greater:than:b}.  Each sub-tensor $m_1 \times m_2 \times m_3$, other than the full tensor $2 \times 2 \times 2$, has all entries $=B$.  Therefore, it is consistent tensor, and the sum of all its elements is a tight upper bound on the sum all elements of any consistent tensor of that size. It remains to check the full tensor $2 \times 2 \times 2$.  Here the sum is $=6B$. Every consistent tensor $\bm M$ has a sum $\leq 6B$,
because:
\begin{align*}
  \sum_{i,j,k = 1,2} M_{ijk} = &   \sum_{i,j = 1,2} M_{ij1} + \sum_{i,j = 1,2} M_{ij2}\\
\leq & 4B \mbox{ (because of the bound $B$) } + 2B \mbox{ (because of $h_2=2B$)}
 = 6B
\end{align*}
The tensor $\bm N$ shown above is a consistent tensor that has the sum
$=6B$.  This proves that $\bm C$ is the pessimistic tensor.
\end{ex}

\subsubsection{Proof of Item~\ref{item:th:main:star:2} (b)}
\label{app:item:1:b}
\noindent\textbf{Item~\ref{item:th:main:star:2}:} \textit{$\bm C$ is a solution to
    Problem~\ref{prob:pessimistic:tensor}, i.e.
    $\bm C \in \calM_{\bm f^{(\bm X)},\infty}$ and it satisfies
    Eq.~\eqref{eq:c:is:max}.  Furthermore, it is tight in the
    following sense: there exists a tensor
    $\bm M \in \calM^+_{\bm f^{(\bm X)},B}$ and non-increasing vectors
    $\bm a^{(p)} \in \R_+^{[n_p]}$, $p=1,d$, such that
    inequality~\eqref{eq:c:is:max} (with $\bm \sigma$ the identity) is
    an equality.}
\\\\

We have already proven part of it in Sec.~\ref{app:item:1:a}.  Here,
it remains to show the second part, that
$\bm C \in \calM_{\bm f^{(\bm X)},\infty}$.  To simplify the
notations, we will restrict the proof to $d=3$ dimensions, both here
and in the proof of items~\ref{item:th:main:star:1},
~\ref{item:th:main:star:4}, and~\ref{item:th:main:star:5}. This is
meant to keep the notations readable and intuitive, and our discussion
generalizes to arbitrary dimensions in a straightforward manner. We
will rename the variables $X_1, X_2, X_3$ to $I,J,K$, and the degree
sequences to
$\bm f \in \R_+^{n_1}, \bm g \in \R_+^{n_2}, \bm h \in \R_+^{n_3}$.

We repeat here the linear program~\eqref{eq:def:v} defining $\bm V$, by specializing it to $d=3$.  For $(m_1,m_2,m_3) \leq (n_1,n_2,n_3)$, the value $V_{m_1m_2m_3}$ is the solution to the following linear program (all indices $i,j,k$ range over $[m_1]$, $[m_2]$, $[m_3]$ respectively):

\noindent \fbox{\parbox{0.4\textwidth}{\footnotesize
\begin{align*}
& V_{m_1m_2m_3} = \mbox{Maximize:\ } \sum_{(i,j,k)\in [m_1]\times [m_2]\times [m_3]} M_{ijk},  \mbox{\ \ Where $M_{ijk}\geq 0$ and:} \\
& \forall i:\ \sum_{jk} M_{ijk} \leq f_i,\ 
\forall j:\ \sum_{ik} M_{ijk} \leq g_j,\ 
\forall k:\ \sum_{ij} M_{ijk} \leq h_k,\ 
\forall ijk:\ M_{ijk} \leq  B
\end{align*}
}}

Equivalently, $V_{m_1m_2m_3}$ is the solution to the dual linear
program:

\noindent \fbox{\parbox{0.4\textwidth}{\footnotesize
\begin{align}
& V_{m_1m_2m_3} = \mbox{Minimize:\ } \sum_i \alpha_if_i + \sum_j \beta_j g_j + \sum_k  \gamma_k h_k + B \sum_{ijk} \mu_{ijk} \label{eq:dual}\\
& \mbox{Where: $\alpha_i,\beta_j,\gamma_k,\mu_{ijk}\geq 0$ and:\ }  \forall i,j,k:\ \alpha_i + \beta_j + \gamma_k + \mu_{ijk} \geq 1\nonumber
\end{align}
}}

In order to prove $\bm C \in \calM_{\bm f^{(\bm X)},\infty}$, we show
a stronger result: any sub-tensor of $\bm C$ is consistent with the
degree sequences $\bm f^{(\bm X)}$.

\begin{lmm} For any $(m_1,m_2,m_3)\leq (n_1,n_2,n_3)$, let $\bm C'$ be the $[m_1]\times [m_2] \times [m_3]$-subtensor of $\bm C$. Then $\bm C'$ is consistent with the degree sequences $\bm f, \bm g, \bm h$.
\end{lmm}

\begin{proof}
  We will show that $\texttt{SUM}_{JK}\left(\bm C'\right) \leq \bm f$,
  which proves that $\bm C$ is consistent with $\bm f$; consistency
  with $\bm g, \bm h$ is proven similarly.  Let $\bm V'$ be the
  $[m_1]\times [m_2] \times [m_3]$-subtensor of $\bm V$:
  \begin{align*}
    \texttt{SUM}_{JK}(\bm C') \stackrel{\eqref{eq:def:c}}{=}& \texttt{SUM}_{JK}(\Delta_I\Delta_J\Delta_K\bm V')=
                                                             \Delta_I\left(\texttt{SUM}_{JK}\left(\Delta_J\Delta_k\bm V'\right)\right)\\
    \stackrel{\eqref{eq:sigma:delta:cancel}}{=} & \Delta_I \left((V_{im_2m_3})_{i=1,m_1}\right) = (V_{im_2m_3}-V_{(i-1)m_2m_3})_{i=1,m_1}
  \end{align*}
  Consider an optimal dual solution
  $\bm \alpha, \bm \beta, \bm \gamma, \bm \mu$ to the
  $(i-1) \times m_2 \times m_3$ problem~\eqref{eq:dual}, which defines the value of $V_{(i-1)m_2m_3}$.  We extended it to a solution to the $i \times m_2 \times m_3$ problem, by setting $\alpha_i \defeq 1$,  $\gamma_{ijk} \defeq 0$ for all $j\in [m_2], k\in [m_3]$.  The expression~\eqref{eq:dual} has increased by $f_i$.  At minimality, the value of $V_{im_2m_3}$ is therefore $\leq V_{(i-1)m_2m_3}+f_i$, proving that $\left(\texttt{SUM}_{JK}(\bm C')\right)_i \leq f_i$.
\end{proof}

\subsubsection{Proof of Item ~\ref{item:th:main:star:1}}

\noindent\textbf{Item ~\ref{item:th:main:star:1}:} \textit{f $B < \infty$, then $\bm C$ may not be consistent with $B$, even if $d=2$.}
\\\\
To prove item~\ref{item:th:main:star:1}, we given an example where $\bm C$ is inconsistent.

\begin{ex} \label{ex:greater:than:b} Consider $d=2$ dimensions.  Let
  the degree sequences be $\bm f = \bm g = (2B, 2B, B+1)$, where $B$
  is some large number representing the max tuple multiplicity.  Consider the two
  matrices below:
\begin{small}
  \begin{align*}
    \mathbf{C} = 
    &
      \begin{pNiceMatrix}[first-row,first-col]
           & 2B & 2B & B+1 \\
        2B  & B & B & 0 \\
        2B  & B & B & 0 \\
        B+1 & 0 & 0 & B+1
      \end{pNiceMatrix}
& 
    \mathbf{N} = 
    &
      \begin{pNiceMatrix}[first-row,first-col]
           & 2B & 2B & B+1 \\
        2B  & B & B & 0 \\
        2B  & B & B-1 & 1 \\
        B+1 & 0 & 1 & B
      \end{pNiceMatrix}
  \end{align*}
\end{small}

  We claim $\bm C$ is the worst-case matrix for the given
  $\bm f, \bm g, B$: since $C_{33}>B$, this represents an example of a
  worst-case matrix that is inconsistent.  It suffices to check that
  $\sum_{i\leq m_1, j \leq m_2} C_{ij} = V_{m_1m_2}$ for all
  $m_1, m_2 \leq 3$ (see Eq.~\eqref{eq:step:3}), i.e. we will show
  that the sum of the $m_1 \times m_2$ submatrix of $\bm C$ is the
  maximum value of the sum of any consistent $m_1 \times m_2$
  submatrix $\bm M$.  One can check directly that
  $\sum_{i\leq m_1, j\leq m_2} C_{ij} = \min(\sum_{i\leq m_1}f_i,
  \sum_{j\leq m_2}g_j)$, which proves that any consistent $\bm M$ has
  $\sum_{ij} M_{ij} \leq \sum_{ij} C_{ij}$.  To prove equality, we
  notice that, when $(m_1,m_2)\neq (3,3)$, the submatrix of $\bm C$ is
  consistent, hence we have equality.  For $m_1=m_2=3$, the full
  matrix $\bm C$ is not consistent (because $C_{33}>B$), but the
  matrix $\bm N$ shown above is consistent and has the same sum as
  $\bm C$, namely $5B+1$.
\end{ex}

\subsubsection{Proof of Item~\ref{item:th:main:star:4}}

\noindent\textbf{Item~\ref{item:th:main:star:4}:} \textit{For any non-increasing vectors $\bm a^{(X_p)} \in \R_+^{[n_p]}$, $p=2,d$, the vector $\bm C \cdot \bm a^{(X_2)} \cdots \bm a^{(X_d)}$ is in $\R^{[n_1]}_+$ and non-increasing.}

Recall that we continue to restrict the discussion to $d=3$, in order to simplify the notation and reduce clutter.  The general case is similar, and omitted.  The proof of this item requires two lemmas.
First:

\begin{lmm} \label{lemma:by:parts}
  For any vectors $\bm x, \bm y \in \R^n$, the following
  ``summation-by-parts'' holds:
  \begin{align}
    \sum_{i=1,n} (\Delta \bm x)_i\cdot  y_i = & x_n y_n - \sum_{i=1,n-1} x_i\cdot(\Delta \bm y)_{i+1}  \label{eq:sum:by:parts}
  \end{align}
\end{lmm}

The formula is analogous with integration by parts:
$\int_0^n f' g dx = fg\mid_0^n - \int_0^n f g' dx$.  The proof follows
immediately by expanding and rearranging the LHS as follows:
\begin{align*}
  \sum_{i=1,n} (\Delta \bm x)_i \cdot y_i = & \sum_{i=1,n} x_i y_i - \sum_{i=1,n} x_{i-1} y_i 
\\=& \sum_{i=1,n} x_i y_i - \sum_{i=1,n-1} x_i y_{i+1}
\\=& x_n y_n - \left(\sum_{i=1,n-1} x_i y_{i+1}-\sum_{i=1,n-1} x_i y_i\right)
\\=& x_n y_n - \sum_{i=1,n-1} x_i (\Delta \bm y)_{i+1}
\end{align*}

The second lemma assets that the second discrete derivative of $\bm V$
on any variable is always $\leq 0$.  Recall that $\bm V$ is the {\em
  value tensor}, introduced in Def.~\ref{def:pessimiistic:tensor}, and
we have described it in terms of a primal linear program, and a dual
linear program~\eqref{eq:dual}.

\begin{lmm} \label{lemma:delta:2:negative}
$\Delta^2_I \bm V \leq 0$,  $\Delta^2_J \bm V \leq 0$, and $\Delta^2_K \bm V \leq 0$. 
\end{lmm}

\begin{proof}
  $\Delta^2_K \bm V \leq 0$ is equivalent to the statement:
  \begin{align}
\forall (m_1, m_2, k) \leq (n_1,n_2,n_3) &&   2 V_{m_1m_2(k-1)} \geq & V_{m_1m_2(k-2)}+ V_{m_1m_2k}\label{eq:concave}
  \end{align}
  Consider an optimal dual solution
  $\bm \alpha, \bm \beta, \bm \gamma, \bm \mu$ to the
  $m_1 \times m_2 \times (k-1)$ problem.  Then $V_{m_1m_2(k-1)}$ is
  given by the objective function of the dual linear program, i.e. by
  expression~\eqref{eq:dual}.  Denote by
  $w \defeq \alpha_{k-1}f_{k-1} + B \sum_{ij} \mu_{ij(k-1)}$, i.e. the
  contribution to~\eqref{eq:dual} of all terms in the plane $K = k-1$.
  On one hand, we derive from this solution a solution to the
  $m_1 \times m_2 \times (k-2)$ problem by simply removing the
  variables in the plane $K=k-1$: expression~\eqref{eq:dual} decreases
  by the quantity $w$, which proves
  $V_{m_1m_2(k-2)} \leq V_{m_1m_2(k-1)}-w$.  On the other hand we can
  extend the solution to a solution of the $m_1 \times m_2 \times k$
  problem, by setting the variables in the plane $K=k$ to those in the
  plane $K=k-1$, more precisely $\alpha_k \defeq \alpha_{k-1}$,
  $\mu_{ijk} \defeq \mu_{ij(k-1)}$: it is obvious that all new
  constraints are satisfied, and the expression~\eqref{eq:dual} has
  increased by $w$.  This implies
  $V_{m_1m_2k} \leq V_{m_1m_2(k-1)}+w$.  By adding these two
  inequalities we prove~\eqref{eq:concave}, as required.
\end{proof}

To prove item~\ref{item:th:main:star:4}, let
$\bm a \in \R_+^{n_1}, \bm b \in \R_+^{n_2}$ be non-increasing, and
let $\bm v \defeq \bm C \cdot \bm a \cdot \bm b$.  We need to prove
that $\bm v \geq 0$ and that $\bm v$ is non-increasing.  We start by
proving that $\bm v \geq 0$ and for that it suffices to assume that
both $\bm a$ and $\bm b$ are one-zero vectors, because any
non-decreasing vector is a non-negative linear combination of one-zero
vectors (see Sec.~\ref{sec:bowtie}).  Suppose the first $m_1$ values
in $\bm a$ are 1 and the others are 0, and similarly the first $m_2$
values in $\bm b$ are 1 and the others are 0.  Then, using an argument
similar to that in Eq.~\eqref{eq:step:2} we have
$(\bm C \cdot \bm a \cdot \bm b)_k = \sum_{(i,j) \leq (m_1,m_2)}
C_{ijk}$ and therefore:
\begin{align*}
v_k = & (\bm C \cdot \bm a \cdot \bm b)_k = (\Sigma_I \Sigma_J \bm C)_{m_1  m_2 k} = (\Sigma_I \Sigma_J
        \Delta_I \Delta_J \Delta_K \bm V)_{m_1  m_2 k} = (\Delta_K \bm
        V)_{m_1 m_2 k} = V_{m_1 m_2 k} - V_{m_1 m_2 (k-1)}
\end{align*}
The value $V_{m_1 m_2 (k-1)}$ is defined by a
$m_1 \times m_2 \times (k-1)$ linear program in Eq.~\eqref{eq:def:v}.
Let $\bm M$ be the optimal solution, such that
$V_{m_1 m_2 (k-1)} = \sum_{(i,j,p) \leq (m_1,m_2,k-1)} M_{ijp}$.  We
define a new tensor $\bm M'$ obtained from $\bm M$ by setting
$M'_{ijk} \defeq 0$ and setting $M'_{ijp}=M_{ijp}$ when $p \neq k$.
Then $\bm M'$ is a feasible solution to the
$m_1 \times m_2 \times k$-linear program Eq.~\eqref{eq:def:v}, proving
that, at optimality,
$V_{m_1 m_2 k} \geq \sum_{(i,j,p) \leq (m_1,m_2,k)}M'_{ijp} =
V_{m_1m_2(k-1)}$, proving that $v_k \geq 0$.

Next we prove that $\bm v$ is non-increasing, or, equivalently, that
$\Delta_K \bm v$ is $\leq 0$.  By definition,
$v_k \defeq (\bm C \cdot \bm a \cdot \bm b)_k = \sum_{ij} C_{ijk} a_i
b_j$.  Since $\bm C=\Delta_I\Delta_J\Delta_K \bm V$, the summation-by
parts formula~\eqref{eq:sum:by:parts} allows us to move the
derivatives $\Delta_I\Delta_J$ from $\bm V$ to $\bm a, \bm b$ and
obtain the following:
\begin{align}
  \Delta_K \left(\sum_{ij} C_{ijk} a_i b_j\right) \stackrel{\eqref{eq:def:c}}{=}  & \Delta_K\left(\sum_{ij} \left(\Delta_I\Delta_J\Delta_KV_{ijk}\right) a_i b_j\right)  = \Delta_K^2 \left(\sum_j \left(\Delta_J\sum_i (\Delta_I V_{ijk})a_i\right)b_j \right)\nonumber\\
\sum_j \left(\Delta_J\sum_i (\Delta_I V_{ijk})a_i\right)b_j \stackrel{\eqref{eq:sum:by:parts}}{=} & \left(\sum_j \left(\Delta_J \left(V_{n_1jk}a_{n_1}-\sum_iV_{ijk}(\Delta_I \bm a)_{i+1}\right)\right)b_j\right)\nonumber\\
= & \left(\sum_j \left(\Delta_J \left(V_{n_1jk}a_{n_1}\right)\right)b_j\right)-\left(\sum_j \left(\Delta_J \left(\sum_iV_{ijk}(\Delta_I \bm a)_{i+1}\right)\right)b_j\right)\nonumber\\
\stackrel{\eqref{eq:sum:by:parts}}{=}&\left(V_{n_1n_2k}a_{n_1}b_{n_2}\right) - \left(\sum_j V_{n_1jk} a_{n_1} (\Delta_J \bm b)_{j+1}\right)
-   \left(\sum_i V_{in_2k}(\Delta_I \bm a)_{i+1}b_{n_2}\right)+\\&\quad\quad  \left(\sum_{ij} V_{ijk}(\Delta_I \bm a)_{i+1}(\Delta_J \bm b)_{j+1}\right)\label{eq:by:parts:3}
\end{align}
(In general, for $d \geq 3$, this expression is a sum of $2^{d-1}$
terms.)  Therefore:
\begin{align}
    \Delta_K \left(\sum_{ij} C_{ijk} a_i b_j\right) =
&
\left(\Delta_K^2(V_{n_1n_2k})a_{n_1}b_{n_2}\right) - \left(\sum_j \Delta_K^2(V_{n_1jk}) a_{n_1} (\Delta_J \bm b)_{j+1}\right)
\\&-   \left(\sum_i \Delta_K^2(V_{in_2k})(\Delta_I \bm a)_{i+1}b_{n_2}\right)
+   \left(\sum_{ij} \Delta_K^2(V_{ijk})(\Delta_I \bm a)_{i+1}(\Delta_J \bm b)_{j+1}\right)
\label{eq:monotonicity}
\end{align}

Since $\bm a, \bm b$ are non-increasing, we have
$\Delta_I \bm a \leq 0$ and $\Delta_J \bm b \leq 0$.  By
Lemma~\ref{lemma:delta:2:negative} each term $\Delta_K^2 V$ is
$\leq 0$, which proves that the expression~\eqref{eq:monotonicity} is
$\leq 0$. 

\subsubsection{Proof of Item~\ref{item:th:main:star:5}}

Finally, we prove item~\ref{item:th:main:star:5}, and continue to restrict the discussion to $d=3$ dimensions. 
\\\\
\noindent\textbf{Item ~\ref{item:th:main:star:5}:} \textit{ Assume $B = \infty$.  Then the following holds:
    \begin{align}
      \forall \bm m \in [\bm n]: &&    V_{\bm m} = & \min\left(F^{(X_1)}_{m_1}, \ldots, F^{(X_d)}_{m_d}\right) 
    \end{align}
    where $F^{(X_p)}_r \defeq \sum_{j\leq r} f^{(X_p)}_j$ is the CDF
    associated to the PDF $\bm f^{(X_p)}$, for $p=1,d$. Moreover,
    $\bm C$ can be computed by Algorithm~\ref{alg:fast_C_alg}, which
    runs in time $\mathbf{O}(\sum_p n_p)$.  This further implies that
    $\bm C \geq 0$, in other words
    $\bm C \in \calM^+_{\bm f^{(\bm X)}, \infty}$.}
\\\\
We start by showing Eq.~\eqref{eq:v:is:max}: for all $(m_1, m_2, m_3) \leq (n_1,n_2,n_3)$,  $V_{m_1m_2m_3} = \min(F_{m_1},G_{m_2},H_{m_3})$.  Assume w.l.o.g. that $F_{m_1} = \min(F_{m_1},G_{m_2},H_{m_3})$.  We claim that the following is an optimal solution to the dual program~\eqref{eq:dual}: $\bm \alpha = (1,1,\cdots,1)$, $\bm \beta = \bm 0$, $\bm \gamma = \bm 0$, $\bm \mu = \bm 0$: the claim implies that the value of $V_{m_1m_2m_3}$ given by Eq.~\eqref{eq:dual} is $V_{m_1m_2m_3}=\sum_{i=1,m_1} \alpha_i f_i = F_{m_1}$, as required. To prove the claim, let $\bm \alpha, \bm \beta, \bm \gamma$ be any
optimal solution to the dual; obviously, $\bm \mu = \bm 0$, because
$B=\infty$.  We first find another optimal solution where
$\bm \gamma = \bm 0$.  Denoting
$\varepsilon \defeq \min_{ij} (\alpha_i + \beta_j)$, we have
$\gamma_k \geq 1 - \varepsilon$ for all $k$.  Define
$\alpha_i' \defeq \alpha_i + (1-\varepsilon)$ for all $i$, and
$\gamma_k' \defeq 0$.  Then $\bm \alpha', \bm \beta, \bm \gamma'$ is
also a feasible solution, and the value of~\eqref{eq:dual} is no
larger, because:
$\sum_i \alpha_i' f_i = \sum_i \alpha_i f_i + (1-\varepsilon) \sum_i
f_i \leq \sum_i f_i \alpha_i + (1-\varepsilon) \sum_k h_k \leq \sum_i
f_i \alpha_i + \sum_k \gamma_k h_k$.  We repeat this for $\beta$ (and
omit the details, for lack of space)
and obtain a new feasible solution
$\bm \alpha'', \bm \beta' (=0), \bm \gamma' (=0)$ where the value of
Eq.~\eqref{eq:dual} is still optimal, proving the claim.

It remains to prove that Algorithm~\ref{alg:fast_C_alg} is correct.
We show that if $\bm f, \bm g, \bm h$ are any 3 vectors (not
necessarily decreasing) then the tensor $\bm C$ computed by the
algorithm satisfies
$\Sigma_I \Sigma_J \Sigma_K \bm C = \min(\bm F, \bm G, \bm H)$,
i.e. $\sum_{i=1,m_1}\sum_{j=1,m_2}\sum_{k=1,m_3}C_{ijk} = \min(F_{m_1},
G_{m_2}, H_{m_3})$ (correctness follows from~\eqref{eq:step:1}).
Assume w.l.o.g. that $f_1 = \min(f_1, g_1, h_1)$; then the algorithm
sets $C_{111}=f_1$ and leaves $C_{1m_2m_3}=0$ for all
$(m_2,m_3)\neq (1,1)$ (since $\bm C$ is initially 0).  Therefore the
claim holds for $m_1=1$: $\sum_{j=1,m_2}\sum_{k=1,m_3}C_{1ij}=f_1$.
For $m_1 > 1$, denote by $\bm C'$ the rest of the tensor produced by
the algorithm, i.e. using $\bm f' \defeq (f_2,f_3, \ldots, f_{n_1})$,
$\bm g' \defeq (g_1-f_1, g_2, \ldots, g_{n_2})$ and
$\bm h' \defeq (h_1-f_1, h_2, \ldots, h_{n_3})$. Then:
\begin{align*}
\sum_{i=1,m_1}&\sum_{j=1,m_2}\sum_{k=1,m_3}C_{ijk} = f_1 + \sum_{i=2,m_1}\sum_{j=1,m_2}\sum_{k=1,m_3}C'_{ijk} \\
= & f_1 + \min(F'_{m_1}, G'_{m_2}, H'_{m_3})\mbox{\ \ \ \ \ \ \ \ \ /*  by induction on $\bm C'$ */}\\
= & \min(f_1+F'_{m_1}, f_1+G'_{m_2},f_1+H'_{m_3})=\min(F_{m_1},G_{m_2},H_{m_3})
\end{align*}
because $f_1+ F'_{m_1} = f_1+\sum_{i=2,m_1}f_i = F_{m_1}$,
$f_1 + G'_{m_2}=f_1 + (g_1-f_1)+\sum_{j=2,m_2}g_j = G_{m_2}$ and
similarly $f_1+H'_{m_3}=H_{m_3}$.

%% file: chain-bound.tex
\subsection{Chain Bound for Berge Acyclic Queries (Theorem~\ref{thm:polyamtroid:bound:sum:of:chains})}

\label{app:chain:bound}

In this section we review the polymatroid bound, then prove its closed
form expression for Berge-acyclic queries stated in
Theorem~\ref{thm:polyamtroid:bound:sum:of:chains}.

{\bf Background on Information Theory} We need to review some notions
from information theory.  Fix a set of variables $\bm X$.  A {\em
  polymatroid} is a function $h: 2^{\bm X} \rightarrow \R_+$
satisfying the following three conditions, where $X, Y$ are sets of
variables.

\begin{align*}
    h(\emptyset) = & 0 \\
    h(X\cup Y) \geq & h(X) \\
    h(X) + h(Y) \geq & h(X \cup Y) + h(X \cap Y)
\end{align*}

The last two conditions are called {\em monotonicity} and {\em
  submodularity}.  We follow standard convention and write $XY$ for
the union $X\cup Y$. A {\em conditional expression} is defined
\begin{align*}
    h(Y|X) \defeq & h(XY) - h(X)
\end{align*}

Given a joint distribution over random variables $\bm X$, we denote by
$H$ its entropy.  Then the mapping $Z (\subseteq \bm X) \mapsto H(Z)$
is a polymatroid.  When $\bm X$ has 4 or more variables, then there exists polymatroids $h$ that are not entropic functions.

{\bf The polymatroid Bound} The polymatroid bound was introduced
in~\cite{DBLP:conf/pods/Khamis0S17,DBLP:journals/corr/Khamis0S16}.  A
good overview can be found in~\cite{DBLP:conf/pods/000118}.  We review
it here briefly.

Fix a general conjunctive query $Q$, not necessarily Berge-acyclic.
Recall that $N_R$ and $f_1^{(R,Z)}$ denote the upper bound on the
cardinality of $R$ and of the maximum degree of its variable $Z$.  The
original definition also considers max degrees for sets of variables,
but in this paper we restrict to a single variable.  Denote by $HDC$
(``degree constraints applied to $h$'') the following set of
constraints on a polymatroid $h$:
\begin{align}
  HDC: && \forall R \in \bm R(Q):\ \  h(\bm X_R) \leq & \log N_R  & \forall Z \in \bm X_R:\ \  h(\bm X_R|Z) \leq & \log f_1^{(R,Z)} \label{eq:constraints:polymatroid}
\end{align}

The polymatroid bound of a query $Q$, which we denote here by $PB(Q)$,
is defined in~\cite[Theorem 1.5]{DBLP:journals/corr/Khamis0S16} by
specifying its logarithm:
\begin{align}
  \log(PB(Q)) \defeq & \max \setof{h(\bm X)}{h \mbox{ is a polymatroid that satisfies $HDC$}}\label{eq:polymatroid:bound}
\end{align}

More precisely, $PB(Q)$ is the exponent of the RHS of the expression
above.  Since $HDC$, the monontonicity, and the submodularity
constraints are all linear inequalities, it follows that $\log(PB(Q))$
can be computed by solving a linear optimization problem.  When the
query has 4 or more variables, then the polymatroid bound is not tight
in general.

{\bf Linear Inequalities} An alternative way to describe the
polymatroid bound uses {\em linear inequalities}, which we describe
next.

Fix a vector $\bm c \in \R^{2^{\bm X}}$.  We view the vector as
defining a linear expression $E$ over polymatroids, namely:
\begin{align}
  E(h) \defeq & \sum_{V \subseteq \bm X}c_V h(V) \label{eq:e:linear:expression}
\end{align}
We will assume that $\bm c$ has integer coefficients, and consider
linear inequalities of the form:
\begin{align}
  E(h) \geq & 0 \label{eq:h:inequality}
\end{align}
We say that the inequality~\eqref{eq:h:inequality} is {\em valid for
  all polymatroids}, or {\em valid for all entropic functions}, if it
holds for all polymatroids $h$, or for all entropic functions $H$
respectively.

{\bf The Polymatroid Bound through Inequalities} Given a conjunctive
query $Q$, we say that a linear expression $E(h)$ is {\em  associated}
to the query  $Q$ if it has the form:
\begin{align}
  E(h) \defeq & \sum_{R\in \bm R(Q)} a_R h(\bm X_R) + \sum_{R\in \bm R(Q), Z \in \bm X_R}b_{R,Z} h(\bm X_R | Z)\label{eq:h:expression:q}
\end{align}
where all coefficients $a_R, b_{R,Z}, k$ are natural numbers.  In
other words the linear expression contains only the terms $h(\bm X_R)$
and $h(\bm X_R|Z)$ that are also used in the HDC constraints. If the
inequality $E(h) \geq k\cdot h(\bm X)$ is valid for all polymatroids
$h$, then the following holds for the polymatroid bound $PB(Q)$:
\begin{align}
  k \cdot \log PB(Q) \leq & \sum_R a_R \log N_R + \sum_{R, Z \in \bm  X_R} b_{R,Z} \log(f_1^{(R,Z)})  \label{eq:polymatroid:bound:inequality}
\end{align}
By applying Farkas' lemma one can prove that there exists a valid
inequality for which~\eqref{eq:polymatroid:bound:inequality} becomes
an inequality.  In other words, one could compute $PB(Q)$ by trying
out all (infinitely many) linear expressions $E(h)$ associated to $Q$,
trying all numbers $k$, then take the minimum of the
expression~\eqref{eq:polymatroid:bound:inequality} where the
inequality $E(h) \geq k\cdot h(\bm X)$ is valid.

\begin{ex} \label{eq:chain:bound:2} We continue with the query
  $Q=R(X,Y),S(Y,Z),T(Z,U),K(U,V)$ in Example~\ref{eq:chain:bound}. The
  following are valid inequalities associated to $Q$:
  \begin{align*}
    & h(YZ)+h(X|Y)+h(U|Z)+h(V|U) \geq  h(XYZU) \\
    & h(XY)+h(ZU)+h(V|U) \geq h(XYZU) \\
    & h(XY)+h(UV)+h(Z|U) \geq h(XYZU) \\
    & h(XY)+h(ZU)+h(UV) \geq h(XYZU) \\
    & 3h(XY)+h(YZ)+h(ZU)+h(UV)+h(X|Y)+h(U|Z)+h(Z|U)+2h(V|U) \geq  4h(XYZU)
  \end{align*}
  All five are valid inequalities for polymatroids.  The first four
  are easy to check directly, and the fifth is simply the sum of the
  first four (and thus is also valid).  These inequalities lead to
  the following upper bounds on $PB(Q)$:
  \begin{align*}
    PB(Q) \leq & f_1^{(R,Y)}N_S f_1^{(T,Z)}f_1^{(K,U)}\\
    PB(Q) \leq & N_R \cdot N_T f^{(K,U)} \\
    PB(Q) \leq & N_R \cdot f^{(T,U)}N_k \\
    PB(Q) \leq & N_R \cdot N_T \cdot N_K \\
    PB(Q) \leq & \left(N_R^3 N_S N_TN_Uf_1^{(R,Y)}f_1^{(T,Z)}f_1^{(T,U)}\left(f_1^{(K,U)}\right)^2\right)^{1/4}
  \end{align*}
  Notice that the last expression is the geometric mean of the first
  four.
\end{ex}

{\bf Term Cover} We describe a simple necessary (but not sufficient)
condition for an inequality to be valid for all entropic functions.
We say that a term $h(V)$ in $E(h)$ {\em covers} a set of variables
$U \subseteq \bm X$, if $U \cap V \neq \emptyset$.  For any linear
expression $E(h) = \sum_V c_Vh(V)$ define the coverage of $U$ by $E$
as:
\begin{align*}
  \texttt{cover}(U,E) \defeq &  \sum_{V: V \cap U \neq \emptyset} c_V
\end{align*}
Once can check that, if $E(h) \geq 0$ holds for all entropic
functions, then $\forall U \subseteq \bm X$, if $U \neq \emptyset$,
then $\texttt{cover}(U,E) \geq 0$, by observing that
$\texttt{cover}(U,E)$ is the value of $E(h)$ computed on the following
polymatroid, called a {\em step function}: $h(V) \defeq 0$ if
$V \cap U = \emptyset$ and $h(V) \defeq 1$ otherwise.  In fact, $h$ is
an entropy, namely it is the entropy of the following probability
space:

\begin{tabular}{|c|c|l} \cline{1-2}
  $U$ & $\bm X-U$ & \\ \cline{1-2}
  $0\cdots 0$ & $0 \cdots 0$ & $1/2$ \\
  $1\cdots 1$ & $0 \cdots 0$ & $1/2$ \\ \cline{1-2}
\end{tabular}

{\bf Simple Inequalities} While a characterization of general
information-theoretic inequalities is still open, a complete
characterization was described in~\cite{DBLP:conf/pods/KhamisK0S20}
for a restricted class.  We present here a slightly more general
version than stated in Theorem 3.9
(ii)~\cite{DBLP:conf/pods/KhamisK0S20}, which follows immediately by using
the same Lemma 3.10 as done in~\cite{DBLP:conf/pods/KhamisK0S20}.

A {\em simple expression} is a linear expression of the form:
\begin{align}
  E(h) \defeq &\sum_{V\subseteq \bm X} a_Vh(V) - \sum_{Z \in \bm X}b_Z h(Z) \label{eq:h:expression:simple}
\end{align}
where all coefficients are natural numbers.  In other words, $E$ must
have only non-negative coefficients, except for singleton terms
$h(Z)$, where $Z$ is a single variable, which may be negative.  A {\em
  simple inequality} is an inequality of the form:
\begin{align}
  E(h)  \geq  & k \cdot h(\bm X) \label{eq:h:inequality:simple}
\end{align}
where $E$ is a simple inequality.  The following was proven
in~\cite{DBLP:conf/pods/KhamisK0S20}.

\begin{thm} \label{th:simple:inequalities} Consider a simple
  expression $E(h)$, as in Eq.~\eqref{eq:h:expression:simple}.Then the
  following are equivalent:
  \begin{enumerate}
  \item The inequality $E(h) \geq k h(\bm X)$ is valid for all polymatroids.
  \item The inequality $E(h) \geq k h(\bm X)$ is valid for all entropic functions.
  \item For every set of variables $U \neq \emptyset$,
    $\texttt{cover(U,E)} \geq k$.
  \end{enumerate}
\end{thm}

In other words, the sufficient condition mentioned earlier becomes a
necessary condition, in the case of simple inequalities.

We are interested in simple expression of a special form.  A {\em
  conditional simple expression} is:
\begin{align}
  E(h) \defeq &\sum_{U,Y\subseteq \bm X} c_{V|Y} h(V|Y) \label{eq:h:expression:simple:conditional}
\end{align}
where $c_{V|Y} \geq 0$ and whenever $c_{V|Y} > 0$ then $Y$ is either
$\emptyset$ or a single variable.   We prove:

\begin{lmm} \label{lemma:matching:chain:bound} Consider a simple expressions $E$
  as in~\eqref{eq:h:expression:simple}, and suppose that for every set
  of variables $U$, $\texttt{cover}(U,E) \geq 0$.  Then $E$ can be
  written (not uniquely) as a conditional simple expression, i.e. as
  in~\eqref{eq:h:expression:simple:conditional}.
\end{lmm}

Every linear expression $E$ associated to a Berge-acyclic query, as in
Equation~\eqref{eq:h:expression:q} is a conditional simple expression
and, therefore, it is also a simple expression, once we expand the
conditionals.  The lemma provides a simple test for the converse to
hold: take a simple expression and regroup it so it becomes associated
to a query.  For example, $h(XY) + h(YZ) + h(ZU) - h(Y) - h(Z)$ can be
rearranged as $h(X|Y) + h(YZ) + h(U|Z)$, and is associated to a
Berge-acyclic query $R(X,Y),S(Y,Z),T(Z,U)$, but the expression
$E = h(XY)+h(YZ) + h(UV) - h(X) - h(Y) - h(Z)$ cannot, because
$\texttt{cover}(XYZ,E) = -1$.

\begin{proof} (of Lemma~\ref{lemma:matching:chain:bound})
  Recall that the coefficients $a_V, b_Z$ of $E$ are natural numbers.
  Consider the following bipartite graph
  $G=(\text{Nodes}_1, \text{Nodes}_2, \text{Edges})$.  For the set
  $\text{Nodes}_1$ we create, for $Z \in \bm X$, $b_Z$ nodes labeled
  $h(Z)$.  For the set $\text{Nodes}_2$ we create, for each
  $V \subseteq \bm X$, $a_V$ nodes labeled $h(V)$.  Finally, we
  connect two nodes $h(Z)$, $h(V)$ if $Z \in V$.  Next, we use Hall's
  theorem to prove that this bipartite graph has a maximal matching,
  and, for that, it suffices to check that, for each subset
  $S \subseteq \text{Nodes}_1$, $|N(S)| \geq |S|$, where $N(S)$ are
  the neighbors of the nodes in $S$.  If $S$ includes some node
  labeled $h(Z)$, then we can assume w.l.o.g. that it includes all
  $b_Z$ copies, otherwise we simply add them, which only increases $S$
  but does not affect $N(S)$.  Therefore, $S$ is uniquely defined by a
  set of variables $U \subseteq \bm X$, and
  $|S| = \sum_{Z \in U} b_Z$, while
  $|N(S)|=\sum_{V \subseteq \bm X, U\cap V \neq \emptyset}a_V$.  It
  follows that $|N(S)|-|S| = \texttt{cover}(U,E) \geq 0$.  By Hall's
  theorem, the graph admits a maximal matching, in other words every
  node $h(Z)$ is uniquely matched with some $h(V)$ where $Z \in V$,
  hence we combine the two terms $h(V)-h(Z)=h(V|Z)$.  The lemma
  follows from here.
\end{proof}


{\bf Chain Expressions} We define a {\em connected, simple chain expression}
an expression of the form:
\begin{align}
  E(h) \defeq &   h(U_0)+h(U_1|Z_1)+ \cdots + h(U_k|Z_k) \label{eq:chain:expression}
\end{align}
where $U_0, U_1, \ldots, U_k$ form a partition of the variables
$\bm X$, for all $i=1,k$, $Z_i$ is a single variable, and
$Z_i \in U_0 \cup U_1 \cup \ldots \cup U_{i-1}$.  

A chain expression is the simplest linear expression we can associate
to a Berge-acyclic query $Q$.  Namely, choose arbitrarily a root
relation $\texttt{ROOT} \in \bm R(Q)$, orient its tree and, as in
Section~\ref{sec:connection:to:agm:polymatroid:bounds}, denote by
$Z_R \in \bm X_R$ the parent variable of the relation $R$,
i.e. connecting $R$ to $\texttt{ROOT}$.  The {\em connected, simple
  chain expression associated to $Q$} is:
\begin{align}
  E_Q(h) \defeq & h(\bm X_{\texttt{ROOT}}) + \sum_{R \in \bm R(Q), R \neq \texttt{ROOT}} h(\bm X_R-\set{Z_R} | Z_R)
\label{eq:simple:expression:berge}
\end{align}
The inequality $E_Q(h) \geq h(\bm X_R)$ is valid, and defines the
following bound  on $PB(Q)$ (see
Equation~\eqref{eq:polymatroid:bound:inequality}):
\begin{align*}
  PB(Q) \leq & N_{\texttt{ROOT}} \cdot \prod_{R \neq   \texttt{ROOT}}f_1^{(R,Z_R)} = PB(Q,\texttt{ROOT})
\end{align*}
where is precisely the expression $PB(Q,\texttt{ROOT})$ defined in
Sec.~\ref{sec:general}.  We also observe that $E_Q$ does not depend on
the choice of the root relation $R$ (but the bound
$PB(Q,\texttt{ROOT})$ does depend!).  To see this, denote by
$\text{atoms}_Q(Z)$ the set of atoms $R \in \bm R(Q)$ that contain the
variable $Z$, and observe that:
\begin{align*}
  E_Q(h) = & \sum_{R \in \bm R(Q)} h(\bm X_R) - \sum_{Z \in \bm X} (|\text{atoms}_Q(Z)|-1)\cdot h(Z)
\end{align*}
Recall from Sec.~\ref{sec:general} that a {\em cover} of $Q$ is a set
$\bm W = \set{Q_1, \ldots, Q_m}$ where each $Q_i$ is a connected
subquery of $Q$, and each variable of $Q$ occurs in at least one
$Q_i$.  The {\em simple chain expression associated to $\bm W$} is:
\begin{align}
  E_{\bm W}(h) = & E_{Q_1}(h)+ \cdots + E_{Q_m}(h)
\end{align}

One can check that $E_{\bm W}(h) \geq h(\bm X)$ is a valid inequality,
and that $E_{\bm W}$ is an expression associated to $Q$, in the sense
that it is of the form~\eqref{eq:h:expression:q}.  If we write each
$E_{Q_i}$ as a conditional simple expression, i.e. as
in~\eqref{eq:simple:expression:berge}, by choosing some root relation
$\texttt{ROOT}_i$, then it defines the following bound on $PB(Q)$ (see
Equation~\eqref{eq:polymatroid:bound:inequality}):
\begin{align}
PB(Q) \leq &   \prod_{i=1,m} PB(Q_i,\texttt{ROOT}_i) \label{eq:polymatroid:bound:inequality:chain}
\end{align}
where $PB(Q_i,\texttt{ROOT}_i)$ was defined in
Equation~\eqref{eq:pb:one:component}, and the form
$|\texttt{ROOT}_i|\prod_R f_1^{(R,Z_R)}$, i.e. the cardinality of
$\texttt{ROOT}_i$ times the max degrees of the other relations in the
$i$'th component.  We prove that {\em every} valid inequality
associated to the query is some combination of inequalities of this
form:

\begin{thm}[Chain Decomposition] \label{th:chain:decomposition}
  Let $Q$ be a Berge-acyclic query, and let $E(h)$ be an expression
  associated to $Q$, as shown in Equation~\eqref{eq:h:expression:q},
  where all coefficients $a_R, b_{R,Z}$ are natural numbers.  If
  $E(h) \geq k\cdot h(\bm X)$, then there exists $k$ covers
  $\bm W_1, \ldots, \bm W_k$, not necessarily distinct, such that:
  \begin{align*}
    E(h) = & \sum_i E_{\bm W_i}(h) + (\cdots)
  \end{align*}
  where $(\cdots)$ represents a sum of positive terms (i.e. some
  $h(V)$'s or some $h(V|Z)$'s).
\end{thm}

In other words, the inequality $E(h) \geq k \cdot h(\bm X)$ is
(implied by) the sum of $k$ chain inequalities
$E_{\bm W_i}(h) \geq h(\bm X)$.
Theorem~\ref{thm:polyamtroid:bound:sum:of:chains} follows immediately
from here.  Indeed, each chain inequality
$E_{\bm W_i}(h) \geq h(\bm X)$ defines a bound on $PB(Q)$ of the
form~\eqref{eq:polymatroid:bound:inequality:chain}, and this is
precisely the form used in
Theorem~\ref{thm:polyamtroid:bound:sum:of:chains}.  The original
inequality $E(h) \geq k\cdot h(\bm X)$ defines a bound that is no
better than their geometric means, which, in turn, is no better than
the smallest of these bounds.  It follows that the minimum bound over
everything we can derive from a valid inequality
$E(h) \geq k\cdot h(\bm X)$, can be already derived from a chain
inequality, and that leads to a bound of the
form~\eqref{eq:polymatroid:bound:inequality:chain}.  It remains to
prove Theorem~\ref{th:chain:decomposition}.

{\bf Proof of Theorem~\ref{th:chain:decomposition}} Consider any valid
inequality $E(h) \geq k h(\bm X)$ associated to the query $Q$
(Equation~\eqref{eq:h:expression:q}), i.e. all terms in $E(h)$
correspond to atoms in the query, possibly conditioned by one of their
variables.  After expanding the conditional terms,
$h(\bm X_R|Z) = h(\bm X_R) - h(Z)$, we can write the expression as:
\begin{align*}
  E(h) = & \sum_{R \in \bm R(Q)}a_R h(\bm X_R) - \sum_{Z \in \bm X} b_Z h(Z)
\end{align*}
where $a_R, b_Z$ are natural numbers.  Since $E(h)$ is a simple
expression, the inequality $E(h) \geq k\cdot h(\bm X)$ is valid iff
$\texttt{cover}(U,E) \geq k$ for all sets of variables $U$.

Let $Z$ be some private variable, in other words it occurs in only one
relation $R$, and suppose $b_Z > 0$.  In that case we must also have
$a_{\bm X_R}\geq b_Z$, otherwise $\texttt{cover}(Z,E) < 0$.  Write
$a_R \cdot h(\bm X_R) - b_Z\cdot h(Z)= (a_R - b_Z)\cdot h(\bm X_R)+
b_Z h(\bm X_R)-b_Z\cdot h(Z) = (a_R-b_Z) \cdot h(\bm X_R) + b_Z\cdot
h(\bm X_R|Z)$.  We claim that we can remove $b_Z \cdot h(\bm X_R|Z)$
from $E$, and the resulting expression $E_0$ still satisfies the
inequality $E_0(h) \geq k \cdot h(\bm X)$.  To prove the claim it
suffices to check $\texttt{cover}(U,E_0) \geq k$ for all
$U \subseteq \bm X$.  If $Z \in U$ or $U \cap \bm X_R=\emptyset$, then
$\texttt{cover}(U,E_0)=\texttt{cover}(U,E) \geq k$, hence it suffices
to assume $Z \not \in U$, and $U \cap \bm X_R\neq \emptyset$.  In that
case we prove that $\texttt{cover}(W,E_0)=\texttt{cover}(WZ,E_0)$, and
the latter is $\geq k$ as we saw.  The only term in $E_0$ that covers
$Z$ is $(a_{\bm X_R}-b_Z)\cdot h(\bm X_R)$, and this also covers $W$,
which proves $\texttt{cover}(W,E_0)=\texttt{cover}(WZ,E_0)$.
Therefore, we can assume w.l.o.g. that no isolated variables $Z$ occur
in the expression $E$.

We prove now the theorem by induction on the number of relations in
$Q$.  If $Q$ has a single relation $R(\bm X)$, then all variables are
isolated, hence $E$ must be $a_{\bm X} h(\bm X)$, and since every
variable is covered $\geq k$ times, we have $a_{\bm X} \geq k$.

Assume $Q$ has $>1$ atoms, let $R$ be a leaf relation, and let $Q'$ be
the query without the relation $R$; we assume that the theorem holds
for $Q'$.  Since $Q$ is Berge-acyclic, there exists a single variable
$Z \in \bm X_R$ that also occurs in $Q'$.  Therefore, $E(h)$ can be
written as:

\begin{align*}
  E(h) = & a_{\bm X_R} \cdot h(\bm X_R) - b_Z \cdot h(Z) + E'(h)
\end{align*}

where $E'(h)$ contains the remaining terms of the expression $E(h)$,
and thus, is an expression associated to the query $Q'$, i.e. it only
contains terms that correspond to atoms and variables in $Q'$.  Let
$\bm X'$ be the variables of $Q'$.  Our plan is to apply induction
hypothesis to $Q'$, and for that we check the coverage in $E'$ of sets
of variables $U' \subseteq \bm X'$.  If $Z \not\in U'$, then
$\texttt{cover}(U',E')= \texttt{cover}(U',E) \geq k$.  If $Z \in U'$,
then we may have $\texttt{cover}(U',E') < k$.   Define:
\begin{align*}
  k' \defeq & \min_{U' \subseteq \bm X'} \texttt{cover}(U',E')
\end{align*}
and let $U' \subseteq \bm X'$ a set s.t. $\texttt{cover}(U',E')=k'$
(i.e. argmin of the expression above).  Assuming $k' < k$, we must
have $Z \in U'$ (the case when $k' \geq k$ is simple, because in that
case $E'$ already satisfies $E'(h) \geq k\cdot h(\bm X')$, and we omit
it).

We use the fact that
$\texttt{cover}(U',E) = a_{\bm X_R} - b_Z + \texttt{cover}(U',E') =
a_{\bm X_R} - b_Z + k' \geq k$ and derive that
$a_{\bm X_R} = \delta+ (k-k')+b_Z$, for some $\delta \geq 0$.  This
implies:
\begin{align}
  E(h) = & a_{\bm X_R} \cdot h(\bm X_R) - b_Z \cdot h(Z) - (k-k')\cdot
           h(Z) + \left((k-k')\cdot h(Z) +  E'(h)\right) \\
=  & \delta \cdot h(\bm X_R) + \left((k-k')+b_Z\right)\cdot h(\bm X_R|Z)
           + \left(\underbrace{(k-k')\cdot h(Z) + E'(h)}_{\defeq E''(h)}\right) \label{eq:yet:nother:label}
\end{align}
Let $\bm X'$ denote the variables of the query $Q'$, including $Z$.
Once can check that, for all $U \subseteq \bm X'$, $\texttt{cover}(U,
E'') \geq k$.  Indeed, if $Z \not\in U$ then
$\texttt{cover}(U,E'')=\texttt{cover}(U,E')\geq k$ as we argued
earlier, and if $Z \in U$ then:
\begin{align*}
  \texttt{cover}(U,E'') = (k-k') + \texttt{cover}(U,E') \geq (k-k')+k' = k
\end{align*}
Therefore, $E''(h) \geq k h(\bm X')$ is a valid simple inequality,
and, by induction on the query $Q'$, we can write $E''$ as:
\begin{align}
  E''(h) = & \sum_{i=1,k} E_{\bm W_i'}(h) + (\cdots)\label{eq:e:double:prime}
\end{align}
where each $W_i'$ is a cover of $Q'$ and $(\cdots)$ hides some
positive terms.  In other words, $E''(h)$ is written as a sum of
simple chain expressions associated to covers $\bm W_i'$ of $Q'$.  We prove
now that $E(h)$ can also be written as a sum of simple chain
expressions associated to covers of $Q$.  For that, we take the first
terms in~\eqref{eq:yet:nother:label}, argue that there are at least
$k$ of them, and combine each with one of the covers $\bm W_i'$. 

More precisely, we consider two cases.  First, if the relation we have
eliminated, $R$, has no private variables.  In other words
$\bm X_R = \set{Z}$.  In that case $E''(h)$ is already of the required
form by the theorem, since it contains all variables of $Q$.  Second,
suppose $R$ has at least one private variable, call it $Y$.  Then,
from Eq.~\eqref{eq:yet:nother:label} we obtain
$\texttt{cover}(Y,E) = \delta + \left((k-k')+b_Z\right) \geq k$.  This
means that we can produce $k$ terms, either $h(\bm X_R)$ or
$h(\bm X_R|Z)$, and write $E(h)$ as a sum of $k$ terms, where each
term has one of these two forms:
\begin{align*}
  & h(\bm X_R) + E_{\bm W_i'}(h) && h(\bm X_R|Z) + E_{\bm W_i'}(h)
\end{align*}
In the first case we extend the cover $\bm W_i'$ with one new set
$\set{R}$, and denote the resulting cover of $Q$ by $\bm W_i$.  In the
second case we extended the cover by including $R$ in the set
containing some other relation that contains $Z$ (if there are
multiple such sets, we choose one arbitrarily) and denote $\bm W_i$
the resulting cover of $Q$.  In both bases, the expression above
becomes $E_{\bm W_i}(h)$, and therefore
$E(h) = \sum_i E_{\bm W_i}(h) + (\cdots)$ as required.

%% file: main.bbl
\begin{thebibliography}{10}

\bibitem{DBLP:conf/focs/AtseriasGM08}
Albert Atserias, Martin Grohe, and D{\'{a}}niel Marx.
\newblock Size bounds and query plans for relational joins.
\newblock In {\em 49th Annual {IEEE} Symposium on Foundations of Computer
  Science, {FOCS} 2008, October 25-28, 2008, Philadelphia, PA, {USA}}, pages
  739--748. {IEEE} Computer Society, 2008.
\newblock \href {https://doi.org/10.1109/FOCS.2008.43}
  {\path{doi:10.1109/FOCS.2008.43}}.

\bibitem{bauer2015best}
Douglas Bauer, Haitze~J Broersma, Jan van~den Heuvel, Nathan Kahl, A~Nevo,
  E~Schmeichel, Douglas~R Woodall, and Michael Yatauro.
\newblock Best monotone degree conditions for graph properties: a survey.
\newblock {\em Graphs and combinatorics}, 31(1):1--22, 2015.

\bibitem{DBLP:conf/sigmod/CaiBS19}
Walter Cai, Magdalena Balazinska, and Dan Suciu.
\newblock Pessimistic cardinality estimation: Tighter upper bounds for
  intermediate join cardinalities.
\newblock In Peter~A. Boncz, Stefan Manegold, Anastasia Ailamaki, Amol
  Deshpande, and Tim Kraska, editors, {\em Proceedings of the 2019
  International Conference on Management of Data, {SIGMOD} Conference 2019,
  Amsterdam, The Netherlands, June 30 - July 5, 2019}, pages 18--35. {ACM},
  2019.
\newblock \href {https://doi.org/10.1145/3299869.3319894}
  {\path{doi:10.1145/3299869.3319894}}.

\bibitem{DBLP:journals/jacm/Fagin83}
Ronald Fagin.
\newblock Degrees of acyclicity for hypergraphs and relational database
  schemes.
\newblock {\em J. {ACM}}, 30(3):514--550, 1983.
\newblock \href {https://doi.org/10.1145/2402.322390}
  {\path{doi:10.1145/2402.322390}}.

\bibitem{DBLP:conf/sigmod/GiladPM21}
Amir Gilad, Shweta Patwa, and Ashwin Machanavajjhala.
\newblock Synthesizing linked data under cardinality and integrity constraints.
\newblock In Guoliang Li, Zhanhuai Li, Stratos Idreos, and Divesh Srivastava,
  editors, {\em {SIGMOD} '21: International Conference on Management of Data,
  Virtual Event, China, June 20-25, 2021}, pages 619--631. {ACM}, 2021.
\newblock \href {https://doi.org/10.1145/3448016.3457242}
  {\path{doi:10.1145/3448016.3457242}}.

\bibitem{DBLP:journals/jacm/GottlobLVV12}
Georg Gottlob, Stephanie~Tien Lee, Gregory Valiant, and Paul Valiant.
\newblock Size and treewidth bounds for conjunctive queries.
\newblock {\em J. {ACM}}, 59(3):16:1--16:35, 2012.
\newblock \href {https://doi.org/10.1145/2220357.2220363}
  {\path{doi:10.1145/2220357.2220363}}.

\bibitem{DBLP:conf/soda/GroheM06}
Martin Grohe and D{\'{a}}niel Marx.
\newblock Constraint solving via fractional edge covers.
\newblock In {\em Proceedings of the Seventeenth Annual {ACM-SIAM} Symposium on
  Discrete Algorithms, {SODA} 2006, Miami, Florida, USA, January 22-26, 2006},
  pages 289--298. {ACM} Press, 2006.
\newblock URL: \url{http://dl.acm.org/citation.cfm?id=1109557.1109590}.

\bibitem{hakimi1978graphs}
S~Louis Hakimi and Edward~F Schmeichel.
\newblock Graphs and their degree sequences: A survey.
\newblock In {\em Theory and applications of graphs}, pages 225--235. Springer,
  1978.

\bibitem{han2021cardinality}
Yuxing Han, Ziniu Wu, Peizhi Wu, Rong Zhu, Jingyi Yang, Liang~Wei Tan, Kai
  Zeng, Gao Cong, Yanzhao Qin, Andreas Pfadler, et~al.
\newblock Cardinality estimation in dbms: A comprehensive benchmark evaluation.
\newblock {\em arXiv preprint arXiv:2109.05877}, 2021.

\bibitem{DBLP:conf/cidr/HertzschuchHHL21}
Axel Hertzschuch, Claudio Hartmann, Dirk Habich, and Wolfgang Lehner.
\newblock Simplicity done right for join ordering.
\newblock In {\em 11th Conference on Innovative Data Systems Research, {CIDR}
  2021, Virtual Event, January 11-15, 2021, Online Proceedings}.
  www.cidrdb.org, 2021.
\newblock URL: \url{http://cidrdb.org/cidr2021/papers/cidr2021\_paper01.pdf}.

\bibitem{DBLP:conf/pods/KhamisK0S20}
Mahmoud~Abo Khamis, Phokion~G. Kolaitis, Hung~Q. Ngo, and Dan Suciu.
\newblock Bag query containment and information theory.
\newblock In Dan Suciu, Yufei Tao, and Zhewei Wei, editors, {\em Proceedings of
  the 39th {ACM} {SIGMOD-SIGACT-SIGAI} Symposium on Principles of Database
  Systems, {PODS} 2020, Portland, OR, USA, June 14-19, 2020}, pages 95--112.
  {ACM}, 2020.
\newblock \href {https://doi.org/10.1145/3375395.3387645}
  {\path{doi:10.1145/3375395.3387645}}.

\bibitem{DBLP:conf/pods/KhamisNR16}
Mahmoud~Abo Khamis, Hung~Q. Ngo, and Atri Rudra.
\newblock {FAQ:} questions asked frequently.
\newblock In Tova Milo and Wang{-}Chiew Tan, editors, {\em Proceedings of the
  35th {ACM} {SIGMOD-SIGACT-SIGAI} Symposium on Principles of Database Systems,
  {PODS} 2016, San Francisco, CA, USA, June 26 - July 01, 2016}, pages 13--28.
  {ACM}, 2016.
\newblock \href {https://doi.org/10.1145/2902251.2902280}
  {\path{doi:10.1145/2902251.2902280}}.

\bibitem{DBLP:conf/pods/KhamisNS16}
Mahmoud~Abo Khamis, Hung~Q. Ngo, and Dan Suciu.
\newblock Computing join queries with functional dependencies.
\newblock In Tova Milo and Wang{-}Chiew Tan, editors, {\em Proceedings of the
  35th {ACM} {SIGMOD-SIGACT-SIGAI} Symposium on Principles of Database Systems,
  {PODS} 2016, San Francisco, CA, USA, June 26 - July 01, 2016}, pages
  327--342. {ACM}, 2016.
\newblock \href {https://doi.org/10.1145/2902251.2902289}
  {\path{doi:10.1145/2902251.2902289}}.

\bibitem{DBLP:journals/corr/Khamis0S16}
Mahmoud~Abo Khamis, Hung~Q. Ngo, and Dan Suciu.
\newblock What do shannon-type inequalities, submodular width, and disjunctive
  datalog have to do with one another?
\newblock {\em CoRR}, abs/1612.02503, 2016.
\newblock URL: \url{http://arxiv.org/abs/1612.02503}, \href
  {http://arxiv.org/abs/1612.02503} {\path{arXiv:1612.02503}}.

\bibitem{DBLP:conf/pods/Khamis0S17}
Mahmoud~Abo Khamis, Hung~Q. Ngo, and Dan Suciu.
\newblock What do shannon-type inequalities, submodular width, and disjunctive
  datalog have to do with one another?
\newblock In Emanuel Sallinger, Jan~Van den Bussche, and Floris Geerts,
  editors, {\em Proceedings of the 36th {ACM} {SIGMOD-SIGACT-SIGAI} Symposium
  on Principles of Database Systems, {PODS} 2017, Chicago, IL, USA, May 14-19,
  2017}, pages 429--444. {ACM}, 2017.
\newblock \href {https://doi.org/10.1145/3034786.3056105}
  {\path{doi:10.1145/3034786.3056105}}.

\bibitem{DBLP:journals/pvldb/LeisGMBK015}
Viktor Leis, Andrey Gubichev, Atanas Mirchev, Peter~A. Boncz, Alfons Kemper,
  and Thomas Neumann.
\newblock How good are query optimizers, really?
\newblock {\em Proc. {VLDB} Endow.}, 9(3):204--215, 2015.
\newblock URL: \url{http://www.vldb.org/pvldb/vol9/p204-leis.pdf}, \href
  {https://doi.org/10.14778/2850583.2850594}
  {\path{doi:10.14778/2850583.2850594}}.

\bibitem{DBLP:journals/pvldb/LiuD0Z21}
Jie Liu, Wenqian Dong, Dong Li, and Qingqing Zhou.
\newblock Fauce: Fast and accurate deep ensembles with uncertainty for
  cardinality estimation.
\newblock {\em Proc. {VLDB} Endow.}, 14(11):1950--1963, 2021.
\newblock URL: \url{http://www.vldb.org/pvldb/vol14/p1950-liu.pdf}.

\bibitem{DBLP:journals/pvldb/NegiMKMTKA21}
Parimarjan Negi, Ryan~C. Marcus, Andreas Kipf, Hongzi Mao, Nesime Tatbul, Tim
  Kraska, and Mohammad Alizadeh.
\newblock Flow-loss: Learning cardinality estimates that matter.
\newblock {\em Proc. {VLDB} Endow.}, 14(11):2019--2032, 2021.
\newblock URL: \url{http://www.vldb.org/pvldb/vol14/p2019-negi.pdf}.

\bibitem{DBLP:conf/pods/000118}
Hung~Q. Ngo.
\newblock Worst-case optimal join algorithms: Techniques, results, and open
  problems.
\newblock In Jan~Van den Bussche and Marcelo Arenas, editors, {\em Proceedings
  of the 37th {ACM} {SIGMOD-SIGACT-SIGAI} Symposium on Principles of Database
  Systems, Houston, TX, USA, June 10-15, 2018}, pages 111--124. {ACM}, 2018.
\newblock \href {https://doi.org/10.1145/3196959.3196990}
  {\path{doi:10.1145/3196959.3196990}}.

\bibitem{DBLP:conf/sigmod/ParkKBKHH20}
Yeonsu Park, Seongyun Ko, Sourav~S. Bhowmick, Kyoungmin Kim, Kijae Hong, and
  Wook{-}Shin Han.
\newblock {G-CARE:} {A} framework for performance benchmarking of cardinality
  estimation techniques for subgraph matching.
\newblock In David Maier, Rachel Pottinger, AnHai Doan, Wang{-}Chiew Tan,
  Abdussalam Alawini, and Hung~Q. Ngo, editors, {\em Proceedings of the 2020
  International Conference on Management of Data, {SIGMOD} Conference 2020,
  online conference [Portland, OR, USA], June 14-19, 2020}, pages 1099--1114.
  {ACM}, 2020.
\newblock \href {https://doi.org/10.1145/3318464.3389702}
  {\path{doi:10.1145/3318464.3389702}}.

\bibitem{DBLP:conf/sigmod/Sun0021}
Ji~Sun, Guoliang Li, and Nan Tang.
\newblock Learned cardinality estimation for similarity queries.
\newblock In Guoliang Li, Zhanhuai Li, Stratos Idreos, and Divesh Srivastava,
  editors, {\em {SIGMOD} '21: International Conference on Management of Data,
  Virtual Event, China, June 20-25, 2021}, pages 1745--1757. {ACM}, 2021.
\newblock \href {https://doi.org/10.1145/3448016.3452790}
  {\path{doi:10.1145/3448016.3452790}}.

\bibitem{DBLP:journals/pvldb/WangQWWZ21}
Xiaoying Wang, Changbo Qu, Weiyuan Wu, Jiannan Wang, and Qingqing Zhou.
\newblock Are we ready for learned cardinality estimation?
\newblock {\em Proc. {VLDB} Endow.}, 14(9):1640--1654, 2021.
\newblock URL: \url{http://www.vldb.org/pvldb/vol14/p1640-wang.pdf}.

\bibitem{DBLP:conf/sigmod/WuC21}
Peizhi Wu and Gao Cong.
\newblock A unified deep model of learning from both data and queries for
  cardinality estimation.
\newblock In Guoliang Li, Zhanhuai Li, Stratos Idreos, and Divesh Srivastava,
  editors, {\em {SIGMOD} '21: International Conference on Management of Data,
  Virtual Event, China, June 20-25, 2021}, pages 2009--2022. {ACM}, 2021.
\newblock \href {https://doi.org/10.1145/3448016.3452830}
  {\path{doi:10.1145/3448016.3452830}}.

\bibitem{wu2020bayescard}
Ziniu Wu, Amir Shaikhha, Rong Zhu, Kai Zeng, Yuxing Han, and Jingren Zhou.
\newblock Bayescard: Revitilizing bayesian frameworks for cardinality
  estimation.
\newblock {\em arXiv e-prints}, pages arXiv--2012, 2020.

\bibitem{DBLP:journals/pvldb/YangKLLDCS20}
Zongheng Yang, Amog Kamsetty, Sifei Luan, Eric Liang, Yan Duan, Xi~Chen, and
  Ion Stoica.
\newblock Neurocard: One cardinality estimator for all tables.
\newblock {\em Proc. {VLDB} Endow.}, 14(1):61--73, 2020.
\newblock URL: \url{http://www.vldb.org/pvldb/vol14/p61-yang.pdf}, \href
  {https://doi.org/10.14778/3421424.3421432}
  {\path{doi:10.14778/3421424.3421432}}.

\bibitem{yang2020neurocard}
Zongheng Yang, Amog Kamsetty, Sifei Luan, Eric Liang, Yan Duan, Xi~Chen, and
  Ion Stoica.
\newblock Neurocard: one cardinality estimator for all tables.
\newblock {\em arXiv preprint arXiv:2006.08109}, 2020.

\bibitem{zhu2020flat}
Rong Zhu, Ziniu Wu, Yuxing Han, Kai Zeng, Andreas Pfadler, Zhengping Qian,
  Jingren Zhou, and Bin Cui.
\newblock Flat: fast, lightweight and accurate method for cardinality
  estimation.
\newblock {\em arXiv preprint arXiv:2011.09022}, 2020.

\bibitem{DBLP:journals/pvldb/ZhuWHZPQZC21}
Rong Zhu, Ziniu Wu, Yuxing Han, Kai Zeng, Andreas Pfadler, Zhengping Qian,
  Jingren Zhou, and Bin Cui.
\newblock {FLAT:} fast, lightweight and accurate method for cardinality
  estimation.
\newblock {\em Proc. {VLDB} Endow.}, 14(9):1489--1502, 2021.
\newblock URL: \url{http://www.vldb.org/pvldb/vol14/p1489-zhu.pdf}.

\end{thebibliography}
